\documentclass[a4paper,english,cleveref, autoref]{lipics-v2019}
\pdfoutput=1

\usepackage{packages}

\usepackage[noend]{algpseudocode}

\usepackage{macros}

\title{Internal Calculi for Separation Logics}
\titlerunning{Internal Calculi for Separation Logics \ifLongVersionWithAppendix \\ (version with the technical appendix) \fi}

\author{St\'ephane Demri}{LSV, CNRS, ENS Paris-Saclay, Universit\'e Paris-Saclay, France}{}{}{}
\author{Etienne Lozes}{I3S, Universit\'e C\^ote d'Azur, France}{}{}{}
\author{Alessio Mansutti}{LSV, CNRS, ENS Paris-Saclay Universit\'e Paris-Saclay, France}{}{}{}

\authorrunning{S.\ Demri, E.\ Lozes, A.\ Mansutti}

\Copyright{S.\ Demri, E.\ Lozes, A.\ Mansutti}

 \ccsdesc[100]{Theory of computation}
 \ccsdesc[100]{Logic}

\keywords{Separation logic, internal calculus, adjunct/quantifier elimination}

\category{}

\relatedversion{}

\supplement{}

\nolinenumbers 

\EventEditors{John Q. Open and Joan R. Access}
\EventNoEds{2}
\EventLongTitle{42nd Conference on Very Important Topics (CVIT 2016)}
\EventShortTitle{CVIT 2016}
\EventAcronym{CVIT}
\EventYear{2016}
\EventDate{December 24--27, 2016}
\EventLocation{Little Whinging, United Kingdom}
\EventLogo{}
\SeriesVolume{42}
\ArticleNo{23}

\hideLIPIcs

\newif\ifLongVersionWithAppendix\LongVersionWithAppendixfalse
\newif\ifCSLProc\CSLProcfalse

\begin{document}
\maketitle
\begin{abstract}
We present a general approach to axiomatise separation logics
with heaplet semantics with no external features such as nominals/labels.
\LongVersionOnly{
It is not possible to rely completely on proof systems for Boolean BI as
the concrete semantics needs to be taken account and without extraneous
help, it is difficult to get a grip on  the constraints on heaps.
}
To start with, we design the first (internal) Hilbert-style axiomatisation for the quantifier-free
separation logic \slSW.
\LongVersionOnly{
The calculus is divided in three parts: the axiomatisation
of core formulae where Boolean combinations of core formulae capture the whole logic,
axioms and inference rules to simulate a bottom-up elimination of separating connectives, and finally
axioms and inference rules from propositional calculus and Boolean BI.
}
We instantiate the method by introducing a new separation logic with essential features: it is equipped with
the separating conjunction, the 
predicate $\ls$, and a natural guarded form
of first-order quantification. We apply our approach for its axiomatisation. 
\LongVersionOnly{that extends, for instance, Boolean combinations of symbolic heaps.}
As a by-product of our method, we also establish
the exact expressive power of this new logic and we show \pspace-completeness of its satisfiability problem.

\end{abstract}
\section{Introduction}\label{section:introduction}
{\bf \textsf{The virtue of axiomatising program logics.}}
Designing a Hilbert-style axiomatisation for your favourite logic is usually quite challenging.
This does not  lead necessarily to optimal decision procedures, but the completeness proof usually provides essential
insights to better understand the logic at hand. That is why many logics related to program verification
have been axiomatised, often requiring non-trivial completeness proofs.
By way of example, there exist axiomatisations for
the linear-time $\mu$-calculus~\cite{Kaivola95,Doumane17},
the modal $\mu$-calculus~\cite{Walukiewicz00}
or for
the alternating-time temporal logic ATL~\cite{Goranko&vanDrimmelen06}.
\LongVersionOnly{
the full computation tree logic CTL$^{*}$~\cite{Reynolds01},
for probabilistic extensions of $\mu$-calculus~\cite{Larsen&Mardare&Xue16}
or for a coalgebraic generalisation~\cite{Schroder&Venema18}.
}
Concerning the separation logics that extend Hoare-Floyd logic
to verify programs with mutable data structures
(see e.g.~\cite{OHearn&Pym99,Reynolds02,Ishtiaq&OHearn01,OHearn12,Pym&Spring&OHearn18}),
a Hilbert-style axiomatisation of Boolean BI has been introduced in~\cite{Galmiche&Larchey06},
but remained at the abstract level of Boolean BI. More recently, HyBBI~\cite{Brotherston&Villard14},
a hybrid version of Boolean BI has been introduced in order
to axiomatise various classes of separation logics; HyBBI
naturally considers  classes of abstract models (typically preordered partial monoids) but it does not fit
exactly the heaplet semantics of separation logics. Furthermore, the addition of nominals (in the sense
of hybrid modal logics, see e.g.~\cite{Areces&Blackburn&Marx01}) extends substantially the object language. Other  frameworks
to axiomatise classes of abstract separation logics can be found
in~\cite{Docherty&Pym18} and in~\cite{Houetal18}, respectively with labelled tableaux calculi
and with sequent-style proof systems.\\[3pt]
{\bf \textsf{Our motivations.}}
Since the birth of separation logics, there has been  a lot of interest in the study of decidability and computational complexity issues,
see e.g.~\cite{Bozga&Iosif&Perarnau10,Calcagno&Yang&OHearn01,Cooketal11,Brotherston&Kanovich18,DemriLM18,Mansutti18},
and comparatively a bit less attention to the design of proof systems, and even less with the puristic approach
that consists in discarding any external feature such as nominals or labels in the calculi.
The well-known advantages of such an approach include an exhaustive understanding of the expressive power of
the logic and discarding the use of any external artifact referring to semantical objects.
For instance, a complete tableaux calculus with labels for quantifier-free separation logic is designed in~\cite{Galmiche&Mery10} --with an extension of
the calculus to handle quantifiers,
whereas Hilbert-style calculi for abstract separation logics with nominals are defined in~\cite{Brotherston&Villard14}
(see also in~\cite{Hou&Tiu16} a proof system for a first-order abstract separation logic with an abstracted version of the
points-to predicate).
Similarly, display calculi for bunched logics are provided in~\cite{Brotherston12} and such calculi extend
Gentzen-style proof systems by allowing new structural connectives.
In this paper, we advocate a puristic approach and aim at designing Hilbert-style proof systems for quantifier-free
separation logic \slSW (which includes the separating conjunction $\separate$ and implication $\magicwand$, as well as all Boolean connectives) and more generally for other separation logics,  while remaining within the very logical
language.
Consequently, in this work we only focus on axiomatising the separation logics,
and  we have no claim for  practical applications in the field of program verification with separation logics.
Aiming at internal calculi is  a non-trivial
task as the general frameworks for abstract separation logics make use of labels, see e.g.~\cite{Docherty&Pym18,Houetal18}.
We cannot fully rely on label-free calculi for BI, see e.g.~\cite{Pym02,Galmiche&Larchey06}, as separation
logics are usually understood
as Boolean BI interpreted on models of heap memory and therefore require calculi that
handle specifically the stack-and-heap models.
Finally, we know many translations from separation logics into
logics or theories, see e.g.~\cite{Calcagno&Gardner&Hague05,Piskac&Wies&Zufferey13,Brochenin&Demri&Lozes12},
but completeness cannot in general be inherited by sublogics as the proof system should
only use the sublogic and therefore the axiomatisation of sublogics may lead
to different methods.\\[3pt]
{\bf \textsf{Our contribution.}}
Though our initial motivation is to design an internal Hilbert-style axiomatisation for \slSW,
 we go beyond this, and we propose a method to axiomatise other separation logics assuming that
key properties are satisfied. Hence, we
consider a broader perspective and
we use our approach on two separation logics: quantifier-free separation logic and a new separation
logic that admits a form of guarded first-order quantification.
Our results are not limited to (internal) axiomatisation, as we provide a complexity analysis based
on the properties of the derivations in the proof system.  Let us be a bit more precise.

In Section~\ref{section:PSL}, we provide the first Hilbert-style proof system
for \slSW
that uses axiom schemas and rules involving only formulae of this logic.
Each formula of \slSW is  equivalent to a Boolean combination of {\em core formulae}: simple formulae of the logic expressing elementary properties
about the models~\cite{Lozes04}.
Though core formulae (also called {\em test formulae}) have been handy in several occasions for
establishing complexity results for separation logics, see e.g.~\cite{Demrietal17,DemriLM18,Echenim&Iosif&Peltier19},
in the paper, these formulae are instrumental for the axiomatisation.
Indeed, we distinguish the axiomatisation of Boolean combinations of core formulae
from the transformation of  formulae into such Boolean combinations.
Thus, we show how to introduce axioms to transform every formula into
a Boolean combination of core formulae, together with axioms to deal with these simple formulae.
Schematically, for a valid formula $\aformula$, we conclude $\vdash \aformula$  from
$\vdash \aformula'$ and $\vdash \aformula' \Leftrightarrow \aformula$, where
$\aformula'$ is a Boolean combination of core formulae.
Another difficulty arises as we have to design an axiomatisation for such Boolean combinations.
So, the calculus is divided in three parts: the axiomatisation
of  Boolean combinations of core formulae,
axioms and inference rules to simulate a bottom-up elimination of separating connectives, and finally
 axioms and inference rules from propositional calculus and Boolean BI.
\LongVersionOnly{Therefore, part of the axiomatisation
is dedicated to such Boolean combinations and another part is dedicated to establish the  equivalence.}
Such an approach that consists in first axiomatising a syntactic fragment of the whole logic (in our case, the core formulae), is best described in~\cite{Doumane17}
(see also~\cite{Walukiewicz00,WangC13,Luck18,Demri&Fervari&Mansutti19}).

\LongVersionOnly{
Section~\ref{section:PSL} concludes by recapitulating the key ingredients of our method to axiomatise \slSW
with the intention to reuse it for other separation logics. For instance, it is essential to be able to design
a family of core formulae whose Boolean combinations capture the expressive power of the logic  and such that
the core formulae are all expressible in the logic, which allows us to provide a calculus without any external
help.
}

In Section~\ref{section:IntervalSL},
our intention is to add standard features to the logic such as first-order quantification and inductive predicates, and to apply
our method for axiomatisation.
As $\seplogic{\separate,\magicwand,\ls}$ (i.e. $\seplogic{\separate,\magicwand}$ enriched with the predicate $\ls$)
is already non-finitely axiomatisable~\cite{Demri&Lozes&Mansutti18bis}, we need to fine-tune the logical formalism.
That is why, we introduce a new separation logic \intervalSL that admits the separating
conjunction $\separate$ (no $\magicwand$) and a guarded form of first-order quantification.
In the formula $\inpath{\avariable}{\avariablebis}{\avariableter} \aformula$, the variable $\avariableter$ is existentially quantified over the set of locations in the minimal non-empty path from $\avariable$ to $\avariablebis$, if any.
\LongVersionOnly{holds whenever $\avariablebis$ is reachable from $\avariable$ in at least one step,
and there is a value $\avariableter$ along the
minimal path between $\avariable$ and  $\avariablebis$ such that  $\aformula$ holds.}
The logic \intervalSL contains the symbolic heap fragment~\cite{Berdine&Calcagno&OHearn04,Cooketal11} but also richer logics such as
\slSRp from~\cite{DemriLM18}.
Hence, the logic \intervalSL
captures the list segment predicate $\ls$ but also allows us to quantify in a guarded  form
over locations in a minimal path, which makes it a promising language.
We provide an internal Hilbert-style axiomatisation for \intervalSL, illustrating
the flexibility of our method.
\LongVersionOnly{even for a logic with a  built-in inductive predicate $\ls$ and a restricted
form of first-order quantification.}
It requires the  design of
an adequate family of core formulae that captures \intervalSL. The axiomatisation of Boolean combinations
of core formulae reveals to be challenging, and the elimination of guarded quantification or separating conjunction
happens also to require complex developments.
\LongVersionOnly{Hence, the presentation of our method for axiomatisation a sophisticated
and expressive separation logic witnesses its strength and generality.}
We analyse the derivations from
the calculus to establish a small model property for the logic and, together with a symbolic model-checking algorithm, prove that the satisfiability problem for \intervalSL is in \pspace.

\cut{
\noindent%
\ifLongVersionWithAppendix%
{\em Omitted proofs are in the technical appendices~\ref{appendix:Section3},~\ref{appendix:Section4} and~\ref{appendix:DerivedTautologies}.}
\else%
{\em Omitted proofs are in the appendix of the long version of the paper available online~\cite{Demri&Lozes&Mansutti19}.}
\fi
}

\section{Preliminaries}\label{section:preliminaries}
{\bf \textsf{Quantifier-free separation logic \slSW.}}
We present the quantifier-free separation logic \slSW, that includes standard features such as
the separating conjunction $\separate$ and the separating implication $\magicwand$.
\LongVersionOnly{and closure under Boolean connectives.}
\LongVersionOnly{
Section~\ref{section:PSL} is dedicated to its internal Hilbert-style axiomatisation
whereas Section~\ref{section:IntervalSL} is  dedicated to the Hilbert-style axiomatisation of the newly introduced
separation logic  \intervalSL.
}
Let $\PVAR = \set{\avariable, \avariablebis, \ldots}$ be a countably infinite set of \defstyle{program variables}.
The formulae $\aformula$ of \slSW and its atomic formulae
$\aatomicformula$ are built from  the grammars below (where $\avariable, \avariablebis \in \PVAR$ and the connectives $\Rightarrow$, $\Leftrightarrow$ and $\vee$ are defined as usually).
\begin{nscenter}
$
\aatomicformula ::= \avariable = \avariablebis \ \mid \
                    \avariable \Ipto \avariablebis \ \mid \
                    \emptyconstant
  \qquad\qquad
\aformula ::= \aatomicformula \ \mid \  \neg \aformula \ \mid \ \aformula \wedge \aformula
                 \ \mid \ \aformula \separate \aformula  \ \mid \ \aformula \magicwand \aformula.
$
\end{nscenter}
In the heaplet semantics, the formulae of \slSW are interpreted on  \defstyle{memory states} that are pairs
$\pair{\astore}{\aheap}$  where
${\astore: \PVAR \rightarrow \LOC}$ is
a variable valuation (the \defstyle{store}) from the set of program variables to a
countably infinite set of \defstyle{locations} $\LOC = \set{\alocation_0,\alocation_1, \alocation_2, \ldots}$ whereas
$\aheap: \LOC \to_{\fin} \LOC$ is a partial function with finite domain (the \defstyle{heap}).
We write $\domain{\aheap}$ to denote its domain and $\range{\aheap}$ to denote its range.
A \defstyle{memory cell} of $\aheap$ is understood as a pair of locations $\pair{\alocation}{\alocation'}$
such that $\alocation \in \domain{\aheap}$ and $\alocation'  = \aheap(\alocation)$.
As usual, the heaps $\aheap_1$ and $\aheap_2$ are said to be \defstyle{disjoint}, written $\aheap_1 \perp \aheap_2$,
if $\domain{\aheap_1} \cap \domain{\aheap_2} = \emptyset$;  when this holds, we write $\aheap_1 + \aheap_2$ to denote the heap
corresponding to the disjoint union of the graphs of $\aheap_1$ and $\aheap_2$, hence $\domain{\aheap_1 + \aheap_2} = \domain{\aheap_1} \uplus \domain{\aheap_2}$.
\LongVersionOnly{When the domains of $\aheap_1$ and $\aheap_2$  are not disjoint, the composition $\aheap_1 + \aheap_2$ is not defined.}
Moreover, we write $\aheap' \sqsubseteq \aheap$ to denote that $\domain{\aheap'} \subseteq \domain{\aheap}$ and for all locations
$\alocation \in \domain{\aheap'}$, we have $\aheap'(\alocation) = \aheap(\alocation)$.
Given a heap $\aheap$, we define a family of  $(\aheap^\alength)_{\alength \in \Nat}$ of partial functions such that
$\aheap^0$ is the identity function on $\LOC$, $\aheap^1 = \aheap$ and
for all $\alength \geq 2$ and $\alocation \in \LOC$, we have
$\aheap^{\alength}(\alocation) \egdef \aheap(\aheap^{\alength-1}(\alocation))$, assuming that $\aheap^{\alength-1}(\alocation)$ is defined and belongs to $\domain{\aheap}$, otherwise
$\aheap^{\alength}(\alocation)$ is undefined.
\LongVersionOnly{
Given a heap $\aheap$, we define a family of heaps $(\aheap^\alength)_{\alength \in \Nat}$ as follows understood that
$\aheap^\alength$ is obtained from  $\alength$ functional composition(s) of $\aheap$. By definition, $\aheap^0$ is the identity function on $\LOC$,
$\aheap^1 \egdef \aheap$ and for all $\alength \geq 2$ and $\alocation \in \LOC$, we have
$\aheap^{\alength}(\alocation) \egdef \aheap(\aheap^{\alength-1}(\alocation))$, assuming that $\aheap^{\alength-1}(\alocation)$ is defined and belongs to $\domain{\aheap}$, otherwise
$\aheap^{\alength}(\alocation)$ is undefined.
}
The satisfaction relation $\models$ is defined as follows
(omitting standard clauses for $\neg$, $\wedge$):
\begin{nscenter}
\begin{tabular}[t]{cc}
  $\pair{\astore}{\aheap} \models \avariable = \avariablebis$ $\equivdef$ $\astore(\avariable) = \astore(\avariablebis)$ &
  $\pair{\astore}{\aheap} \models \emp$ $\equivdef$ $\domain{\aheap} = \emptyset$ \\
  \multicolumn{2}{l}{$\pair{\astore}{\aheap} \models \avariable {\Ipto} \avariablebis$ $\equivdef$
  $\astore(\avariable)\in\domain{\aheap}$ and $\aheap(\astore(\avariable)) = \astore(\avariablebis)$} \\
  \multicolumn{2}{l}{$\pair{\astore}{\aheap} \models \aformula_1 \separate \aformula_2$ $\equivdef$ $\exists\aheap_1,\aheap_2.$
  $\aheap_1 \bot \aheap_2$, $(\aheap_1 + \aheap_2) = \aheap$,
  $\pair{\astore}{\aheap_1}  \models \aformula_1$ and  $\pair{\astore}{\aheap_2}  \models \aformula_2$}\\
  \multicolumn{2}{l}{$\pair{\astore}{\aheap} \models \aformula_1 \magicwand \aformula_2$ $\equivdef$ $\forall\aheap_1.$
  ($\aheap_1 \bot \aheap$ and $\pair{\astore}{\aheap_1} \models \aformula_1$)
  implies $\pair{\astore}{\aheap + \aheap_1} \models \aformula_2$.}
\end{tabular}
\end{nscenter}
We denote with $\bot$ the contradiction $\avariable \neq \avariable$, and with $\top$ its negation $\neg\bot$.
The septraction operator $\septraction$ (kind of dual of $\magicwand$), defined
by $\aformula\septraction\aformulabis \egdef \neg(\aformula\magicwand\neg\aformulabis)$,
has the following semantics:
\begin{nscenter}
$\pair{\astore}{\aheap}\models \aformula\septraction\aformulabis$
$\iff$ there is a heap $\aheap'$ such that $\aheap\bot\aheap'$, $\pair{\astore}{\aheap'}\models\aformula$, and $\pair{\astore}{\aheap + \aheap'}\models\aformulabis$.
\end{nscenter}
Moreover, we introduce the following (important) shortcuts:
\begin{itemize}[nosep,before=\vspace{3pt},after=\vspace{3pt}]
\item $\alloc{\avariable}$ which is satisfied by $\pair{\astore}{\aheap}$ iff $\astore(\avariable) \in \domain{\aheap}$.
It is defined as $(\avariable \Ipto \avariable) \magicwand \false$.
\item
$\size \geq \inbound$ which is satisfied by $\pair{\astore}{\aheap}$ iff $\card{\domain{\aheap}} \geq \inbound$, where $\inbound \in \Nat$ and $\card{\aset}$ denotes the cardinality of the set $\aset$.
This shortcut is inductively defined as
$\sizegeq{0} \egdef \true$, $\sizegeq{1} \egdef \lnot \emp$ and, for each $\inbound \in \Nat$, $\sizegeq{\inbound{+}2} \egdef \lnot\emp \separate \sizegeq{\inbound{+}1}$.
\end{itemize}
We use $\size {=} \inbound$ as a shorthand for $\size {\geq }\inbound \land \lnot \size {\geq} \inbound {+} 1$.
A formula $\aformula$ is \defstyle{valid} if $\pair{\astore}{\aheap}\models\aformula$ for all $\pair{\astore}{\aheap}$ (and we write $\models \aformula$).
For a complete description of separation logic, see e.g.~\cite{Reynolds02}.\\[3pt]
\noindent
{\bf \textsf{Hilbert-style proof systems.}}
A \defstyle{Hilbert-style proof system} $\aproofsystem$ is defined as a set of \defstyle{derivation step schemata}
 $((\aformulaschema_1,\dots,\aformulaschema_n),\aformulaschemabis)$ with $n\geq 0$,
where $\aformulaschema_1,\dots,\aformulaschema_n,\aformulaschemabis$ are \defstyle{formula schemata}.
When $n \geq 1$, $((\aformulaschema_1,\dots,\aformulaschema_n),\aformulaschemabis)$ is called an \defstyle{inference rule},
otherwise it is an \defstyle{axiom}. As usual, formula schemata generalise the notion of formulae by allowing metavariables
for formulae (typically $\aformula, \aformulabis, \aformulater$), for program variables (typically $\avariable, \avariablebis, \avariableter$)
or for any type of syntactic objects in formulae, depending on the context.
 The set of formulae \defstyle{derivable} from
$\aproofsystem$ is the least set $S$ such that for all
 $((\aformulaschema_1,\dots,\aformulaschema_n),\aformulaschemabis)\in\aproofsystem$ and for all
substitutions $\asubstitution$ such that
$\aformulaschema_1\asubstitution,\dots,\aformulaschema_n\asubstitution\in S$, $\aformulaschemabis\asubstitution\in S$.
We write $\pfentails{\aproofsystem}{\aformula}$ if $\aformula$ is
derivable from $\aproofsystem$.
A proof system $\aproofsystem$
is \defstyle{sound} if all derivable formulae are valid. $\aproofsystem$
is \defstyle{complete} if
all valid formulae are derivable.
$\aproofsystem$ is \defstyle{strongly complete}
iff for all sets of formulae $\Gamma$ and formulae $\aformula$,
we have $\Gamma \models \aformula$ (semantical entailment) iff
$\pfentails{\aproofsystem\cup \Gamma}{\aformula}$.

Interestingly enough, there is no strongly complete proof system
for separation logic, as strong completeness implies compactness
and separation logic is not compact.
Indeed, $\set{\size\geq \inbound \mid \inbound\in\Nat}$ is unsatisfiable, as heaps have finite domains, but all finite subsets of it
are satisfiable.
Even for the weaker notion of completeness, deriving an Hilbert-style axiomatisation for \slSW remains challenging.
Indeed, the satisfiability problem  for \slSW reduces to its validity problem,
making \slSW an unusual logic from a proof-theoretical point of view.
Let us develop a bit further this point.
Let $\aformula$ be a formula  with program variables in $\asetvar \subseteq_{\fin} \PVAR$, and let $\approx$ be an equivalence relation on $\asetvar$.
  The formula
  $\aformulabis_{\approx} \egdef (\emp \land \bigwedge_{\avariable \approx \avariablebis} \avariable = \avariablebis \land  \bigwedge_{\substack{\avariable \not\approx \avariablebis}} \avariable \neq \avariablebis )\implies (\aformula \septraction \true)$
  can be shown to be valid iff
  for every store $\astore$ agreeing on $\approx$, there is a heap $\aheap$ such that $\pair{\astore}{\aheap} \models \aformula$.
  It is known
  that for  all stores $\astore,\astore'$ agreeing on $\approx$, and every heap $\aheap$,
$\pair{\astore}{\aheap}$ and $\pair{\astore'}{\aheap}$
  satisfy the same set of formulae having variables from $\asetvar$.
  Since the antecedent of $\aformulabis_{\approx}$ is satisfiable, we conclude that
   $\aformulabis_{\approx}$ is valid iff there are a store $\astore$ agreeing on $\approx$ and a heap $\aheap$ such that $\pair{\astore}{\aheap} \models \aformula$.
   To check whether $\aformula$ is satisfiable, it is sufficient to find an equivalence relation $\approx$ on $\asetvar$ such that
   $\aformulabis_{\approx}$ is valid. As the number of equivalence relations on $\asetvar$ is finite, we obtain a Turing reduction from satisfiability to validity.
Consequently, it is not possible to define sound and complete axiom systems for any extension of \slSW admitting an undecidable validity problem
(as long as there is a reduction from satisfiability to validity, as above).
A good example is
$\seplogic{\separate,\magicwand,\ls}$~\cite{Demri&Lozes&Mansutti18bis} (extension of \slSW with  $\ls$).
Indeed, in order to obtain a sound and complete axiom system, the validity problem has to be recursively enumerable (r.e.).
However,  this would imply that the satisfiability problem is also r.e.. As
$\aformula$ is not valid iff
$\lnot \aformula$ is satisfiable, we then conclude that the set of valid formulae is recursive, hence decidable, a contradiction.

It is worth also noting that quantifier-free $\slSW$ axiomatised below admits a \pspace-complete validity problem, see e.g.~\cite{Calcagno&Yang&OHearn01},
and should not be confused with propositional separation logic with the stack-heap models shown undecidable in~\cite[Corollary 5.1]{Brotherston&Kanovich14}
(see also~\cite{DemriDeters15bis}),
in which there are propositional variables interpreted by sets of
memory states.
\section{Hilbert-style proof system for \slSW}\label{section:PSL}

\begin{figure*}
\begin{axiombox}[label=axioms:corePSL]{$\coresys$: Axioms for Boolean combinations of core formulae}
\begin{footnotesize}
\begin{enumerate}[align=left,leftmargin=*]
\begin{minipage}{0.4\linewidth}
\item[\axlab{A^\corepedix}{coreAx:EqRef}] $\avariable = \avariable$
\item[\axlab{A^\corepedix}{coreAx:EqSub}] $\aformula \land \avariable = \avariablebis \implies \aformula\completesubstitute{\avariable}{\avariablebis}$
\item[\axlab{A^\corepedix}{coreAx:PointAlloc}] $\avariable \Ipto \avariablebis \implies \alloc{\avariable}$
\end{minipage}%
\begin{minipage}{0.6\linewidth}
\item[\axlab{A^\corepedix}{coreAx:PointInj}] $\avariable \Ipto \avariablebis \land \avariable \Ipto \avariableter \implies \avariablebis = \avariableter$
\item[\axlab{I^\corepedix}{coreAx:Size}] $\size \geq \inbound{+}1 \implies \size \geq \inbound$
\item[\axlab{I^\corepedix}{coreAx:AllocSize}] $\bigwedge_{\avariable \in \asetvar}(\alloc{\avariable} \land \bigwedge_{\avariablebis \in \asetvar \setminus \{\avariable\}} \avariable \neq \avariablebis) \implies \size \geq \card{\asetvar}$
\end{minipage}
\end{enumerate}
\margindown
\end{footnotesize}
\end{axiombox}
\vspace{-5pt}
\begin{axiombox}[label=axioms:starPSL]{Axioms and inference rule for the separating conjunction}
\begin{footnotesize}
\begin{enumerate}[align=left,leftmargin=*]
\addtocounter{enumi}{6}
\begin{minipage}{0.42\linewidth}
\item[\axlab{A^\separate}{starAx:Commute}] $(\aformula \separate \aformulabis) \iff (\aformulabis \separate \aformula)$
\item[\axlab{A^\separate}{starAx:Assoc}] $(\aformula \separate \aformulabis) \separate \aformulater \iff \aformula \separate (\aformulabis \separate \aformulater)$
\item[\axlab{I^\separate}{starAx:DistrOr}] $(\aformula \lor \aformulabis) \separate \aformulater \implies (\aformula \separate \aformulater) \lor (\aformulabis \separate \aformulater)$
\item[\axlab{I^\separate}{starAx:False}] $(\bot \separate \aformula) \iff \bot$
\item[\axlab{A^\separate}{starAx:Emp}] $\aformula \iff \aformula \separate \emp$
\item[\axlab{I^\separate}{starAx:StarAlloc}] $\alloc{\avariable} \separate \true \implies \alloc{\avariable}$
\item[\axlab{I^\separate}{starAx:DoubleAlloc}] $(\alloc{\avariable} \separate \alloc{\avariable})\iff\bot$
\end{minipage}
\begin{minipage}{0.58\linewidth}
\item[\axlab{A^\separate}{starAx:MonoCore}] $\aelement \separate \true \implies \aelement \assuming{\aelement \text{ is } \lnot\emp,\ \avariable = \avariablebis,\ \avariable \neq \avariablebis \text{ or } \avariable \Ipto \avariablebis}$
\item[\axlab{A^\separate}{starAx:AllocNeg}] $\lnot\alloc{\avariable} \separate \lnot\alloc{\avariable} \implies \lnot\alloc{\avariable}$
\item[\axlab{A^\separate}{starAx:PointsNeg}] $(\alloc{\avariable} \land \lnot \avariable \Ipto \avariablebis) \separate \true \implies \lnot \avariable \Ipto \avariablebis$
\item[\axlab{A^\separate}{starAx:AllocSizeOne}] $
\alloc{\avariable} \implies
(\alloc{\avariable} \land \size = 1) \separate \true$
\item[\axlab{A^\separate}{starAx:SizeOne}] $\lnot \emp \implies \size = 1 \separate \true$
\item[\axlab{A^\separate}{starAx:SizeNeg}] $\lnot \size \geq \inbound_1 \separate \lnot \size \geq \inbound_2 \implies \lnot \size \geq \inbound_1{+}\inbound_2{\dotminus}1$
\item[\axlab{A^\separate}{starAx:SizeTwo}] $\alloc{\avariable}\wedge\alloc{\avariablebis}\wedge \avariable\neq\avariablebis\implies \size\geq2$
\end{minipage}
\end{enumerate}
\begin{itemize}[align=left,leftmargin=*]
\item[\rulelab{\textbf{$\separate$-Intro}}{rule:starinference}]
$\inference{\aformula \implies \aformulater}{\aformula \separate \aformulabis \implies \aformulater \separate \aformulabis }{}$
\qquad\qquad\qquad\quad where $a\dotminus b=a-b$ if $a\geq b$, $0$ otherwise.
\end{itemize}
\margindown
\end{footnotesize}
\end{axiombox}
\vspace{-5pt}
\begin{axiombox}[label=axioms:magicwandPSL]{Axioms and inference rules for the separating implication}
\begin{footnotesize}
\begin{enumerate}[align=left,leftmargin=*]
\addtocounter{enumi}{20}
\begin{minipage}{0.56\linewidth}
\item[\axlab{A^{\magicwand}}{wandAx:Size}]$(\size = 1 \land \bigwedge_{\avariable \in \asetvar}\lnot \alloc{\avariable}) \septraction \true\!\assuming{\asetvar \subseteq_{\fin} \PVAR}$
\item[\axlab{A^{\magicwand}}{wandAx:PointsTo}] $\lnot \alloc{\avariable} \implies ((\avariable \Ipto \avariablebis \land \size = 1) \septraction \true)$
\end{minipage}%
\begin{minipage}{0.44\linewidth}
  \item[\rulelab{\textbf{$\separate$-Adj}}{rule:staradj}]
  $\inference{\aformula \separate \aformulabis \implies \aformulater}{\aformula \implies (\aformulabis \magicwand \aformulater)}{}$
  \rulelab{\textbf{$\magicwand$-Adj}}{rule:magicwandadj}
  $\inference{\aformula \implies (\aformulabis \magicwand \aformulater)}{\aformula \separate \aformulabis \implies \aformulater}{}$
\end{minipage}
\item[\axlab{A^{\magicwand}}{wandAx:Alloc}] $\lnot \alloc{\avariable} \implies ((\alloc{\avariable} \land \size = 1 \land \bigwedge_{\avariablebis \in \asetvar}\lnot \avariable \Ipto \avariablebis ) \septraction \true) \assuming{\asetvar \subseteq_{\fin} \PVAR}$
\end{enumerate}
\margindown
\end{footnotesize}
\end{axiombox}
\vspace{-0.4cm}
\end{figure*}
We define a proof system for \slSW, namely $\magicwandsys$,
by relying on its \defstyle{core formulae}: simple \slSW
formulae capturing  essential properties of the models, see e.g.~\cite{Lozes04bis,Yang01}.
It is known that  every \slSW  formula is logically equivalent to a Boolean combination of
core formulae~\cite{Lozes04bis}.
However, as every core formula is  an \slSW formula, we stay
in the original  language and we can derive an axiomatisation of \slSW by
extending the axiom system of propositional calculus with three sets of axioms and inference rules:
the axioms and inference rules of the propositional logic of core formulae (System~\ref{axioms:corePSL}),
the axioms and inference rules witnessing that every formula of the form
$\aformula_1 \separate \aformula_2$, where $\aformula_1, \aformula_2$ are Boolean combinations of core formulae
is logically equivalent to a Boolean combination of core formulae (System~\ref{axioms:starPSL}),
and the axioms and inference rules to eliminate formulae whose outermost connective is the separating implication
$\magicwand$ (System~\ref{axioms:magicwandPSL}).
\LongVersionOnly{ (typically, by taking advantage of inference rules expressing that $\separate$ and $\magicwand$
are adjuncts).
Still, all of this can be performed in different ways, and below we provide a modular approach that can be
extended to richer logics, typically with the list segment predicate $\ls$ and with a guarded
form of first-order quantification, see Section~\ref{section:IntervalSL}.}
The  \defstyle{core formulae}  are expressions of the form
$\avariable = \avariablebis$,
$\alloc{\avariable}$,
$\avariable \Ipto \avariablebis$ and
$\size \geq \inbound$,
where $\avariable,\avariablebis \in \PVAR$ and $\inbound \in \Nat$.
As previously shown, these formulae are from \slSW and are used in the axiom system as abbreviations.
Given
$\asetvar\subseteq_\fin \PVAR$  and $\bound \in \Nat$,
we define $\coreformulae{\asetvar}{\bound}$ as
the set $\{ \avariable = \avariablebis,\ \alloc{\avariable},\ \avariable \Ipto \avariablebis,\ \size \geq \inbound \mid \avariable,\avariablebis \in \asetvar,\ \inbound \in \interval{0}{\bound}\}$.
$\boolcomb{\coreformulae{\asetvar}{\bound}}$ is the set of Boolean combinations of formulae from $\coreformulae{\asetvar}{\bound}$, whereas
$\conjcomb{\coreformulae{\asetvar}{\bound}}$ is the set of conjunctions of literals built upon $\coreformulae{\asetvar}{\bound}$
(a literal being a core formula or its negation).
Given $\aformula = \aliteral_1 \land \dots \land \aliteral_n \in \conjcomb{\coreformulae{\asetvar}{\bound}}$, every $\aliteral_i$ being a literal, $\literals{\aformula} \egdef \{\aliteral_1,\dots,\aliteral_n\}$.
$\aformulabis \inside \aformula$ stands for $\literals{\aformulabis} \subseteq \literals{\aformula}$.
We write $\aformulater \inside \orliterals{\aformula}{\aformulabis}$,
$\orliterals{\aformula}{\aformulabis} \inside \aformulater$ and
$\aformulater \inside \andliterals{\aformula}{\aformulabis}$
for
``$\aformulater \inside \aformula$ or $\aformulater \inside \aformulabis$'',
``$\aformula \inside \aformulater$ or $\aformulabis \inside \aformulater$'', and
``$\aformulater \inside \aformula$ and $\aformulater \inside \aformulabis$'', respectively.\\[3pt]
\noindent
{\bf \textsf{Example.}}
\begin{figure}
\centering
\scalebox{0.9}{
\nddim{3.1ex}{3ex}{1ex}{1em}{1.6em}{.5em}{2.5em}{.2mm}
$
\begin{nd}
\have {825} {\emp \implies \lnot \size \geq 1} \by{\ref{axiom:dubneg} and def. of $\size \geq 1$}{}
\have {850} {\alloc{\avariable} \land \size = 1 \implies \lnot \size \geq 2} \by{\ref{axiom:andelim}}{}
\have {9} {\emp \separate (\alloc{\avariable} \land \size = 1 )\implies \lnot \size \geq 1 \separate \lnot \size \geq 2} \by{\ref{rule:starintroLR}}{825,850}
\have {10} {\lnot \size \geq 1 \separate \lnot \size \geq 2 \implies \lnot \size \geq 2} \by{\ref{starAx:SizeNeg}}{}
\have {11} {\big(\emp \separate (\alloc{\avariable} \land \size = 1)\big) \implies \lnot \size \geq 2} \by{\ref{rule:imptr}}{9,10}
\have {12} {\emp \implies \big((\alloc{\avariable} \land \size = 1) \magicwand \lnot \size \geq 2\big)} \by{\ref{rule:staradj} rule}{11}
\end{nd}
$}
\vspace{-0.1cm}
\caption{Proof of \ $\emp \implies \big((\alloc{\avariable} \land \size = 1) \magicwand \lnot \size \geq 2\big)$.}\label{figure:aproof}
\vspace{-0.3cm}
\end{figure}
%
To show the flavour of the axioms and the rules, in Figure~\ref{figure:aproof}
\ifCSLProc \else (as well as in Appendix~\ref{appendix-paragraphe-an-example}) \fi
we present a proof in $\coresys(\separate,\magicwand)$.
In the proof, a line ``$j\, \mid\, \aformulater \ \ A, i_1,\dots,i_k$''
states that $\aformulater$ is a theorem denoted by the index $j$ and derivable by the axiom or the rule $A$.
If $A$ is a rule, the indices $i_{1},\dots,i_{k} < j$ denote the theorems used as premises in order to derive $\aformulater$.
In the example,
we use the rule~\ref{rule:staradj}, which together with the rule~\ref{rule:magicwandadj} states that the connectives
$\separate$ and $\magicwand$ are  adjoint operators,
as well as the axiom~\ref{starAx:SizeNeg}, stating that $\card{\domain{\aheap}} \leq \inbound_1 {+} \inbound_2$ holds whenever a heap $\aheap$ can be split into two subheaps whose domains have less than $\inbound_1{+}1$ and $\inbound_2{+}1$ elements, respectively.
We also use the following theorems and rules, which can be shown derivable/admissible in the forthcoming calculus:
\begin{nscenter}
  \scalebox{0.9}{
  \lemmalab{\textbf{($\land$Er)}}{axiom:andelim}\,
  $
  \aformulabis \land \aformula \implies \aformula
  $}
  \hfill
    \scalebox{0.9}{
    \lemmalab{\textbf{($\lnot\lnot$E)}}{axiom:dubneg}\,
    $
    \lnot\lnot \aformula \implies \aformula
    $}%
  \hfill
  \scalebox{0.9}{
  \rulelab{\textbf{$\separate$-Ilr}}{rule:starintroLR}
  $
  \inference{\aformula \implies \aformula'\quad \aformulabis \implies \aformulabis'}{\aformula \separate \aformulabis \implies \aformula' \separate \aformulabis'}{}
  $}%
  \hfill
  \scalebox{0.9}{
  \rulelab{\textbf{$\implies$-Tr}}{rule:imptr}
  $
  \inference{\aformula \implies \aformulater\quad\aformulater \implies \aformulabis}{\aformula \implies \aformulabis}
  $}
\end{nscenter}

\subsection{A simple calculus for the core formulae}
\label{subsection:axiomCoreFormulae}
To axiomatise \slSW, we start by introducing the proof
system $\coresys$ (presented in System~\ref{axioms:corePSL}) dedicated to Boolean combinations of core formulae.
$\coresys$ and all the subsequent proof systems contain
the axiom schemata and modus ponens for the propositional calculus.
The  axioms $I^?_i$ in System $n$ are necessary for the fragment the System $n$ governs, but are admissible when the
axioms/rules from the System $n+1$ are present.
In~\ref{coreAx:EqSub},
$\aformula\completesubstitute{\avariable}{\avariablebis}$ stands for
the formula obtained from $\aformula$ by replacing with $\avariable$ every occurrence of $\avariablebis$.
Let $\pair{\astore}{\aheap}$ be a memory state.
The axioms state that $=$ is an equivalence relation (first two axioms),
$\aheap(\astore(\avariable)) = \astore(\avariablebis)$ implies $\astore(\avariable) \in \domain{\aheap}$ (axiom~\ref{coreAx:PointAlloc}) and
that $\aheap$ is a (partial) function (axiom~\ref{coreAx:PointInj}).
Furthermore, there are two intermediate axioms about size formulae: \ref{coreAx:Size} states that if $\domain{\aheap}$ has at least $\inbound{+}1$ elements,
then it has at least $\inbound$ elements,
whereas \ref{coreAx:AllocSize} states instead that if there are $\inbound$ distinct memory cells corresponding to program variables, then indeed $\domain{\aheap} \geq \inbound$.
It is easy to check that $\coresys$ is sound (right-to-left direction of Theorem~\ref{theo:corePSLcompl}, below).
%
In order to establish its completeness with respect to $\boolcomb{\coreformulae{\asetvar}{\bound}}$, we first establish that $\coresys$ is
complete for a fragment of  $\boolcomb{\coreformulae{\asetvar}{\bound}}$, made of \defstyle{core types}.
\LongVersionOnly{
Then, we show that every formula in $\boolcomb{\coreformulae{\asetvar}{\bound}}$ is provably equivalent to a disjunction of core types
(Lemma~\ref{prop:corePSLone}).
}
Let $\asetvar {\subseteq_{\fin}} \PVAR$, $\bound \in \Nat^+$ and $\widehat{\bound} {=} \bound{+}\card{\asetvar}$.
We write $\coretype{\asetvar}{\bound}$ to denote the set of \defstyle{core types} defined by
$
\scaledformulasubset{
  \aformula \in \conjcomb{\coreformulae{\asetvar}{\widehat{\bound}}}
}{
  \bmat[
    \forall\aformulabis{\in} \coreformulae{\asetvar}{\widehat{\bound}},\  \orliterals{\aformulabis}{\lnot \aformulabis} \inside \aformula, \text{ and } (\aformulabis \land \lnot \aformulabis) \not\inside \aformula
  ]
}{0.9}{1}$.
\cut{
\begin{nscenter}
$
\scaledformulasubset{
  \aformula \in \conjcomb{\coreformulae{\asetvar}{\widehat{\bound}}}
}{
  \bmat[
    \forall\aformulabis{\in} \coreformulae{\asetvar}{\widehat{\bound}},\ \orliterals{\aformulabis}{\lnot\aformulabis} \inside \aformula
  ]
}{0.9}{1}$
\end{nscenter}
}
Note that if $\aformula \in \coretype{\asetvar}{\bound}$, then $\aformula$ is a conjunction such that
for every  $\aformulabis \in \coreformulae{\asetvar}{\widehat{\bound}}$, there is exactly one literal in $\aformula$ built upon $\aformulabis$.

\begin{restatable}{lemma}{propcorePSLtwo}
\label{prop:corePSLtwo}
Let $\aformula \in \coretype{\asetvar}{\bound}$.
We have $\neg \aformula$ is valid iff $\vdash_{\coresys} \neg \aformula$.
\end{restatable}

By classical reasoning, one can show that every $\aformula \in \boolcomb{\coreformulae{\asetvar}{\bound}}$
is provably equivalent to a disjunction of core types. Together with Lemma~\ref{prop:corePSLtwo}, this implies that $\coresys$ is complete.

\begin{restatable}{theorem}{theocorePSLcompl}
\stmtdesc{Adequacy}
\label{theo:corePSLcompl}
A Boolean combination of core formulae $\aformula$ is valid iff $\prove_{\coresys} \aformula$.
\end{restatable}

\subsection{A constructive elimination of $\separate$  to axiomatise \slSA
}\label{section:starelimination}
\LongVersionOnly{In order to define a proof system for \slSA,}
We enrich $\coresys$ by adding axioms and inference rule
that handle  $\separate$  (System~\ref{axioms:starPSL}).
The axioms deal with the commutative monoid properties of $\pair{\separate}{\emp}$ and its distributivity
over $\vee$ (as for Boolean BI, see e.g.~\cite{Galmiche&Larchey06}).
In~\ref{starAx:MonoCore}, the notation $\aformula\!\!\assuming{\mathcal{B}}$
refers to the axiom schema $\aformula$ assuming that the Boolean condition $\mathcal{B}$ holds.
The rule~\ref{rule:starinference}
\LongVersionOnly{,
sometimes called ``frame rule'' by analogy with the
rule of the same name in program logic,}
states that logical equivalence is a congruence
for $\separate$.
This allows us to remove the intermediate
axioms~\ref{coreAx:Size} and~\ref{coreAx:AllocSize}
from the proof system.
Hence, we call $\coresys(\separate)$ the proof system obtained from $\coresys$ by adding all schemata from
System~\ref{axioms:starPSL} and removing~\ref{coreAx:Size} and~\ref{coreAx:AllocSize}. It is easy to check that $\starsys$ is sound.
More importantly, $\starsys$ enjoys the $\separate$ elimination property with respect to core types.

\begin{restatable}{lemma}{lemmastarPSLelim}\label{lemma:starPSLelim}
  Let $\aformula$ and $\aformulabis$ in $\coretype{\asetvar}{\bound}$. There is a conjunction of core formulae literals $\aformulater \in \conjcomb{\coreformulae{\asetvar}{2\bound}}$ such that $\prove_{\starsys} \aformula \separate \aformulabis \iff \aformulater$.
\end{restatable}

\begin{proof}(sketch)
  Let $\aformula,\aformulabis{\in} \coretype{\asetvar}{\bound}$. If
  $\aformula$ is unsatisfiable, then $\prove_{\coresys} \aformula\implies \bot$, by Lemma~\ref{prop:corePSLtwo}. By the rule~\ref{rule:starinference}
and the axiom~\ref{starAx:False},
  we get $\prove_{\starsys}\aformula*\aformulabis\implies\bot$ and we take $\aformulater = \bottom$.
  Assume now both $\aformula$ and $\aformulabis$ to be satisfiable.
  Then $\aformula \separate \aformulabis$ can be shown provably equivalent to:
  \begin{nscenter}
    \setlength{\tabcolsep}{3pt}
    \begin{tabular}{rlcl}
      &
      $\bigwedge \formulasubset{\avariable \sim \avariablebis \inside \orliterals{\aformula}{\aformulabis}}{\bmat[\sim \in \{=,\neq\}]}$
      & $\land$ &
      $\bigwedge \aformulasubset{\alloc{\avariable}\inside \orliterals{\aformula}{\aformulabis}}$
      \\ $\land$ &
      $\bigwedge \aformulasubset{\avariable \Ipto \avariablebis \inside \orliterals{\aformula}{\aformulabis}}$
      & $\land$ &
      $\bigwedge \aformulasubset{\lnot\alloc{\avariable} \inside\andliterals{\aformula}{\aformulabis}}$
      \\ $\land$ &
      $\bigwedge \formulasubset{\bottom}{\bmat[\alloc{\avariable}\inside\andliterals{\aformula}{\aformulabis}]}$
      & $\land$ &
      $\bigwedge \formulasubset{\lnot \avariable \Ipto \avariablebis}{\bmat[\alloc{\avariable}\land\lnot\avariable\Ipto\avariablebis\inside \orliterals{\aformula}{\aformulabis}]}$
      \\ $\land$ &
      $\bigwedge \formulasubset{\size\geq\inbound_1{+}\inbound_2}{\bmat[\size\geq\inbound_1\inside\aformula\\\size\geq\inbound_2\inside\aformulabis]}$
      & $\land$ &
      $\bigwedge \formulasubset{\lnot\size\geq\inbound_1{+}\inbound_2{\dotminus}1}{\bmat[\lnot\size\geq\inbound_1\inside\aformula\\\lnot\size\geq\inbound_2\inside\aformulabis]}$
    \end{tabular}
  \end{nscenter}
  This equivalence is reminiscent to the one in~\cite[Lemma 3]{Echenim&Iosif&Peltier19} that is proved semantically.
  In a way, because $\starsys$ will reveal to be complete, the restriction of the proof
  of \cite[Lemma 3]{Echenim&Iosif&Peltier19} to \slSA can actually be
  replayed completely syntactically within $\starsys{}$.
\end{proof}
By the distributivity axiom~\ref{starAx:DistrOr}, this result is extended from core types to arbitrary Boolean combinations of core formulae.
$\starsys$ is therefore complete for \slSA, i.e.\ the logic obtained from \slSW by removing $\magicwand$ and adding
the formulae
 $\alloc{\avariable}$ (only core formulae requiring $\magicwand$).
Then, to prove that a formula $\aformula\in\slSA$ is valid, we repeatedly apply the $\separate$ elimination bottom-up, starting from the leaves of $\aformula$ (which
are Boolean combinations of core formulae) and obtaining a Boolean combination of core formulae $\aformulabis$ that is equivalent to $\aformula$.
We then rely on the completeness of $\coresys$ (Theorem~\ref{theo:corePSLcompl}) to prove that $\aformulabis$ is valid.

\begin{restatable}{theorem}{theostarCompleteness}\label{theo:starCompleteness}
A formula $\aformula$ in \slSA{} is valid iff $\prove_{\starsys} \aformula$.
\end{restatable}

\vspace{-0.2cm}
\subsection{A constructive elimination of $\magicwand$  to axiomatise \slSW
}\label{section:magicwandelimination}
\vspace{-0.1cm}
The proof system  $\magicwandsys$ is defined as $\starsys$ augmented with
the axioms and  inference rules from System~\ref{axioms:magicwandPSL} dedicated to separating implication.
The axioms involving $\septraction$ (kind of dual of $\magicwand$ introduced in Section~\ref{section:preliminaries}) express that it is always
possible to extend a given heap with an extra cell, and that the address and the content of this cell can be fixed arbitrarily
(provided it is not already allocated). The adjunction rules are from the Hilbert-style axiomatisation
of Boolean BI~\cite[Section 2]{Galmiche&Larchey06}.
One can observe that the axioms~\ref{starAx:DistrOr},~\ref{starAx:False},~\ref{starAx:StarAlloc} and~\ref{starAx:DoubleAlloc} are derivable
in $\coresys(\separate,\magicwand)$.
\LongVersionOnly{Moreover, the whole proof system is sound.}
It is easy to check that $\magicwandsys$ is sound.
Analogously, $\magicwandsys$ enjoys the $\magicwand$ elimination property,
stated below by means of $\septraction$.
\begin{restatable}{lemma}{lemmamagicwandPSLelim}\label{lemma:magicwandPSLelim}
  Let $\aformula$ and $\aformulabis$ in $\coretype{\asetvar}{\bound}$. There is a conjunction of core formulae literals $\aformulater \in \conjcomb{\coreformulae{\asetvar}{\bound}}$
  such that $\prove_{\magicwandsys} (\aformula\septraction\aformulabis) \iff \aformulater$.
\end{restatable}

\begin{proof} (sketch)
If either $\aformula$ or $\aformulabis$ is unsatisfiable, then one can show that $\prove_{\magicwandsys}\aformula\septraction\aformulabis \implies \bot$.
Otherwise,
$\aformula \septraction \aformulabis$ can be shown provably equivalent to
\begin{nscenter}
  \setlength{\tabcolsep}{2pt}
  \begin{tabular}{rlclcl}
  &$\bigwedge\formulasubset{\avariable \sim \avariablebis \inside \orliterals{\aformula}{\aformulabis}}{\bmat[\sim \in \{=,\neq\}]}$
  &
  $\land$
  &
  $\bigwedge \aformulasubset{\lnot\alloc{\avariable} \inside \aformulabis}$
  &
  $\land$
  &
  $ \bigwedge \aformulasubset{\lnot \avariable {\Ipto} \avariablebis \inside \aformulabis}$
  \\
  $\land$ &
  $\bigwedge \formulasubset{\alloc{\avariable}}{\bmat[\lnot\alloc{\avariable}\inside\aformula\\ \alloc{\avariable}\inside\aformulabis]}$
  &
  $\land$
  &
  $
  \bigwedge \formulasubset{\avariable \Ipto \avariablebis}{\bmat[\lnot\alloc{\avariable}\inside\aformula\\ \avariable \Ipto \avariablebis \inside \aformulabis]}$
  &
  $\land$
  &
  $
  \bigwedge \formulasubset{\lnot\alloc{\avariable}}{\bmat[\alloc{\avariable}\inside\aformula]}
  $
  \\
  $\land$ &
  $\bigwedge \formulasubset{\size\geq\inbound_2{+}1{\dotminus}\inbound_1}{\bmat[\lnot\size\geq\inbound_1\inside\aformula\\\size\geq\inbound_2\inside\aformulabis]}$
  &
  $\land$
  &
  $
  \bigwedge \formulasubset{\bottom}{\bmat[\avariable\Ipto\avariablebis\inside\aformula\\ \lnot\avariable \Ipto \avariablebis\inside\aformulabis]}$
  &
  $\land$
  &
  $
  \bigwedge \formulasubset{\bottom}{\bmat[\alloc{\avariable}\land\lnot\avariable\Ipto\avariablebis\inside\aformula\\ \avariable \Ipto \avariablebis\inside\aformulabis]}$
  \\
  $\land$ &
  $\bigwedge \formulasubset{\lnot\size\geq\inbound_2{\dotminus}\inbound_1}{\bmat[\size\geq\inbound_1\inside\aformula\\\lnot\size\geq\inbound_2\inside\aformulabis]}$
  &
  $\land$
  &
  $
  \bigwedge \formulasubset{\bottom}{\bmat[\alloc{\avariable}\inside\aformula\\ \lnot\alloc{\avariable}\inside\aformulabis]}$
  \end{tabular}
\end{nscenter}
  where $a\dotminus b$ stands for $a-b$ if $a\geq b$, $0$ otherwise.
  Again, this equivalence is reminiscent to the one in~\cite[Lemma 4]{Echenim&Iosif&Peltier19}  proved semantically.
  Herein, the proof is completely syntactical.
  \LongVersionOnly{
  In a way, because $\magicwandsys{}$ will reveal to be complete, the restriction of the proof
  of \cite[Lemma 4]{Echenim&Iosif&Peltier19}  to $\slSW{}$ can  be
  replayed within $\magicwandsys{}$.
  }
\end{proof}
Again, this result for core types can be extended to arbitrary Boolean combinations of core formulae, as we show that the distributivity of $\septraction$ over disjunctions is provable in $\magicwandsys$.
As a consequence of this development, we achieve one of the main results of the paper.

\begin{restatable}{theorem}{theoPSLcompleteAx}\label{theo:PSLcompleteAx}
$\magicwandsys$ is sound and complete for \slSW.
\end{restatable}

\noindent
{\bf \textsf{What's next?}}
To provide further evidence that our method is robust, we shall apply it
to axiomatise other separation logics, for instance
by  adding the list segment predicate $\ls$~\cite{Berdine&Calcagno&OHearn04}
(or inductive predicates in general)
or first-order quantification.
\cut{
A key step in our approach is first to show  that
the logic admits a characterisation in terms of core formulae and such formulae
need to be designed adequately.}
Of course, the set of valid formulae must be
r.e.,
which discards any attempt with $\seplogic{\separate,\magicwand,\ls}$ or with
the first-order version of \slSW~\cite{DemriLM18,Brochenin&Demri&Lozes12}.
In Section~\ref{section:IntervalSL}, we introduce an extension of  $\seplogic{\separate,\ls}$
 and we axiomatise it with our method, whose main ingredients are recalled below.
\cut{
\noindent
{\bf Big-step vs. small-step axiom schemas.} $\magicwandsys$ is able to simulate the bottom-up elimination of
separating connectives (see Lemmata~\ref{lemma:starPSLelim} and~\ref{lemma:magicwandPSLelim})
when the arguments are Boolean combinations of core formulae.
To do so, $\magicwandsys$ contains axiom schemas that perform such an elimination in multiple ``small-step'' derivations, e.g.\ by deriving a single $\alloc{\avariable}$ predicate from $\alloc{\avariable}\separate\true$ (Axiom~\ref{starAx:StarAlloc}).
Alternatively, it would have been possible to include ``big-step'' axiom schemas that, given two core types, derive in the equivalent formula in one single derivation step (see the proof sketch of Lemma~\ref{lemma:starPSLelim} and the axiom~\ref{starAx2:StarElim} in Section~\ref{section:IntervalSL}).
The main difference is that small-step axioms provide
a simpler understanding of the essential properties of the logic.

\noindent
{\bf What's next? }
To provide further evidence that our method is robust, it is desirable to
apply it to axiomatise other separation logics, for instance
by  adding the list segment predicate $\ls$~\cite{Berdine&Calcagno&OHearn04}
(or more generally user-defined inductive predicates)
or by adding first-order quantification. A key step in our approach is first to show  that
the logic admits a characterisation in terms of core formulae and such formulae
need to be designed adequately. Of course, it is required that the set of valid formulae is recursively
enumerable, which discards any attempt with $\seplogic{\separate,\magicwand,\ls}$ or with
the first-order version of \slSW~\cite{DemriLM18,Brochenin&Demri&Lozes12}.
Below, we introduce an extension of  $\seplogic{\separate,\ls}$
 and we axiomatise it with our method.
}
\subsection{Ingredients of the method}\label{subsection:ingredients}
The Hilbert-style axiomatisation of \slSW has culminated
with Theorem~\ref{theo:PSLcompleteAx} that states the adequateness of  $\magicwandsys$.
Below, we would like to recapitulate the key ingredients of the proposed method, not only to provide
a vade-mecum for axiomatising other separation logics (which we illustrate on the newly
introduced logic  \intervalSL in Section~\ref{section:IntervalSL}), but also to identify the
essential features  and where variations are still possible.

\noindent
{\bf \textsf{Core formulae.}} To axiomatise \slSW internally, the core formulae have played an essential
role. The main properties of these formulae is that their Boolean combinations capture the full logic \slSW~\cite{Lozes04bis}
and all the core formulae can be expressed in \slSW.
Generally speaking, our axiom system naturally leads to a form of constructive completeness, as
advocated in~\cite{Doumane17,Luck18}: the axiomatisation provides
proof-theoretical means to transform any formula into an equivalent Boolean combination
of core formulae, and it contains also a part dedicated to the derivation of valid Boolean combinations
of core formulae (understood as a syntactical fragment of \slSW).
What is specific to each logic is the design of the set of core formulae and in the case of \slSW,
this was already known since~\cite{Lozes04bis}.

\noindent
{\bf \textsf{Big-step vs.\ small-step axiom schemas.}} $\magicwandsys$ simulates the bottom-up elimination of
separating connectives (see Lemmata~\ref{lemma:starPSLelim} and~\ref{lemma:magicwandPSLelim})
when the arguments are two Boolean combinations of core formulae.
To do so, $\magicwandsys$ contains axiom schemas that perform such an elimination in multiple ``small-step'' derivations, e.g.\ by deriving a single $\alloc{\avariable}$ predicate from $\alloc{\avariable}\separate\true$ (axiom~\ref{starAx:StarAlloc}).
Alternatively, it would have been possible to include ``big-step'' axiom schemas that, given the two Boolean combinations of core formulae, derive the equivalent formula in one single derivation step. Instances of this are given in the proof sketch of Lemma~\ref{lemma:starPSLelim}, and later in Section~\ref{section:IntervalSL} (axiom~\ref{starAx2:StarElim}).
The main difference is that small-step axioms provide
a simpler understanding of the key properties of the logic.

\cut{
\noindent
{\bf Relationships with axiomatisation by reduction axioms.}
As described above, in order to axiomatise \slSW, axiom schemas and inference rules are dedicated
to transform formulae into Boolean combinations of core formulae, and part of $\magicwandsys$ is dedicated
to axiomatise Boolean combinations of core formulae, which leads to completeness. In a sense, the calculus
$\magicwandsys$ provides a constructive means to transform any formula into a logically equivalent
Boolean combination of core formulae.
A similar proof schema has been already used for axiomatising so-called dynamic epistemic
logics, see e.g.~\cite{vanBenthem2011ldii,WangC13}, with the introduction of reduction axioms.
In a nutshell, every formula  containing a dynamic operator is provably reduced to a formula without such an operator.
Completeness is then established thanks to the completeness of the underlying `basic' language,
which in our case concerns Boolean combination of core formulae, though presently the axiomatisation of core formulae
requires also some work, and can be very challenging (see e.g. Section~\ref{section:IntervalSL}).
Hilbert-style axiomatisations following similar high-level principles for the
modal separation logics MSL($\separate$,$\Diamond$) and MSL($\separate$,$\langle \neq \rangle$)
introduced in~\cite{Demri&Fervari18}, have been designed in~\cite{Demri&Fervari&Mansutti19}.
}
\section{How to axiomatise internally the  separation logic \intervalSL}\label{section:IntervalSL}
Though core formulae are handful for several existing  separation logics, see e.g. recently~\cite{DemriLM18,Mansutti18,Echenim&Iosif&Peltier19},
we would like to test our method with first-order quantification and reachability predicates, standard features in
 specifications. However,  $\seplogic{\separate,\magicwand,\ls}$ is already known to be
 non-finitely axiomatisable, see the developments in Section~\ref{section:preliminaries}.
So, we need to downgrade our ambitions and we suggest to consider a new logic with guarded quantification and $\ls$
and this is  \intervalSL presented below. Note that the idea of having guarded quantification with second-order
features is not new, see e.g. in~\cite{Gradel&Walukiewicz99}  extensions of the guarded fragment of first-order logic
with fixed points, but herein, this is done in the framework of separation logics and their
axiomatisation.
In short, we introduce the new separation logic \intervalSL that admits the
connective
$\separate$,  the list segment predicate $\ls$ (implicitly) and a guarded form of first-order quantification
involving $\ls$.
It contains  the symbolic heap fragment~\cite{Berdine&Calcagno&OHearn04,Cooketal11} but also richer logics such as
\slSRp (see e.g.~\cite{DemriLM18}).
As a by-product of our completeness proof, we are  able to characterise the complexity of the satisfiability problem for \intervalSL.

\subsection{A guarded logic with $\ls$: \intervalSL}
\label{section-introduction-to-guarded-logic}

Formulae of \intervalSL are defined according to the grammar below (where $\avariable, \avariablebis,\avariableter \in \PVAR$):
\begin{minipage}{0.12\textwidth}
   \vspace{1cm}

   \noindent
   \hfill
   \scalebox{0.9}{
    \begin{tikzpicture}
      \coordinate (cir) at (0,0);
      \def\radius{0.45cm}

      \draw (cir) ++(90:\radius) node[highlightnode] (k) {};
      \draw (cir) ++(0:\radius) node[dot] (g) {};
      \draw (cir) ++(270:\radius) node[highlightnode,label={[yshift=-15pt]{$\scriptstyle{\avariablebis}$}}] (u) {};

      \node[highlightnode] (m) [above = 0.4cm of k] {};
      \node[highlightnode,label=right:{$\scriptstyle{\avariable}$}] (i) [above left=0.4cm and 0.6cm of m] {};
      \draw  (cir) ++(180:\radius) node[highlightnode] (ze) {};

      \node[dot] (j) [above right=0.4cm and 0.6cm of m] {};

      \draw[pto] (j.center) -- (m);
      \draw[pto] (i.center) -- (m);
      \draw[pto] (m.center) -- (k);

      \draw[pto] (cir) ++(90:\radius) arc (90:175:\radius);
      \draw[pto] (cir) ++(180:\radius) arc (180:265:\radius);
      \draw[pto] (cir) ++(-90:\radius) arc (-90:-3:\radius);
      \draw[pto] (cir) ++(0:\radius) arc (0:85:\radius);

      \begin{pgfonlayer}{bg}
        \draw[highlightpath] (i.center) -- (m.center);
        \draw[highlightpath] (cir) ++(90:\radius) arc (90:175:\radius);
        \draw[highlightpath] (cir) ++(180:\radius) arc (180:265:\radius);
        \draw[highlightpath] (m.center) -- (k.center);
      \end{pgfonlayer}
    \end{tikzpicture}
    }
    \hfill\,
\end{minipage}%
\begin{minipage}{0.88\textwidth}
\vspace{2pt}
\noindent\hfill
$
\aformula := \avariable = \avariablebis \mid \avariable \Ipto \avariablebis \mid \emp \mid \lnot \aformula \mid \aformula \land \aformula \mid \aformula \separate \aformula \mid
\inpath{\avariable}{\avariablebis}{\avariableter} \aformula \qquad
$\hfill\,
\vspace{2pt}

\noindent All the syntactic ingredients are standard except the quantifier (denoted
with $\weirdexists$).
Intuitively (the formal definition is provided below), $\inpath{\avariable}{\avariablebis}{\avariableter} \aformula$ is a guarded form of
quantification that is intended to hold true whenever $\avariablebis$ is reachable from $\avariable$ in at least one step, and there is a location
$\alocation$ along the
minimal path between $\avariable$ and  $\avariablebis$ so that the formula $\aformula$ holds whenever $\alocation$ is assigned to $\avariableter$. The figure on the left highlights the possible assignments of $\avariableter$%
\parfillskip=0pt
\end{minipage}%

\noindent
(arrows represent the heap).
Given a heap $\aheap$ and ${\alocation_1,\alocation_2\, {\in}\, \LOC}$, we define $\minpath{\alocation_1}{\alocation_2}{\aheap}$ as the set of locations in the shortest path from
$\alocation_1$ to $\alocation_2$ ($\alocation_2$ possibly excluded). Formally:
\begin{nscenter}
$\minpath{\alocation_1}{\alocation_2}{\aheap} \egdef \left\{ \alocation\in\LOC\ \middle|\
\begin{aligned}
&\text{there are } \alength_1\geq 0 \text{ and } \alength_2 \geq 1 \text{ such that } \aheap^{\alength_1}(\alocation_1) = \alocation, \\
& \aheap^{\alength_2}(\alocation) = \alocation_2 \text{ and, for every } \alength \in \interval{1}{\alength_1{+}\alength_2{-}1},  \aheap^{\alength}(\alocation_1)
\neq \alocation_2
\end{aligned}\right\}$
\end{nscenter}
For example,  $\minpath{\alocation}{\alocation}{\aheap} = \emptyset$ holds iff
$\alocation$ is not in a cycle. Otherwise, $\minpath{\alocation}{\alocation}{\aheap}$ contains all the locations in
the cycle containing $\alocation$. By definition, the minimal paths are preserved when considering heap extensions.
Then, the satisfaction relation $\models$ is completed with
\begin{nscenter}
\begin{tabular}{l}
$\pair{\astore}{\aheap} \models \inpath{\avariable}{\avariablebis}{\avariableter}\ \aformula$ $\equivdef$
$
\minpath{\astore(\avariable)}{\astore(\avariablebis)}{\aheap}\neq\emptyset$ and $\exists
\alocation\, \in\, \minpath{\astore(\avariable)}{\astore(\avariablebis)}{\aheap}\cup\{\astore(\avariablebis)\}$
s.t.\ $\pair{\astore[\avariableter \gets \alocation]}{\aheap} \models \aformula$.
\end{tabular}
\end{nscenter}
We define $\forallinpath{\avariable}{\avariablebis}{\avariableter} \aformula \egdef
      \neg \inpath{\avariable}{\avariablebis}{\avariableter} \neg \aformula$.
In a separation logic {\em lingua} admitting first-order quantification of program variables over the set of locations $\LOC$,
and a
predicate $\reachplus(\avariable,\avariablebis)$ (reachability in at least one step, as in~\cite{DemriLM18}),
the formula $\inpath{\avariable}{\avariablebis}{\avariableter} \aformula$ is equivalent to
\begin{nscenter}
$
\reachplus(\avariable,\avariablebis)
\wedge
\exists \ \avariableter \
\aformula \wedge
(
\avariableter = \avariable \vee \avariableter = \avariablebis \vee
 ((\reachplus(\avariable,\avariableter) \wedge \neg \reachplus(\avariable,\avariablebis))\separate
\reachplus(\avariableter,\avariablebis))).
$
\end{nscenter}
Obviously, \intervalSL does not allow unrestricted
first-order
quantification but it can faithfully define
the
reachability predicates
classically
studied in separation logic~\cite{DemriLM18,Reynolds02}.
$\reachplus(\avariable,\avariablebis)$ is definable as $\inpath{\avariable}{\avariablebis}{\avariableter} \top$,
and allows us to define $\ls(\avariable,\avariablebis)$ and $\reach(\avariable,\avariablebis)$ as shown in~\cite{DemriLM18}:
$\ls(\avariable,\avariablebis) \egdef (\avariable = \avariablebis \wedge \emp) \vee (\avariable \neq \avariablebis \wedge \reachplus(\avariable,\avariablebis)
\wedge \neg (\neg \emp \separate \reachplus(\avariable,\avariablebis)))$,
whereas
$\reach(\avariable,\avariablebis) \egdef \avariable = \avariablebis \lor \reachplus(\avariable,\avariablebis)$.
There are two features of \intervalSL, we would like to emphasize.
First, it is possible to enforce a heap domain of exponential size.
\cut{
Two other key aspects of this logic should be stated.
First, its formulae may be satisfiable only by memory states whose heap domain has an exponential
amount of locations.
}
Indeed, we define the formula $\mbox{\sc R}^n(\avariable,\avariablebis)$ of size linear in $n$, but enforcing the existence of a path of length at least $2^n$ between two distinct locations corresponding to
$\avariable$ and $\avariablebis$, respectively.
$\mbox{\sc R}^0(\avariable,\avariablebis) \egdef \avariable \neq \avariablebis \wedge
\inpath{\avariable}{\avariablebis}{\avariableter} \top$,
\LongVersionOnly{
(existence of a path of length at least one between $\avariable$ and $\avariablebis$),
}
whereas for $n \geq 0$, $\mbox{\sc R}^{n+1}(\avariable,\avariablebis)$ is defined as
\begin{nscenter}
$\avariable {\neq} \avariablebis \wedge
\inpath{\avariable}{\avariablebis}{\avariableter}\,
\forallinpath{\avariable}{\avariablebis}{\avariableter'}\, \forallinpath{\avariable}{\avariablebis}{\avariableter''}
\left((\avariableter' = \avariable \wedge \avariableter'' = \avariableter) \lor
     (\avariableter' = \avariableter \wedge \avariableter'' = \avariablebis) \Rightarrow \mbox{\sc R}^n(\avariableter',\avariableter'')\right).$
\end{nscenter}
Nevertheless, in Section~\ref{subsection:PSpace} we show how the satisfiability and validity problems for
\intervalSL  are in \pspace.
Another interesting feature of \intervalSL is illustrated by its ability to state that from two
locations corresponding to program variables (say $\avariable, \avariablebis$),
it is possible to reach a different location, which in turn reaches another location corresponding to a program variable
(say $\avariableter$). This can be done with the
formula $\inpath{\avariable}{\avariableter}{\avariablesix}(\reachplus({\avariablebis},{\avariablesix}) \land \bigwedge_{\avariablefifth \in \{\avariable,\avariablebis,\avariableter\}} \avariablesix \neq \avariablefifth)$.
Thus, the logic is able to express that two paths meet at a specific location. This naturally leads to the notion of meet-points, introduced next in order to define the core formulae for \intervalSL.
\subsection{Core formulae are back!}
In order to axiomatise internally \intervalSL with our method, we need to possess
a set of core formulae that captures \intervalSL.
Below, we design
such core formulae and establish its appropriateness.
They make intensive use of meet-point terms, a concept introduced in~\cite{DemriLM18} but that will play a crucial
role herein.
Informally, given a memory state $\pair{\astore}{\aheap}$,
a meet-point between $\astore({\avariable})$ and $\astore({\avariablebis})$ leading to $\astore({\avariableter})$  is a location
$\alocation$ such that (\rom{1}) $\alocation$  reaches $\astore({\avariableter})$, (\rom{2}) both locations
$\astore(\avariable)$ and $\astore({\avariablebis})$ reach $\alocation$,
and (\rom{3}) there is no location
\begin{minipage}{0.83\textwidth}
\noindent $\alocation'$ satisfying these properties and reachable from $\astore({\avariable})$ in strictly
fewer steps. A meet-point term of the form $\ameetvar{\avariable}{\avariablebis}{\avariableter}$, where
$\avariable,\avariablebis,\avariableter \in \PVAR$, is then an expression that, given a memory state $\pair{\astore}{\aheap}$, is intended to be interpreted by a
meet-point between $\astore(\avariable)$ and $\astore(\avariablebis)$ leading to $\astore(\avariableter)$ (if it exists).
The figure on the right shows
some of the meet-points between $\avariable$ and other program variables, highlighting their distribution in a
memory state.
In particular,%
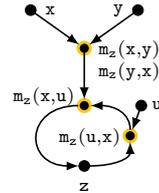
\parfillskip=0pt
\end{minipage}%
\begin{minipage}{0.17\textwidth}
   \noindent\hfill
    \begin{tikzpicture}
      \coordinate (cir) at (0,0);
      \def\radius{0.4cm}

      \draw (cir) ++(90:\radius) node[highlightnode,label={[xshift=-15pt,yshift=-6pt]{$\scriptstyle{\ameetvar{\avariable}{\avariablefour}{\avariableter}}$}}] (k) {};

      \draw (cir) ++(270:\radius) node[dot,label={[yshift=-15pt]{$\scriptstyle{{\avariableter}}$}}] (u) {};

      \node[highlightnode,label={[xshift=17pt,yshift=-18pt]$\substack{
      \ameetvar{\avariable}{\avariablebis}{\avariableter}\\\ameetvar{\avariablebis}{\avariable}{\avariableter}}$}] (m) [above = 0.6cm of k] {};
      \node[dot,label=right:{$\scriptstyle{\avariable}$}] (i) [above left=0.4cm and 0.6cm of m] {};
      \node[dot,label=left:{$\scriptstyle{\avariablebis}$}] (j) [above right=0.4cm and 0.6cm of m] {};

      \draw  (cir) ++(0:\radius) node (wow) {};
      \node[dot,label={[xshift=6pt,yshift=-7pt]{$\scriptstyle{\avariablefour}$}}] (z) [right = 0.6cm of k] {};

      \node[highlightnode,label={[xshift=-15pt,yshift=-10pt]$\scriptstyle{
      \ameetvar{\avariablefour}{\avariable}{\avariableter}}$}] (ze) [right = 0cm of wow] {};

      \draw[pto] (i.center) -- (m);
      \draw[pto] (j.center) -- (m);
      \draw[pto] (m.center) -- (k);
      \draw[pto] (z.center) -- (ze);

      \draw[pto] (k) .. controls +(180:0.8) and +(180:0.8) .. (u);
      \draw[pto] (ze) .. controls +(90:0.3) and +(0:0.4) .. (k);
      \draw[pto] (u) .. controls +(0:0.4) and +(-90:0.3) .. (ze);
    \end{tikzpicture}
\end{minipage}

\noindent notice how in the figure, $\ameetvar{\avariable}{\avariablefour}{\avariableter}$ is different from $\ameetvar{\avariablefour}{\avariable}{\avariableter}$, which happens
because of the condition (\rom{3}) and as the two corresponding locations are in a cycle. We call this type of meet-points \emph{asymmetric}.
We now formalise these concepts.
Given $\asetvar \subseteq \PVAR$, we write $\atermset{\asetvar}$ to denote the set
 $\asetvar \cup \set{ \ameetvar{\avariable}{\avariablebis}{\avariableter} \mid \avariable,\avariablebis,\avariableter \in \asetvar}$.
Elements of $\atermset{\PVAR}$ are
called \defstyle{terms}.
The terms of the form $\ameetvar{\avariable}{\avariablebis}{\avariableter}$ are syntactic constructs called \defstyle{meet-point terms}.
Terms are denoted with $\aterm,\aterm_1,\aterm_2,\dots$, when we do not need to distinguish between variables and meet-point terms.
In order to give a semantics to these objects, we interpret the terms
on memory states by means of the interpretation function
$\semantics{.}_{\astore,\aheap}\!: \atermset{\PVAR} \to \LOC$ such that $\semantics{\avariable}_{\astore,\aheap} \egdef
 \astore(\avariable)$ for each $\avariable \in \PVAR$, and
$\semantics{\ameetvar{\avariable}{\avariablebis}{\avariableter}}_{\astore,\aheap}$ is defined and takes the value $\alocation$ iff
there are $\alength_1,\alength_2 \geq 0$ such that
\begin{itemize}[nosep]
\item $\aheap^{\alength_1}(\astore(\avariable)) = \aheap^{\alength_2}(\astore(\avariablebis)) = \alocation$ and there is $\alength \geq 0$ such that  $\aheap^\alength(\alocation) = \astore(\avariableter)$;
\item for every $\alength_1' \in \interval{0}{\alength_1-1}$ and $\alength_2' \geq 0$,
$\aheap^{\alength_1'}(\astore(\avariable)) \neq \aheap^{\alength_2'}(\astore(\avariablebis))$.
\end{itemize}
%
%
%
%
%
One last object is needed in order to define the core formulae.
Given a memory state $\pair{\astore}{\aheap}$
\LongVersionOnly{a finite set of terms $\asetmeetvar \subseteq_{\fin} \atermset{\PVAR}$} and
a finite set of pairs of terms $\asetpath \subseteq_{\fin} \atermset{\PVAR} \times \atermset{\PVAR}$,
we write
\LongVersionOnly{
$\pathset{\alocation}{\astore,\aheap}{\asetmeetvar}$ to denote
the set of locations reachable from $\alocation$ without going through locations
(different from $\alocation$) corresponding to terms in $\asetmeetvar$ and
}
$\remset{\astore,\aheap}{\asetpath}$ to denote the subset of $\domain{\aheap}$
made of the locations that are not in the path between two locations corresponding to terms in a pair of $\asetpath$.
Formally: $\remset{\astore,\aheap}{\asetpath} \egdef \domain{\aheap} \setminus
\big(\bigcup_{\pair{\aterm_1}{\aterm_2}\in\asetpath} \minpath{\semantics{\aterm_1}_{\astore,\aheap}}{\semantics{\aterm_2}_{\astore,\aheap}}{\aheap}\big)$.
\LongVersionOnly{
\begin{itemize}
\itemsep 0 cm
\item
$\pathset{\alocation}{\astore,\aheap}{\asetmeetvar} \egdef \{ \alocation' \in \domain{\aheap} \mid \exists \alength \geq 0\ \aheap^\alength(\alocation) = \alocation' \text{ and } \aheap^{\alength'}(\alocation)
\not\in\semantics{\asetmeetvar}_{\astore,\aheap} \text{ for every } 0 < \alength' \leq \alength \}$
where $\semantics{\asetmeetvar}_{\astore,\aheap}$ is equal to
$\set{\semantics{\aterm}_{\astore,\aheap} \mid \aterm \in \asetmeetvar}$.
\item $\remset{\astore,\aheap}{\asetpath} \egdef \domain{\aheap} \setminus \big(\bigcup_{\pair{\aterm_1}{\aterm_2}\in\asetpath} \minpath{\semantics{\aterm_1}_{\astore,\aheap}}{\semantics{\aterm_2}_{\astore,\aheap}}{\aheap}\big)$.
\end{itemize}
Note that the only location in $\pathset{\alocation}{\astore,\aheap}{\asetmeetvar}$ that may correspond
to the interpretation of a term in  $\asetmeetvar$, is $\alocation$ itself.
Moreover, $\pathset{\alocation}{\astore,\aheap}{\asetmeetvar}=\emptyset$ if and only if $\alocation \not\in\domain{\aheap}$.
}

The \defstyle{core formulae} are expressions of the form:
$\aterm_1 = \aterm_2$,
$\seesgeq{\aterm_1}{\aterm_2}{\asetmeetvar}{\inbound{+}1}$ and
$\remgeq{\asetpath}{\inbound}$,
where  $\aterm_1,\aterm_2 \in \atermset{\PVAR}$,  $\asetmeetvar \subseteq_\fin \atermset{\PVAR}$, $\asetpath \subseteq_\fin \atermset{\PVAR}\times \atermset{\PVAR}$ and $\inbound \in \Nat$.
We write $\sees{\aterm_1}{\aterm_2}{\asetmeetvar}$ for $\seesgeq{\aterm_1}{\aterm_2}{\asetmeetvar}{1}$.
The satisfaction relation $\models$ is extended to core formulae:
\begin{itemize}[nosep]
\item $\pair{\astore}{\aheap} \models \aterm_1 = \aterm_2$ $\equivdef$ $\semantics{\aterm_1}_{\astore,\aheap} =
\semantics{\aterm_2}_{\astore,\aheap}$; \hfill\labelitemi\ $\pair{\astore}{\aheap} \models \remgeq{\asetpath}{\inbound}$ $\equivdef$
$\card{\remset{\astore,\aheap}{\asetpath}} \geq \inbound$;
\item {
\addtolength{\jot}{-4pt}
$
\begin{aligned}[t]
\pair{\astore}{\aheap} \models \seesgeq{\aterm_1}{\aterm_2}{\asetmeetvar}{\inbound} \, \equivdef \, &
\text{there is}\ \delta \geq \inbound\ \text{such that}\ \aheap^{\delta}(\semantics{\aterm_1}_{\astore,\aheap}) = \semantics{\aterm_2}_{\astore,\aheap}\
\text{and for }\\ &\text{all}\
\delta' \in \interval{1}{\delta-1},
\aheap^{\delta'}(\semantics{\aterm_1}_{\astore,\aheap})\,{\not\in}\,
\set{\semantics{\aterm_2}_{\astore,\aheap}} \cup \set{\semantics{\aterm}_{\astore,\aheap} \mid \aterm\,{\in}\,\asetmeetvar}.
\end{aligned}$}
\LongVersionOnly{
the conditions
below hold:
\begin{enumerate}
\item $\card{\minpath{\semantics{\aterm_1}_{\astore,\aheap}}{\semantics{\aterm_2}_{\astore,\aheap}}{\aheap}} \geq \inbound$
      (existence of a path of length at least $\inbound$),
\item $\minpath{\semantics{\aterm_1}_{\astore,\aheap}}{\semantics{\aterm_2}_{\astore,\aheap}}{\aheap} \subseteq
      \pathset{\semantics{\aterm_1}_{\astore,\aheap}}{\astore,\aheap}{\asetmeetvar}$ (no location corresponding to the interpretation
      of a term in $\asetmeetvar$ is on the path from $\semantics{\aterm_1}_{\astore,\aheap}$ to $\semantics{\aterm_2}_{\astore,\aheap}$).
\end{enumerate}
}
\LongVersionOnly{
\item $\pair{\astore}{\aheap} \models \remgeq{\asetpath}{\inbound}$ iff
$\card{\remset{\astore,\aheap}{\asetpath}} \geq \inbound$.
}
\end{itemize}
As earlier in Section~\ref{section:PSL},
we write $\coreformulae{\asetvar}{\bound}$ to denote the set of core formulae restricted to terms from
$\atermset{\asetvar}$, where $\asetvar \subseteq_\fin \PVAR$ and $\inbound$ is bounded above by $\bound$.
In order to become more
\begin{minipage}{0.17\textwidth}
   \noindent
   \begin{tikzpicture}
     \coordinate (cir) at (0,0);
     \def\radius{0.45cm}

     \node[highlightnode,label=above:{$\scriptstyle{\avariable}$}] (x) at (0,0) {};
     \node[highlightnode,label=right:{$\scriptstyle{\avariablefour}$}] (u) [below right= 0.6cm and 0.8cm of x] {};
     \node[highlightnode,label=right:{$\scriptstyle{\ameetvar{\avariable}{\avariablebis}{\avariableter}}$}] (m) [below = 0.6cm of u] {};
     \node[highlightnode,label=above:{$\scriptstyle{\avariablebis}$}] (y) [above right = 0.6cm and 0.8cm of m] {};
     \node[highlightnode,label=below:{$\scriptstyle{\avariableter}$}] (z) [below left = 0.6cm and 0.8cm of m] {};

     \draw[pto] (x.center) -- (u);
     \draw[pto] (u.center) -- (m);
     \draw[pto] (y.center) -- (m);
     \draw[pto] (m.center) -- (z);
     \begin{pgfonlayer}{bg}
       \draw[highlightpath] (x.center) -- (u.center);
       \draw[highlightpath] (y.center) -- (m.center);
       \draw[highlightpath] (m.center) -- (z.center);
     \end{pgfonlayer}
   \end{tikzpicture}
   \hfill
\end{minipage}%
\begin{minipage}{0.83\textwidth}
\noindent familiar with these core formulae, let us consider the memory state $\pair{\astore}{\aheap}$ outlined on the left.
Since both $\astore(\avariable)$ and $\astore(\avariablebis)$ reach $\astore(\avariableter)$,  $\semantics{\ameetvar{\avariable}{\avariablebis}{\avariableter}}_{\astore,\aheap}$ is defined, or
alternatively $\pair{\astore}{\aheap} \models \ameetvar{\avariable}{\avariablebis}{\avariableter} = \ameetvar{\avariable}{\avariablebis}{\avariableter}$.
Therefore, we have that $\pair{\astore}{\aheap} \models \sees{\avariable}{\ameetvar{\avariable}{\avariablebis}{\avariableter}}{\emptyset}$.
We also note that $\astore(\avariablefour)$ is a location in the minimal path from $\astore(\avariable)$ to $\semantics{\ameetvar{\avariable}{\avariablebis}{\avariableter}}_{\astore,\aheap}$.
However, as $\astore(\avariablefour)$ is distinct from these two locations, we conclude that
$\pair{\astore}{\aheap} \models \lnot\sees{\avariable}{\ameetvar{\avariable}{\avariablebis}{\avariableter}}{\{\avariablefour\}}$.
Lastly, let us take for example the sets of locations corresponding to the two paths%
\parfillskip=0pt
\end{minipage}
\vskip1pt
\noindent highlighted in yellow:
$\minpath{\astore(\avariable)}{\astore(\avariablefour)}{\aheap}$ and $\minpath{\astore(\avariablebis)}{\astore(\avariableter)}{\aheap}$.
The location $\astore(\avariablefour)$ does not belong to any of these sets.
As it is in $\domain{\aheap}$, we conclude that $\pair{\astore}{\aheap} \models \remgeq{\{\pair{\avariable}{\avariablefour},\pair{\avariablebis}{\avariableter}\}}{1}$.
\cut{
We write $\ameetset{\asetvar}$ to denote the set
$\set{ \ameetvar{\avariable}{\avariablebis}{\avariableter} \mid \avariable,\avariablebis,\avariableter \in \asetvar}$.
We denote with $\MEET \egdef \ameetset{\PVAR}$ the set of all meet-point variables.
Given a Boolean combination of core formulae $\aformula$, we denote with $\terms{\aformula}$ the set of terms appearing in $\aformula$.
We define the set of free variables $\freevars{\aformula}$ of a formula $\aformula$ as follows:
\begin{multicols}{2}
\begin{itemize}
\item $\freevars{\avariable \Ipto \avariablebis} \egdef \freevars{\avariable = \avariablebis} \egdef \{\avariable,\avariablebis\}$
\item $\freevars{\emp} = \emptyset$
\item $\freevars{\lnot \aformulabis} = \freevars{\aformulabis}$
\item $\freevars{\aformulabis_1 \land \aformulabis_2} = \freevars{\aformulabis_1} \cup \freevars{\aformulabis_2}$
\item $\freevars{\aformulabis_1 \separate \aformulabis_2} = \freevars{\aformulabis_1} \cup \freevars{\aformulabis_2}$
\item $\freevars{\inpath{\avariable}{\avariablebis}{\avariableter}\aformulabis} \egdef (\freevars{\aformulabis}\setminus\{\avariableter\}) \cup \{\avariable,\avariablebis\}$
\end{itemize}
\end{multicols}

 Given a Boolean combination of core formulae $\aformula$, we denote with $\terms{\aformula}$ the set of terms appearing in $\aformula$. This set is then recursively defined as
\begin{multicols}{2}
\begin{itemize}
\item $\terms{\aterm_1 = \aterm_2} \egdef \{\aterm_1,\aterm_2\}$
\item $\terms{\lnot \aformulabis} = \terms{\aformulabis}$
\item $\terms{\aformulabis_1 \land \aformulabis_2} = \terms{\aformulabis_1} \cup \terms{\aformulabis_2}$
\item $\terms{\seesgeq{\aterm_1}{\aterm_2}{\asetmeetvar}{\inbound}} \egdef \{\aterm_1,\aterm_2\} \cup \asetmeetvar$
\item $\terms{\remgeq{\asetpath}{\inbound}} \egdef \terms{\asetpath}$
\end{itemize}
\end{multicols}
where, for a given set of pairs of terms $\asetpath$, $\terms{\asetpath}$ is defined as $\pi_1(\asetpath)\cup\pi_2(\asetpath)$.

Given a term $\aterm$ we define the set of its characters $\chars{\aterm}$ as $\chars{\avariable} \egdef \{ \avariable\}$ and $\chars{\ameetvar{\avariable}{\avariablebis}{\avariableter}} \egdef \{\avariable,\avariablebis,\avariableter\}$.
The definition is trivially extended to set of terms $\asetmeetvar$ i.e. $\chars{\asetmeetvar} \egdef \bigcup_{\aterm \in \asetmeetvar} \chars{\aterm}$. Similarly, $\chars{\aformula} \egdef \chars{\terms{\aformula}}$ and $\chars{\asetpath} \egdef \chars{\terms{\asetpath}}$.
Moreover, the set of variables of a set of terms (resp. formula and set of pairs of terms) is defined as $\vars{\asetmeetvar} \egdef \asetmeetvar \cap \PVAR$ (resp. $\vars{\aformula} \egdef \vars{\terms{\aformula}}$ and $\vars{\asetpath} \egdef \vars{\terms{\asetpath}}$).
Lastly, the set of meet-points of a set of terms (resp. formula and set of pairs of terms) is defined as $\meets{\asetmeetvar} \egdef \asetmeetvar \cap \MEET$ (resp. $\meets{\aformula} \egdef \meets{\terms{\aformula}}$ and $\meets{\asetpath} \egdef \meets{\terms{\asetpath}}$).
}

\vspace{3pt}
\noindent
{\bf \textsf{Expressing core formulae in \intervalSL.}}
A crucial point 
for axiomatising 
\slSW
is that every core formula is a mere abbreviation for a formula of the logic.
This is the property that leads to an {\em internal} axiomatisation.
The same holds for \intervalSL as one can show that
every core formula  can be defined in \intervalSL and, in the forthcoming axiomatisation, should be considered as an abbreviation.
For example,
the formula $\seesgeq{\avariable}{\avariablebis}{\emptyset}{\inbound}$ can be shown equivalent to
$(\strict{\reachplus(\avariable,\avariablebis)} \land \sizegeq{\inbound}) \separate \true$,
where $\strict{\aformula}$ is a shortcut for $\aformula \land \lnot (\lnot\emp \separate \aformula)$ and states that $\aformula$ holds in the current model, say $\pair{\astore}{\aheap}$ but does not hold in any submodel
(i.e. in $\pair{\astore}{\aheap'}$ where $\aheap' \sqsubset \aheap$).
Similarly, $\avariable = \ameetvar{\avariablebis}{\avariableter}{\avariablefour}$
is equivalent to
\begin{nscenter}
$
\reach(\avariable,\avariablefour) \land
(\reach(\avariablebis,\avariable)\separate \reach(\avariableter,\avariable)) \land
(\reachplus(\avariable,\avariable) \implies (\reach(\avariablebis,\avariable) \separate \reachplus(\avariable,\avariable))),
$
\end{nscenter}
whereas $\ameetvar{\avariable}{\avariablebis}{\avariableter} {=} \ameetvar{\avariablefour}{\avariablefifth}{\avariablesix}$
is
$
\inpath{\avariable}{\avariableter}{\avariableoct} (\ameetvar{\avariable}{\avariablebis}{\avariableter} = \avariableoct \land \avariableoct = \ameetvar{\avariablefour}{\avariablefifth}{\avariablesix})$,
where $\avariableoct \not \in \set{\avariable,\avariablebis,\avariableter,\avariablefour,\avariablefifth,\avariablesix}$.

\begin{restatable}{lemma}{lemmaintcoreexpress}\label{lemma:coreIntCoreExpress}
Every core formula is logically equivalent to a formula of \intervalSL.
\end{restatable}

\LongVersionOnly{
Below, we provide the definitions for the different core formulae.
For each core formula $\aformulabis$, we define its extension $\extension{\aformulabis}$ in \intervalSL
that is logically equivalent. The notation $\extension{\aformulabis}$ is used for formal statements but by slight abuse,
we avoid it when there is no confusion.
As a rule of thumb, when core formulae
are used to define other core formulae, their expansion in \intervalSL has been defined above. For example,
equalities are defined firstly when at least one variable is involved and then we consider equalities
when two meet-point terms are involved.
A similar principle applies for other core formulae, and moreover, for the core formulae involving a finite set
of terms, our definitions are done inductively involving strictly smaller subsets. Let us start by the equalities.

\begin{itemize}
\itemsep 0 cm
\item $\avariable = \avariablebis$ belongs to \intervalSL whereas $\avariable = \ameetvar{\avariablebis}{\avariableter}{\avariablefour}$ and $\ameetvar{\avariablebis}{\avariableter}{\avariablefour} = \avariable$ are both equal to
\begin{nscenter}
$
\reach(\avariable,\avariablefour) \land
(\reach(\avariablebis,\avariable)\separate \reach(\avariableter,\avariable)) \land
(\reachplus(\avariable,\avariable) \implies (\reach(\avariablebis,\avariable) \separate \reachplus(\avariable,\avariable))).
$
\end{nscenter}
\item $\ameetvar{\avariable}{\avariablebis}{\avariableter} {=} \ameetvar{\avariablefour}{\avariablefifth}{\avariablesix}$
$\equivbydef$
$
\inpath{\avariable}{\avariableter}{\avariableoct} (\ameetvar{\avariable}{\avariablebis}{\avariableter} = \avariableoct \land \avariableoct = \ameetvar{\avariablefour}{\avariablefifth}{\avariablesix})$,
where $\avariableoct \not \in \set{\avariable,\avariablebis,\avariableter,\avariablefour,\avariablefifth,\avariablesix}$.
\end{itemize}

Now, let us define the core formulae with the predicate $\mathtt{sees}$ but only those involving as main arguments
program variables only. We use $\strict{\aformula}$ as a shortcut for $\aformula \land \lnot (\lnot\emp \separate \aformula)$.

\begin{itemize}
\item $\seesgeq{\avariable}{\avariablebis}{\emptyset}{\inbound}$ is equal to
$(\strict{\reachplus(\avariable,\avariablebis)} \land \sizegeq{\inbound}) \separate \true$.
\item $\seesgeq{\avariable}{\avariablebis}{\asetmeetvar \cup \{\avariableter\}}{\inbound}$ is equal to:
\begin{nscenter}
 $\seesgeq{\avariable}{\avariablebis}{\asetmeetvar}{\inbound} \land
\big(\avariableter = \avariable \lor \avariableter = \avariablebis \lor (\true \separate (\reachplus(\avariable,\avariablebis) \land \lnot \reachplus(\avariable,\avariableter)))\big).
$
\end{nscenter}
\item $\seesgeq{\avariable}{\avariablebis}{\asetmeetvar \cup \{\ameetvar{\avariableter}{\avariablefour}{\avariablefifth}\}}{\inbound}$
is equal to
\begin{nscenter}
$\seesgeq{\avariable}{\avariablebis}{\asetmeetvar}{\inbound} \land
 ( \ameetvar{\avariableter}{\avariablefour}{\avariablefifth} = \avariable \lor \ameetvar{\avariableter}{\avariablefour}{\avariablefifth} = \avariablebis \lor
  \lnot \inpath{\avariable}{\avariablebis}{\avariablesept} \ameetvar{\avariableter}{\avariablefour}{\avariablefifth} = \avariablesept
  )$
\end{nscenter}
  where $\avariablesept$ is a program variable that does not appear in $\set{\avariable,\avariablebis,\avariableter,\avariablefour,\avariablefifth} \cup\chars{\asetmeetvar}$.
\end{itemize}
Now, let us define the core formulae with the predicate $\mathtt{sees}$ involving at least one meet-point term.
They can be reduced easily to the previous ones by using the $\weirdexists$ quantifier.
\begin{itemize}
\item $\seesgeq{\avariable}{\ameetvar{\avariablebis}{\avariableter}{\avariablefour}}{\asetmeetvar}{\inbound}$ is equal to
\begin{nscenter}
$(
\avariablebis = \ameetvar{\avariablebis}{\avariableter}{\avariablefour} \land \seesgeq{\avariable}{\avariablebis}{\asetmeetvar}{\inbound}) \lor
 \inpath{\avariablebis}{\avariablefour}{\avariablesept} (\avariablesept = \ameetvar{\avariablebis}{\avariableter}{\avariablefour} \land \seesgeq{\avariable}{\avariablesept}{\asetmeetvar}{\inbound}),
$
\end{nscenter}
where $\avariablesept$ is a program variable that is not in $\set{\avariable,\avariablebis,\avariableter,\avariablefour} \cup \chars{\asetmeetvar}$.
It is worth noting that the atomic predicates involved in these expansions only
contains program variables.
\item Similarly, the formula $\seesgeq{\ameetvar{\avariablebis}{\avariableter}{\avariablefour}}{\avariable}{\asetmeetvar}{\inbound}$ is equal to
\begin{nscenter}
$(
\avariablebis = \ameetvar{\avariablebis}{\avariableter}{\avariablefour} \land \seesgeq{\avariablebis}{\avariable}{\asetmeetvar}{\inbound}) \lor \inpath{\avariablebis}{\avariablefour}{\avariablesept} (\avariablesept = \ameetvar{\avariablebis}{\avariableter}{\avariablefour} \land \seesgeq{\avariablesept}{\avariable}{\asetmeetvar}{\inbound})$,
\end{nscenter}
where $\avariablesept$ is a program variable that does not appear in $\set{\avariable,\avariablebis,\avariableter,\avariablefour} \cup \chars{\asetmeetvar}$.
\item Lastly, $\seesgeq{\ameetvar{\avariable}{\avariablebis}{\avariableter}}{\ameetvar{\avariablefour}{\avariablefifth}{\avariablesix}}{\asetmeetvar}{\inbound}$ is equal to
\begin{nscenter}
$
(\avariable = \ameetvar{\avariable}{\avariablebis}{\avariableter} \land \seesgeq{\avariable}{\ameetvar{\avariablefour}{\avariablefifth}{\avariablesix}}{\asetmeetvar}{\inbound}) \lor
\inpath{\avariable}{\avariableter}{\avariablesept} (\avariablesept = \ameetvar{\avariable}{\avariablebis}{\avariableter} \land
\seesgeq{\avariablesept}{\ameetvar{\avariablefour}{\avariablefifth}{\avariablesix}}{\asetmeetvar}{\inbound}),
$
\end{nscenter}
where $\avariablesept$ is a program variable that does not appear in $\set{\avariable,\avariablebis,\avariableter,\avariablefour,\avariablefifth,\avariablesix} \cup \chars{\asetmeetvar}$.
\end{itemize}
In order to define $\remgeq{\asetpath}{\inbound}$, we need first to identify the
pairs $\pair{\aterm_1}{\aterm_2} \in \asetpath$ such that $\sees{\aterm_1}{\aterm_2}{\emptyset}$ holds.
So, the first idea to express $\remgeq{\asetpath}{\inbound}$ is to identify the maximal subset $\asetpath' \subseteq \asetpath$
such that all the pairs $\pair{\aterm_1}{\aterm_2}$ in $\asetpath'$ satisfies $\sees{\aterm_1}{\aterm_2}{\emptyset}$
and to state that the heap truncated from the paths enforcing $\sees{\aterm_1}{\aterm_2}{\emptyset}$ has at least
$\inbound$ memory cells. This would work if all the terms were variables but with meet-point terms, taking a subheap may change
the existence of a location interpreting a meet-point term.
That is why, when a meet-point term $\ameetvar{\avariable}{\avariablebis}{\avariableter}$ is present in
$\asetpath'$, we first check whether it is equal to $\avariable$ or to $\avariableter$. If it is not the case,
that is $\reachplus(\avariable,\avariableter)$ holds, then we can use a guarded first-order quantification
of the form $\inpath{\avariable}{\avariableter}{\avariablesept}$ to capture the value
$\ameetvar{\avariable}{\avariablebis}{\avariableter}$ in the current heap, and then this value can be referred to
in any subheap thanks to the variable $\avariablesept$. Hence, in the definition of the extension of core formula
 $\remgeq{\asetpath}{\inbound}$ below, we perform a disjunction on the potential subsets
$\asetpath' \subseteq \asetpath$ and on whether if $\ameetvar{\avariable}{\avariablebis}{\avariableter}
\in \meets{\asetpath'}$, then $\avariable = \avariableter$ or not.
The formula $\remgeq{\asetpath}{\inbound}$ is equal to
\begin{nscenter}
$
\begin{aligned}[t]
\bigvee_{\mathclap{\substack{\asetpath' \subseteq \asetpath\\
  \meets{\asetpath'} = \{\ameetvar{\avariable_i}{\avariablebis_i}{\avariableter_i} \mid i \in \interval{1}{n}\}\\
  K \subseteq \interval{1}{n}
  }}}
  \Big(
  {\textstyle\bigwedge_{\pair{\aterm_1}{\aterm_2} \in \asetpath'}} \sees{\aterm_1}{\aterm_2}{\emptyset} \land
    \bigwedge_{\mathclap{\pair{\aterm_1}{\aterm_2} \in \asetpath \setminus \asetpath'}} \lnot \sees{\aterm_1}{\aterm_2}{\emptyset}
     \land
     \bigwedge_{\mathclap{i \in K}} \avariable_i \neq \avariableter_i
     \land
     \bigwedge_{\mathclap{i \in \interval{1}{n}\setminus K}} \avariable_i = \avariableter_i
     \land
     \\
\inpathindex{\avariable_i}{\avariableter_i}{\avariablesept_i}{i \in K}
        \big(\bigwedge_{\mathclap{i \in K}}
          \avariablesept_i = \ameetvar{\avariable_i}{\avariablebis_i}{\avariableter_i}
          \land
(\strict{
            \bigwedge_{
            \mathclap{\substack{
            \asetpath'' = \asetpath'[\ameetvar{\avariable_i}{\avariablebis_i}{\avariableter_i} \gets \avariablesept_i \mid i \in K][\ameetvar{\avariable_i}{\avariablebis_i}{\avariableter_i} \gets \avariable_i \mid i \in \interval{1}{n}\setminus K]\\
            \pair{\avariable}{\avariablebis} \in \asetpath''
            }}} \sees{\avariable}{\avariablebis}{\emptyset}
          } \separate \size \geq \inbound)
        \big)
     \Big),
\end{aligned}$
\end{nscenter}
where $\inpathindex{\avariable_i}{\avariableter_i}{\avariablesept_i}{i \in K}$ is a shortcut for a sequence of
$\card{K}$ applications of $\inpath{\avariable_i}{\avariableter_i}{\avariablesept_i}$ and every $\avariablesept_i$ are program variables not appearing in $\asetpath$.
}

%

\subsection{Axiomatisation of the logic of core formulae}\label{subsection:axiom2CoreFormulae}
As done in Section~\ref{section:PSL}, to axiomatise \intervalSL, we start by extending the axiom system for the propositional calculus in order to obtain the proof system $\coresys$
dedicated to Boolean combinations of core formulae.
The axioms, presented in
System~\ref{axioms2:coreaxioms}, are divided into axioms for equalities between terms, whose name is of the form $=^\corepedix_i$; axioms essentially about the  predicates $\mathtt{sees}$, whose name is of the form $\mathtt{s}^\corepedix_i$;
and axioms essentially about the predicates $\mathtt{rem}$, whose name is of the form $\mathtt{r}^\corepedix_i$.
In order to obtain this axiom system, the two main difficulties (which lead to very technical formulae) are given by the distribution of meet-points within the memory state and the axiomatisation of the predicates $\mathtt{sees}$.
For the former, it is important to distinguish between symmetric and asymmetric meet-points.
For this reason, System~\ref{axioms2:coreaxioms} uses the formulae $\defined{\ameetvar{\avariable}{\avariablebis}{\avariableter}} \egdef \ameetvar{\avariable}{\avariablebis}{\avariableter}{=} \ameetvar{\avariable}{\avariablebis}{\avariableter}$, which checks if a meet-point is defined,
$\symmetric{\ameetvar{\avariable}{\avariablebis}{\avariableter}} \egdef \ameetvar{\avariable}{\avariablebis}{\avariableter}{=}\ameetvar{\avariablebis}{\avariable}{\avariableter}$ for symmetric meet-points, and
$\asymmetric{\ameetvar{\avariable}{\avariablebis}{\avariableter}} \egdef \defined{\ameetvar{\avariable}{\avariablebis}{\avariableter}} \land \lnot \symmetric{\ameetvar{\avariable}{\avariablebis}{\avariableter}}$ for asymmetric ones.
The definition of these formulae, as well as the ones below, is extended on a variable $\avariable \in \PVAR$ simply by replacement with the meet-point
$\ameetvar{\avariable}{\avariable}{\avariable}$ (the two terms are always equivalent, see the axiom~\ref{core2Ax:Self}).
So, for example $\defined{\avariable}$ is defined as $\defined{\ameetvar{\avariable}{\avariable}{\avariable}}$.
For $\mathtt{sees}$ predicates, an important distinction is given by terms corresponding to different locations in the same tree (no cycle is involved) and terms that correspond to different locations in the same cycle. Hence,
we define the abbreviations $\before{\aterm_1}{\aterm_2}$ and $\sameloop{\aterm_1}{\aterm_2}$ with the following meanings:\\
{
\setlength{\tabcolsep}{3pt}
\begin{tabular}{lcl}
$\pair{\astore}{\aheap} \models \before{\aterm_1}{\aterm_2}$ & iff &
$\semantics{\aterm_1}_{\astore,\aheap} {\neq} \semantics{\aterm_2}_{\astore,\aheap}$ and,
there is a path from $\semantics{\aterm_1}_{\astore,\aheap}$ to $\semantics{\aterm_2}_{\astore,\aheap}$ s.t.\\
&&the only location on the path that may belong to a cycle is $\semantics{\aterm_2}_{\astore,\aheap}$.\\
$\pair{\astore}{\aheap} \models \sameloop{\aterm_1}{\aterm_2}$ & iff &
$\semantics{\aterm_1}_{\astore,\aheap} \neq \semantics{\aterm_2}_{\astore,\aheap}$
and
there is a cycle  with both
$\semantics{\aterm_1}_{\astore,\aheap}$ and $\semantics{\aterm_2}_{\astore,\aheap}$.
\end{tabular}}\\
They are defined as follows for meet-points (and extended for $\avariable \in \PVAR$ as shown for $\defined{\avariable}$)
\begin{itemize}[nosep,after=\vspace{1pt}]
\item The formulae $\before{\ameetvar{\avariable}{\avariablebis}{\avariableter}}{\ameetvar{\avariable}{\avariablefour}{\avariablefifth}}$ and $\before{\ameetvar{\avariablebis}{\avariable}{\avariableter}}{\ameetvar{\avariable}{\avariablefour}{\avariablefifth}}$ are both defined as\\
\hfill
$\symmetric{\ameetvar{\avariable}{\avariablebis}{\avariableter}} \land \defined{\ameetvar{\avariable}{\avariablebis}{\avariablefifth}}
  \land \defined{\ameetvar{\avariable}{\avariablefour}{\avariablefifth}}
  \land \ameetvar{\avariable}{\avariablebis}{\avariableter} \neq \ameetvar{\avariable}{\avariablefour}{\avariablefifth}
  \land \ameetvar{\avariable}{\avariablebis}{\avariableter} \neq \ameetvar{\avariablebis}{\avariablefour}{\avariablefifth};
$\hfill\,
\item $
\before{\ameetvar{\avariable}{\avariablebis}{\avariableter}}{\ameetvar{\avariablefour}{\avariablefifth}{\avariablesix}} \egdef
  \bigvee_{a \in \{\avariablefour,\avariablefifth\}} \before{\ameetvar{\avariable}{\avariablebis}{\avariableter}}{\ameetvar{\avariable}{a}{\avariablesix}} \land
    \ameetvar{\avariable}{a}{\avariablesix} = \ameetvar{\avariablefour}{\avariablefifth}{\avariablesix};
$
\item $\sameloop{\ameetvar{\avariable}{\avariablebis}{\avariableter}}{\ameetvar{\avariablefour}{\avariablefifth}{\avariablesix}} \egdef
\ameetvar{\avariable}{\avariablebis}{\avariableter} = \ameetvar{\avariable}{\avariablefour}{\avariablesix} \land \ameetvar{\avariablefour}{\avariablefifth}{\avariablesix} = \ameetvar{\avariablefour}{\avariable}{\avariableter} \land \asymmetric{\ameetvar{\avariable}{\avariablefour}{\avariableter}}$.
\end{itemize}
We write $\aterm \in \asetmeetvar$ (finite set of terms $\asetmeetvar$) to denote $\bigvee_{\aterm_2 \in \asetmeetvar} \aterm {=} \aterm_2$.
Like the axiom~\ref{coreAx:EqSub}, the axiom~\ref{core2Ax:Substitute} performs a substitution of every occurrence of $\aterm_1$ with $\aterm_2$.
We have to be careful here:  when substituting a variable $\avariable$ with a meet-point $\ameetvar{\avariablebis}{\avariableter}{\avariablefour}$, we only substitute the occurrences of $\avariable$ that are not inside meet-point terms.
For example, $\sees{\avariable}{\ameetvar{\avariable}{\avariable}{\avariable}}{\{\avariable,\ameetvar{\avariable}{\avariable}{\avariable}\}}\completesubstitute{ \ameetvar{\avariablebis}{\avariableter}{\avariablefour}}{\avariable}$ is equal to
$\sees{\ameetvar{\avariablebis}{\avariableter}{\avariablefour}}{\ameetvar{\avariable}{\avariable}{\avariable}}{\{\ameetvar{\avariablebis}{\avariableter}{\avariablefour},\ameetvar{\avariable}{\avariable}{\avariable}\}}$.
By way of example, let us explain why all the instances of the  axiom~\ref{core2Ax:EqDef} are valid.
Suppose
 $\pair{\astore}{\aheap} \models  \defined{\ameetvar{\avariable}{\avariablebis}{\avariableter}} \land
  \defined{\ameetvar{\avariable}{\avariablebis}{\avariablefour}}$.
Since $\semantics{\ameetvar{\avariable}{\avariablebis}{\avariableter}}_{\astore,\aheap}$ is defined  (say equal to $\alocation$),
there are $\alength_1,\alength_2 \geq 0$ such that
\begin{itemize}[nosep]
\item $\aheap^{\alength_1}(\astore(\avariable)) = \aheap^{\alength_2}(\astore(\avariablebis)) = \alocation$ and there is $\alength \geq 0$ such that  $\aheap^\alength(\alocation) = \astore(\avariableter)$;
\item for every $\alength_1' \in \interval{0}{\alength_1-1}$ and $\alength_2' \geq 0$,
$\aheap^{\alength_1'}(\astore(\avariable)) \neq \aheap^{\alength_2'}(\astore(\avariablebis))$.
\end{itemize}
Similarly, as $\semantics{\ameetvar{\avariable}{\avariablebis}{\avariablefour}}_{\astore,\aheap}$ is also
defined  (say equal to $\alocation'$) ,
there are also $\gamma_1,\gamma_2 \geq 0$ such that
\begin{itemize}[nosep]
\item $\aheap^{\gamma_1}(\astore(\avariable)) = \aheap^{\gamma_2}(\astore(\avariablebis)) = \alocation'$ and there is $\alength' \geq 0$ such that  $\aheap^{\alength'}(\alocation') = \astore(\avariablefour)$;
\item for every $\gamma_1' \in \interval{0}{\gamma_1-1}$ and $\gamma_2' \geq 0$,
$\aheap^{\gamma_1'}(\astore(\avariable)) \neq \aheap^{\gamma_2'}(\astore(\avariablebis))$.
\end{itemize}
Combining the two types of inequality constraints, we can conclude that $\alength_1 = \gamma_1$ and therefore $\alocation = \alocation'$, i.e.
$\pair{\astore}{\aheap} \models \ameetvar{\avariable}{\avariablebis}{\avariableter} = \ameetvar{\avariable}{\avariablebis}{\avariablefour}$.
Soundness of  $\coresys$ is certainly not immediate but this can be done similarly to the above developments for
the axiom~\ref{core2Ax:EqDef}.

\begin{restatable}{lemma}{lemmaaxiomstwocoresound}\label{lemma:axiomstwocoresound}
$\coresys$ is sound.
\end{restatable}
\begin{figure*}
\vspace{-0.5cm}
\begin{axiombox}[label=axioms2:coreaxioms]{$\coresys$: Axioms for Boolean combinations of core formulae}
\begin{footnotesize}
\begin{enumerate}[align=left,leftmargin=*,labelsep=2pt]
\begin{minipage}{0.375\linewidth}
\item[\axlab{=^\corepedix}{core2Ax:Self}] $\avariable = \ameetvar{\avariable}{\avariable}{\avariable}$
\item[\axlab{=^\corepedix}{core2Ax:EqSymm}] $\aterm_1 = \aterm_2 \Rightarrow \aterm_2 = \aterm_1$
\item[\axlab{=^\corepedix}{core2Ax:Substitute}] $\aformula \land \aterm_1 = \aterm_2 \Rightarrow \aformula\completesubstitute{\aterm_2}{\aterm_1}$
\item[\axlab{=^\corepedix}{core2Ax:Self2}] $\defined{\ameetvar{\avariable}{\avariablebis}{\avariable}} \Rightarrow \avariable = \ameetvar{\avariable}{\avariablebis}{\avariable}$
  \item[\axlab{=^\corepedix}{core2Ax:Bothdef}] $\defined{\ameetvar{\avariable}{\avariablebis}{\avariableter}} \Rightarrow \defined{\ameetvar{\avariablebis}{\avariable}{\avariableter}}$
\end{minipage}
\begin{minipage}{0.65\linewidth}
  \item[\axlab{=^\corepedix}{core2Ax:EqDef}] $\defined{\ameetvar{\avariable}{\avariablebis}{\avariableter}} \land
  \defined{\ameetvar{\avariable}{\avariablebis}{\avariablefour}} \Rightarrow \ameetvar{\avariable}{\avariablebis}{\avariableter} = \ameetvar{\avariable}{\avariablebis}{\avariablefour}$
  \item[\axlab{=^\corepedix}{core2Ax:PedInv}] $\ameetvar{\avariable}{\avariablebis}{\avariableter} = \ameetvar{\avariablefour}{\avariablefifth}{\avariablesix} \Rightarrow \defined{\ameetvar{\avariable}{\avariablebis}{\avariablesix}}$
  \item[\axlab{=^\corepedix}{core2Ax:Pedix}] $\defined{\ameetvar{\avariable}{\avariablebis}{\avariableter}} \land \defined{\ameetvar{\avariableter}{\avariableter}{\avariablefifth}} \Rightarrow \defined{\ameetvar{\avariable}{\avariablebis}{\avariablefifth}}$
  \item[\axlab{=^\corepedix}{core2Ax:PedLinear}] $\defined{\ameetvar{\avariable}{\avariablebis}{\avariableter}} \land \defined{\ameetvar{\avariable}{\avariablefour}{\avariablefifth}} \Rightarrow \defined{\ameetvar{\avariableter}{\avariableter}{\avariablefifth}} \lor \defined{\ameetvar{\avariablefifth}{\avariablefifth}{\avariableter}}$
  \item[\axlab{=^\corepedix}{core2Ax:PedDefines}] $\defined{\ameetvar{\avariable}{\avariablebis}{\avariableter}} \land \defined{\ameetvar{\avariablefour}{\avariablefifth}{\avariableter}} \Rightarrow \defined{\ameetvar{\avariable}{\avariablefour}{\avariableter}}$
  \end{minipage}

  \item[\axlab{=^\corepedix}{core2Ax:Before}] $\symmetric{\ameetvar{\avariable}{\avariablebis}{\avariableter}} \land \defined{\ameetvar{\avariable}{\avariablefour}{\avariableter}} \land
  \ameetvar{\avariable}{\avariablefour}{\avariableter} \neq \ameetvar{\avariablebis}{\avariablefour}{\avariableter}
   \Rightarrow (
   \ameetvar{\avariable}{\avariablebis}{\avariableter} = \ameetvar{\avariable}{\avariablefour}{\avariableter}
   \lor
   \ameetvar{\avariable}{\avariablebis}{\avariableter} = \ameetvar{\avariablebis}{\avariablefour}{\avariableter}
   )$
  \item[\axlab{=^\corepedix}{core2Ax:Symmetric}] $\ameetvar{\avariable}{\avariablebis}{\avariableter} = \ameetvar{\avariablefour}{\avariablefifth}{\avariableter} \Rightarrow \symmetric{\ameetvar{\avariable}{\avariablefour}{\avariableter}} \land
  (
  \ameetvar{\avariable}{\avariablebis}{\avariableter} = \ameetvar{\avariable}{\avariablefour}{\avariableter} \lor
  \ameetvar{\avariable}{\avariablebis}{\avariableter} = \ameetvar{\avariable}{\avariablefifth}{\avariableter}
  )$
  \item[\axlab{=^\corepedix}{core2Ax:OneAsym}] $\symmetric{\ameetvar{\avariable}{\avariablebis}{\avariableter}} \land \asymmetric{\ameetvar{\avariable}{\avariablefour}{\avariablefifth}}
  \Rightarrow \ameetvar{\avariablebis}{\avariablefour}{\avariablefifth} = \ameetvar{\avariable}{\avariablefour}{\avariablefifth} \land \ameetvar{\avariablefour}{\avariablebis}{\avariablefifth} = \ameetvar{\avariablefour}{\avariable}{\avariablefifth}$
  \item[\axlab{=^\corepedix}{core2Ax:BothAsym}] $\asymmetric{\ameetvar{\avariable}{\avariablebis}{\avariableter}} \land \asymmetric{\ameetvar{\avariable}{\avariablefour}{\avariablefifth}} \Rightarrow \ameetvar{\avariable}{\avariablebis}{\avariableter} = \ameetvar{\avariable}{\avariablefour}{\avariablefifth}$
\end{enumerate}

\hrule
\begin{enumerate}[align=left,leftmargin=*,labelsep=2pt]
\addtocounter{enumi}{14}
\item[\axlab{\mathtt{s}^\corepedix}{core2Ax:SeesTermEq}]
$\aterm{=}\aterm' \land \sees{\aterm_1}{\aterm_2}{\{\aterm\}\cup\asetmeetvar} \Rightarrow \sees{\aterm_1}{\aterm_2}{\{\aterm,\aterm'\}\cup\asetmeetvar}$
\hfill
\axlab{\mathtt{s}^\corepedix}{core2Ax:SeesMono1} $\seesgeq{\aterm_1}{\aterm_2}{\asetmeetvar\cup\{\aterm\}}{\inbound} \Rightarrow \seesgeq{\aterm_1}{\aterm_2}{\asetmeetvar}{\inbound}$
\item[\axlab{\mathtt{s}^\corepedix}{core2Ax:SeesRef}]
$\seesgeq{\aterm_1}{\aterm_2}{\asetmeetvar}{\inbound} \Rightarrow \seesgeq{\aterm_1}{\aterm_2}{\asetmeetvar\cup\{\aterm_1,\aterm_2\}}{\inbound}$
\hfill
\axlab{\mathtt{s}^\corepedix}{core2Ax:SeesMono2}
$\seesgeq{\aterm_1}{\aterm_2}{\asetmeetvar}{\inbound{+}2} \Rightarrow\! \seesgeq{\aterm_1}{\aterm_2}{\asetmeetvar}{\inbound{+}1}$
\item[\axlab{\mathtt{s}^\corepedix}{core2Ax:SeesFunc}]
$\sees{\aterm_1}{\aterm_2}{\{\aterm_3\}} \land \sees{\aterm_1}{\aterm_3}{\{\aterm_2\}} \Rightarrow \aterm_2 = \aterm_3$
\hfill
\axlab{\mathtt{s}^\corepedix}{core2Ax:SeesDef}
$
\sees{\aterm_1}{\aterm_2}{\asetmeetvar} \implies \defined{\aterm_1} \land \defined{\aterm_2}
$
\item[\axlab{\mathtt{s}^\corepedix}{core2Ax:SeesEWCycl}]
$\sees{\aterm_1}{\aterm_1}{\emptyset} {\land} \lnot \sees{\aterm_1}{\aterm_1}{\{\aterm_2\}}  \iff \sameloop{\aterm_1}{\aterm_2}$
\hfill
\axlab{\mathtt{s}^\corepedix}{core2Ax:SeesBefore}
$\before{\aterm_1}{\aterm_2} \Rightarrow \sees{\aterm_1}{\aterm_2}{\emptyset}$
\item[\axlab{\mathtt{s}^\corepedix}{core2Ax:SeesMax}]
$\seesgeq{\aterm_1}{\aterm_2}{\asetmeetvar}{\inbound} \land \seesgeq{\aterm_1}{\aterm_2}{\asetmeetvar'}{\inbound'} \Rightarrow \seesgeq{\aterm_1}{\aterm_2}{\asetmeetvar\cup\asetmeetvar'}{\max(\inbound,\inbound')}$
\item[\axlab{\mathtt{s}^\corepedix}{core2Ax:SeesSum}]
$\seesgeq{\aterm_1}{\aterm_2}{\asetmeetvar}{\inbound_1}\!\land \seesgeq{\aterm_2}{\aterm_3}{\asetmeetvar}{\inbound_2} \land \aterm_2 {\not\in}\asetmeetvar \land \aterm_3 {\in} \asetmeetvar
  \Rightarrow \seesgeq{\aterm_1}{\aterm_3}{\asetmeetvar}{\inbound_1{+}\inbound_2} \land \lnot \sees{\aterm_1}{\aterm_3}{\{\aterm_2\}}$
\item[\axlab{\mathtt{s}^\corepedix}{core2Ax:SeesNegSum}]
$\seesgeq{\aterm_1}{\aterm_3}{\asetmeetvar}{\inbound} \land \lnot \sees{\aterm_1}{\aterm_3}{\{\aterm_2\}} \Rightarrow
{\bigvee_{\mathrlap{
  \raisebox{-0.2cm}[0pt][0pt]{$\kern-1em\scriptstyle{\inbound_1+\inbound_2=\max(2,\beta)-2}$}}}
}
(\seesgeq{\aterm_1}{\aterm_2}{\asetmeetvar}{\inbound_1{+}1} \land \seesgeq{\aterm_2}{\aterm_3}{\asetmeetvar}{\inbound_2{+}1})$
\item[\axlab{\mathtt{s}^\corepedix}{core2Ax:SeesSubscript}]
$\sees{\ameetvar{\avariable}{\avariablebis}{\avariableter}}{\ameetvar{\avariablefour}{\avariablefifth}{\avariablesix}}{\asetmeetvar} \implies \defined{\ameetvar{\avariable}{\avariablebis}{\avariablesix}}$
\item[\axlab{\mathtt{s}^\corepedix}{core2Ax:SeesToLoop}]
$
\sees{\ameetvar{\avariable}{\avariablebis}{\avariableter}}{\ameetvar{\avariablefour}{\avariablefifth}{\avariablesix}}{\asetmeetvar} \land \asymmetric{\ameetvar{\avariable}{\avariablefour}{\avariablesix}}
\implies \ameetvar{\avariablefour}{\avariablefifth}{\avariablesix} = \ameetvar{\avariablefour}{\avariable}{\avariablesix}
$
\item[\axlab{\mathtt{s}^\corepedix}{core2Ax:SeesToBefore}]
$\sees{\ameetvar{\avariable}{\avariablebis}{\avariableter}}{\ameetvar{\avariablefour}{\avariablefifth}{\avariablesix}}{\asetmeetvar} \land \symmetric{\ameetvar{\avariable}{\avariablefour}{\avariablesix}}
\land \ameetvar{\avariable}{\avariablebis}{\avariableter} \neq \ameetvar{\avariablefour}{\avariablefifth}{\avariablesix}
\Rightarrow \before{\ameetvar{\avariable}{\avariablebis}{\avariableter}}{\ameetvar{\avariablefour}{\avariablefifth}{\avariablesix}}$
\item[\axlab{\mathtt{s}^\corepedix}{core2Ax:SeesElsewhere}]
$\before{\aterm_1}{\aterm_2} \land \lnot \sees{\aterm_1}{\aterm_2}{\{\aterm_3\}} \implies \lnot \sees{\aterm_2}{\aterm_3}{\emptyset} \land \before{\aterm_1}{\aterm_3}$
\item[\axlab{\mathtt{s}^\corepedix}{core2Ax:SeesLoopOrder}]
$\sameloop{\aterm_1}{\aterm_2} \land \sameloop{\aterm_2}{\aterm_3} \land \aterm_1\neq\aterm_3 \Rightarrow
(\sees{\aterm_1}{\aterm_3}{\{\aterm_2\}} \iff \lnot \sees{\aterm_3}{\aterm_1}{\{\aterm_2\}})
$
\end{enumerate}
\hrule
\begin{enumerate}[align=left,leftmargin=*,labelsep=2pt]
\addtocounter{enumi}{30}
\begin{minipage}{0.37\linewidth}
\item[\axlab{\mathtt{r}^\corepedix}{core2Ax:RemPos}] $\remgeq{\asetpath}{0}$
\item[\axlab{\mathtt{r}^\corepedix}{core2Ax:RemMono1}] $\remgeq{\asetpath}{\inbound+1} \implies \remgeq{\asetpath}{\inbound}$
\item[\axlab{\mathtt{r}^\corepedix}{core2Ax:RemMono2}] $\remgeq{\{\pair{\aterm_1}{\aterm_2}\}\cup\asetpath}{\inbound} \implies \remgeq{\asetpath}{\inbound}$
\end{minipage}%
\begin{minipage}{0.63\linewidth}
\item[\axlab{\mathtt{r}^\corepedix}{core2Ax:RemTEq1}] $\aterm_1 = \aterm_2 \land \remgeq{\{\pair{\aterm_1}{\aterm_3}\} \cup \asetpath}{\inbound} \implies \remgeq{\{\pair{\aterm_1}{\aterm_3},\pair{\aterm_2}{\aterm_3}\} \cup \asetpath}{\inbound}$
\item[\axlab{\mathtt{r}^\corepedix}{core2Ax:RemTEq2}] $\aterm_1 = \aterm_2 \land \remgeq{\{\pair{\aterm_3}{\aterm_1}\} \cup \asetpath}{\inbound} \implies \remgeq{\{\pair{\aterm_3}{\aterm_1},\pair{\aterm_3}{\aterm_2}\} \cup \asetpath}{\inbound}$
\item[\axlab{\mathtt{r}^\corepedix}{core2Ax:RemNotSees}] $\lnot \seesgeq{\aterm_1}{\aterm_2}{\emptyset}{\beta_2{+}1} \land \remgeq{\asetpath}{\inbound_1}
  \implies
    \remgeq{\asetpath \cup \{\pair{\aterm_1}{\aterm_2}\}}{\inbound_1{\dotminus}\inbound_2}$
\end{minipage}
\item[\axlab{\mathtt{r}^\corepedix}{core2Ax:RemBetween}]
  $
    \sees{\aterm_1}{\aterm_2}{\asetmeetvar} \land \lnot \sees{\aterm_1}{\aterm_2}{\{\aterm_3\}} \land
    \remgeq{\{\pair{\aterm_1}{\aterm_3},\pair{\aterm_3}{\aterm_2}\}\cup\asetpath}{\inbound}
    \implies \remgeq{\{\pair{\aterm_1}{\aterm_2}\}\cup\asetpath}{\inbound}
  $
\item[\axlab{\mathtt{r}^\corepedix}{core2Ax:RemBetweenTwo}]
  $
    \sees{\aterm_1}{\aterm_2}{\asetmeetvar} \land \lnot \sees{\aterm_1}{\aterm_2}{\{\aterm_3\}} \land
    \remgeq{\{\pair{\aterm_1}{\aterm_2}\}\cup\asetpath}{\inbound} \implies \remgeq{\{\pair{\aterm_1}{\aterm_2},\pair{\aterm_1}{\aterm_3},\pair{\aterm_3}{\aterm_2}\}\cup\asetpath}{\inbound}
  $

\item[\axlab{\mathtt{r}^\corepedix}{core2Ax:RemShort}]$\!\!
\begin{aligned}[t]
&\big(\seesgeq{\aterm_1}{\aterm_2}{\emptyset}{\beta_2}
    \land
      \textstyle\bigwedge_{
      \scriptstyle{\pair{\aterm_3}{\aterm_4} \in \asetpath}}
      (\sees{\aterm_3}{\aterm_4}{\emptyset} \implies \sees{\aterm_1}{\aterm_2}{\{\aterm_3,\aterm_4\}} \land \sees{\aterm_3}{\aterm_4}{\{\aterm_1,\aterm_2\}} \land \aterm_3 {\neq} \aterm_1 )\\
& \land \remgeq{\asetpath \cup \{\pair{\aterm_1}{\aterm_2}\}}{\inbound_1} \big) \implies
    \remgeq{\asetpath}{\inbound_1{+}\inbound_2}
\end{aligned}$
\end{enumerate}
\end{footnotesize}
\end{axiombox}
\vspace{-5pt}
\begin{axiombox}[label=axioms3:Exists]{Axioms and inference rule for the guarded quantification $\weirdexists$}
\begin{footnotesize}
\begin{enumerate}[align=left,leftmargin=*,labelsep=2pt]
\setcounter{enumi}{39}
\begin{minipage}{0.55\linewidth}
\item[\axlab{\exists}{existsAx:Renaming}] $\inpath{\avariable}{\avariablebis}{\avariableter} \aformula \implies \inpath{\avariable}{\avariablebis}{\avariablefour} (\aformula\completesubstitute{\avariablefour}{\avariableter}) \assuming{\avariablefour \not\in \chars{\aformula}}$
\item[\axlab{\exists}{existsAx:Conjunction}] $\inpath{\avariable}{\avariablebis}{\avariableter} (\aformula \land \aformulabis)
\iff (\inpath{\avariable}{\avariablebis}{\avariableter} \aformula) \land \aformulabis \assuming{\avariableter\not\in\chars{\aformulabis}}$
\item[\axlab{\exists}{existsAx:Disjunction}] $\inpath{\avariable}{\avariablebis}{\avariableter} (\aformula_1 \lor \aformula_2) \iff (\inpath{\avariable}{\avariablebis}{\avariableter} \aformula_1) \lor (\inpath{\avariable}{\avariablebis}{\avariableter} \aformula_2)$
\end{minipage}%
\begin{minipage}{0.45\linewidth}
\hfill
\begin{itemize*}[align=left,leftmargin=*]
\item[\rulelab{\textbf{$\exists$-Intro}}{rule:wierdexistsinference}]
 $\inference{\aformula \implies \aformulabis}{\inpath{\avariable}{\avariablebis}{\avariableter} \aformula \implies \inpath{\avariable}{\avariablebis}{\avariableter} \aformulabis}{}
$
\end{itemize*}
\end{minipage}
\item[\axlab{\exists}{existsAx:ZeroStep}] $\sees{\avariable}{\avariablebis}{\emptyset} \implies \inpath{\avariable}{\avariablebis}{\avariableter}\ \avariableter = \avariable \assuming{\avariableter\not\in\{\avariable,\avariablebis\}}$
\hfill\axlab{\exists}{existsAx:Bottom}  $\lnot\inpath{\avariable}{\avariablebis}{\avariableter} \false$
\item[\axlab{\exists}{existsAx:AtLeastOneStep}] $\sees{\avariable}{\avariablebis}{\emptyset} \land \sees{\avariable}{\aterm_1}{\{\avariablebis\}}
\implies
\inpath{\avariable}{\avariablebis}{\avariableter}\ \avariableter = \aterm_1 \assuming{\avariableter\not\in\chars{\{\avariable,\avariablebis,\aterm_1\}}}$
\item[\axlab{\exists}{existsAx:Split}]
$
\begin{aligned}[t]
&({\avariable} {=} {\aterm_1} \lor \sees{\avariable}{\aterm_1}{\asetmeetvar'} ) \land \seesgeq{\aterm_1}{\aterm_2}{\asetmeetvar}{\inbound_1{+}\inbound_2} \land
({\aterm_2}{=}{\avariablebis} \lor \sees{\aterm_2}{\avariablebis}{\asetmeetvar''})\land(\avariablebis {=} \aterm_1 \implies \avariable{=} \avariablebis)\\
&\implies
\inpath{\avariable}{\avariablebis}{\avariableter} (\sees{\aterm_1}{\avariableter}{\asetmeetvar} \sim_1 \inbound_1
\land
\sees{\avariableter}{\aterm_2}{\asetmeetvar} \sim_2 \inbound_2 \land \avariableter\not\in \{\aterm_1, \aterm_2\})\\
&\blacktriangleleft\kern-0.45pc\left[\{\avariable,\avariablebis,\aterm_1,\aterm_2\} \subseteq \asetmeetvar,\asetmeetvar',\asetmeetvar'',\ \avariableter\not\in\chars{\asetmeetvar},\ \inbound_1,\inbound_2 \in \Nat^+\!,\ \geq\ \in \{\sim_1,\sim_2\} \subseteq \{\geq,=\}\right]
\end{aligned}$
\item[\axlab{\exists}{existsAx:SeesSem}]
$\lnot \inpath{\avariable}{\avariablebis}{\avariableter}((\avariable \neq \avariableter \land \avariablebis \neq \avariableter \land \sees{\avariable}{\avariablebis}{\{\avariable,\avariableter,\avariablebis\}})
\lor \lnot \sees{\avariable}{\avariablebis}{\emptyset}) \assuming{\avariableter \not\in \set{\avariable,\avariablebis}}$
\end{enumerate}
\margindown
\end{footnotesize}
\end{axiombox}
\vspace{-5pt}
\begin{axiombox}[label=axioms3:Star]{Axioms and inference rule for the separating conjunction}
\begin{footnotesize}
\begin{enumerate}[align=left,leftmargin=*,labelsep=2pt]
\addtocounter{enumi}{47}
\begin{minipage}{0.62\linewidth}
\item[\axlab{\separate}{starAx2:StarElim}]
$
\begin{aligned}[t]
&\charsymbform(\asms_1) \separate \charsymbform(\asms_2) \iff {\textstyle\bigvee_{\asms\ \text{s.t.}\
\symbunion{\asms_1}{\asms_2}{\asms}}}
{\charsymbform(\asms)}\\
&\assuming{\asms_1, \asms_2\ \text{resp. over}\
\pair{\asetvar}{\bound_1}\ \text{and}\ \pair{\asetvar}{\bound_2}}
\end{aligned}
$
\end{minipage}%
\begin{minipage}{0.38\linewidth}
\begin{itemize}[align=left,leftmargin=*]
\item[\rulelab{\textbf{$\separate$-Intro}}{rule:star2inference}]
$\inference{\aformula \implies \aformulater}{\aformula \separate \aformulabis \implies \aformulater \separate \aformulabis }{}$
\end{itemize}
\end{minipage}
\item[\axlab{\separate}{starAx2:DistrOr}] $(\aformula \lor \aformulabis) \separate \aformulater \implies (\aformula \separate \aformulater) \lor (\aformulabis \separate \aformulater)$
\hfill \axlab{\separate}{starAx2:StarCommute} $(\aformula \separate \aformulabis) \iff (\aformulabis \separate \aformula)$
\hfill \axlab{\separate}{starAx2:False} $(\bot \separate \aformula) \iff \bot$
\end{enumerate}
\margindown
\end{footnotesize}
\end{axiombox}
\end{figure*}

As done in Section~\ref{section:PSL}, in order to establish that $\coresys$ is complete, we first show its completeness with respect to \emph{core types}, where $\coretype{\asetvar}{\bound}$ is here defined as the set of formulae
$
\scaledformulasubset{
  \aformula \in \conjcomb{\coreformulae{\asetvar}{\bound}}
}{
  \bmat[
    \forall\aformulabis{\in} \coreformulae{\asetvar}{\bound}, \orliterals{\aformulabis}{\lnot \aformulabis} \inside \aformula, \text{ and } (\aformulabis \land \lnot \aformulabis) \not\inside \aformula
  ]
}{0.9}{1}$.

\begin{restatable}{lemma}{lemmaaxiomtwoRCct}\label{lemma:axiomtwoRCct}
Let $\aformula \in \coretype{\asetvar}{\bound}$.
We have $\lnot\aformula$ is valid iff $\vdash_{\coresys} \neg \aformula$. If $\vdash_{\coresys} \neg \aformula$ is provable then it has a proof
where all derivation steps only have formulae from $\boolcomb{\coreformulae{\asetvar}{\bound}}$.
\end{restatable}
%
%
%
%
%
\LongVersionOnly{
Indeed, the Core Type Lemma (Lemma~\ref{prop:corePSLone}) for \slSW still holds for the new notion of core formulae and core types (the proof is the same).
}
\noindent
Then, the proof of completeness of $\coresys$ follows with the same arguments used for Theorem~\ref{theo:corePSLcompl}.
\begin{restatable}{theorem}{theocoretwoPSLcompl}\label{theo:core2PSLcompl}
A Boolean combination of core formulae $\aformula$ is valid iff $\prove_{\coresys} \aformula$.
\end{restatable}


\subsection{Constructive elimination of $\exists{:}{\rightsquigarrow}$}\label{section-exists-elimination}
We write $\coresys(\weirdexists)$ to denote the system $\coresys$ augmented by the axioms and the inference
rule from System~\ref{axioms3:Exists}.
In System~\ref{axioms3:Exists},
given an arbitrary object $\mathtt{O}$ (this can be a term, a set of terms, a formula etc.), we write $\chars{\mathtt{O}}$ to denote the set of program
variables occurring in $\mathtt{O}$. For instance,
$\chars{\ameetvar{\avariable}{\avariablebis}{\avariableter}} = \{\avariable,\avariablebis,\avariableter\}$.
Axioms from~\ref{existsAx:Renaming} to~\ref{existsAx:Disjunction} and the introduction rule are classical tautologies of 
first-order quantification,
whereas the other axioms characterise the peculiar semantics of $\weirdexists$.
By way of example, let us explain why the axiom~\ref{existsAx:AtLeastOneStep}, equal to $\sees{\avariable}{\avariablebis}{\emptyset} 
\land \sees{\avariable}{\aterm_1}{\{\avariablebis\}} \implies
\inpath{\avariable}{\avariablebis}{\avariableter}\ \avariableter = \aterm_1$ ($\avariableter\not\in\chars{\{\avariable,\avariablebis,\aterm_1\}}$) 
is sound. Suppose $\pair{\astore}{\aheap} \models \sees{\avariable}{\avariablebis}{\emptyset} \land \sees{\avariable}{\aterm_1}{\{\avariablebis\}}$.
By the semantics of core formulae, we have 
$\emptyset \neq \minpath{\astore(\avariable)}{\semantics{\aterm_1}_{\astore,\aheap}}{\aheap} \subseteq 
\minpath{\astore(\avariable)}{\astore(\avariablebis)}{\aheap}$ and therefore
$\semantics{\aterm_1}_{\astore,\aheap}$ is defined.
Given $\avariableter \not \in\chars{\{\avariable,\avariablebis,\aterm_1\}}$, we have 
$\pair{\astore[\avariableter \gets \semantics{\aterm_1}_{\astore,\aheap}]}{\aheap} \models \avariableter = \aterm_1$.
This holds because  $\avariableter \not \in \chars{\aterm_1}$ as we want to guarantee  
$\semantics{\aterm_1}_{\astore,\aheap} = \semantics{\aterm_1}_{\astore[\avariableter \gets \semantics{\aterm_1}_{\astore,\aheap}],\aheap}$.
From $\emptyset \neq \minpath{\astore(\avariable)}{\semantics{\aterm_1}_{\astore,\aheap}}{\aheap} \subseteq \minpath{\astore(\avariable)}{\astore(\avariablebis)}{\aheap}$, 
we conclude that
 $\minpath{\astore(\avariable)}{\astore(\avariablebis)}{\aheap} \neq \emptyset$ and  $\semantics{\aterm_1}_{\astore,\aheap} \in 
\minpath{\astore(\avariable)}{\astore(\avariablebis)}{\aheap} \cup \{\astore(\avariablebis)\}$.
Therefore, $\pair{\astore}{\aheap} \models \inpath{\avariable}{\avariablebis}{\avariableter}\ \avariableter = \aterm_1$.
\cut{
By way of example, let us explain why the axiom~\ref{existsAx:SeesSem} is sound. By propositional reasoning,
the axiom~\ref{existsAx:SeesSem}
is equivalent to
$
\sees{\avariable}{\avariablebis}{\emptyset} \implies
\forallinpath{\avariable}{\avariablebis}{\avariableter}
(\avariable =  \avariableter \lor \avariablebis =  \avariableter \lor \neg \sees{\avariable}{\avariablebis}{\{\avariable,\avariableter,\avariablebis\}}),
$
where $\avariableter \not\in \set{\avariable,\avariablebis}$.
Assume that $\pair{\astore}{\aheap} \models \sees{\avariable}{\avariablebis}{\emptyset}$, which means that there is a path from $\astore(\avariable)$
to $\astore(\avariablebis)$ of length at least one. Pick some location $\alocation \in
\minpath{\astore(\avariable)}{\astore(\avariablebis)}{\aheap}$. Then, either $\alocation = \astore(\avariable)$, or
$\alocation = \astore(\avariablebis)$, or $\alocation$ is strictly between  $\astore(\avariable)$ and
$\astore(\avariablebis)$, and therefore
$\pair{\astore[\avariableter \leftarrow \alocation]}{\aheap} \models \neg
\sees{\avariable}{\avariablebis}{\{\avariable,\avariableter,\avariablebis\}}$.
Hence,
$
\pair{\astore}{\aheap} \models \forallinpath{\avariable}{\avariablebis}{\avariableter}
(\avariable =  \avariableter \lor \avariablebis =  \avariableter \lor \neg \sees{\avariable}{\avariablebis}{\{\avariable,\avariableter,\avariablebis\}})
$.
}
As done in Section~\ref{section:PSL} for 
$\separate$ and $\magicwand$, given a formula $\inpath{\avariable}{\avariablebis}{\avariableter} \aformula$, where $\aformula$ is in $\coretype{\asetvar}{\bound}$, we can show within $\coresys(\weirdexists)$ that there is a conjunction  $\aformulater$ from $\conjcomb{\coreformulae{\asetvar}{2\bound}}$ equivalent to it.
By the 
axiom~\ref{existsAx:Disjunction}, this applies when $\aformula$ is a
Boolean combination of core formulae.

\begin{restatable}{lemma}{lemmaaxiomseliminateexists}\label{lemma:axiomseliminateexists}
  Let $\aformula \in \boolcomb{\coreformulae{\asetvar\cup\{\avariableter\}}{\bound}}$ with $\avariableter \not\in \asetvar \supseteq \{\avariable,\avariablebis\}$. There is a Boolean combination of core formulae $\aformulater \in \boolcomb{\coreformulae{\asetvar}{2\bound}}$ such that $\prove_{\coresys(\weirdexists)} \inpath{\avariable}{\avariablebis}{\avariableter}\aformula \iff \aformulater$.
\end{restatable}

\subsection{Eliminating  $\separate$ with a big-step axiom}\label{section:axiom2starelim}
The proof system $\coresys(\separate,\weirdexists)$ for \intervalSL is defined as
$\coresys(\weirdexists)$ augmented by the axioms and the rule from System~\ref{axioms3:Star}.
Its main ingredient is given by the axiom~\ref{starAx2:StarElim} which, following the description in
Section~\ref{subsection:ingredients}, is clearly a big-step axiom. Indeed, as much as we would like to give a set of small-step axioms as we did for \slSW, we argue that producing such an axiomatisation for \intervalSL is
unfeasible.
In the proof system for \slSW, we found out that given two core types $\aformula$ and $\aformulabis$,
$\aformula \separate \aformulabis$ is equivalent to a conjunction of core formulae literals (see the proof
sketch of Lemma~\ref{lemma:starPSLelim}). Similar results hold for the separating implication $\magicwand$
(Lemma~\ref{lemma:magicwandPSLelim}) and the $\weirdexists$ quantifier. This property of being
equivalent to a simple conjunction of core formulae literals facilitates the design of
small-step axioms.
This is not the case for $\separate$ within \intervalSL: given two core types $\aformula$ and $\aformulabis$,
the formula $\aformula \separate \aformulabis$ is equivalent to a non-trivial disjunction of possibly exponentially many conjunctions.
Because of this, small-step axioms are hard to obtain and  some technical developments are needed in order to produce
an adequate
axiom system.
These developments are centered around the notions of symbolic memory states and characteristic formulae.
A \defstyle{symbolic memory state} is an abstraction on the memory state $\pair{\astore}{\aheap}$ that
is guided by the definition of core formulae, essentially highlighting the properties of
$\pair{\astore}{\aheap}$ that are expressible through these formulae, while removing the ones that are not expressible.
Given $\asetvar \subseteq_{\fin} \PVAR$ and $\bound \in \Nat^+$,
a \defstyle{symbolic memory states} $\asms$ over $\pair{\asetvar}{\bound}$ is
defined as a finite structure $\triple{\symbterms}{\amap}{\symbrem}$ such that
\begin{itemize}[nosep]
\item $\symbterms$ is a partition of a subset of $\atermset{\asetvar}$, encoding (dis)equalities. We introduce
the partial function $\equivclass{\ .\ }{\symbterms}: \atermset{\asetvar} \to \symbterms$ such that
given $\aterm \in \atermset{\asetvar}$ returns $\asymbterm \in \symbterms$ and $\aterm \in \asymbterm$, if it exists;
\item $\amap: \symbterms \to \symbterms\times\interval{1}{\bound}$ is a partial function encoding paths between terms and their length;
\item $\symbrem \in \interval{0}{\bound}$, encoding the number of memory cells (up to $\bound$)
not in paths between terms.
\end{itemize}
We denote with $\symbdomain{\asetvar}{\bound}$ the set of these structures.
The \defstyle{abstraction} $\symbms{\astore}{\aheap}{\asetvar}{\bound}$ of a memory state $\pair{\astore}{\aheap}$ is defined as the symbolic memory state $\triple{\symbterms}{\amap}{\symbrem}$ over $\pair{\asetvar}{\bound}$ such that
\begin{itemize}[nosep]
\item $\symbterms \egdef \{ \{ \aterm_1 \in \atermset{\asetvar} \mid \pair{\astore}{\aheap} \models \aterm_1 = \aterm_2 \} \mid \aterm_2 \in \atermset{\asetvar} \}$;
\item {
\addtolength{\jot}{-4pt}
$
\begin{aligned}[t]
\amap(\asymbterm) = (\asymbterm',\inbound) \, \equivdef \,
& \text{there are}\ \aterm_1 {\in} \asymbterm\ \text{and}\
\aterm_2 {\in} \asymbterm'\ \text{such that}\ \pair{\astore}{\aheap} \models \seesgeq{\aterm_1}{\aterm_2}{\atermset{\asetvar}}{\inbound}\ \text{and}\\
&\text{if}\ \inbound < \bound\ \text{then}\
\pair{\astore}{\aheap} \models \lnot \seesgeq{\aterm_1}{\aterm_2}{\atermset{\asetvar}}{\inbound{+}1};
\end{aligned}
$}
\item $\symbrem = \inbound$ $\equivdef$ $\pair{\astore}{\aheap} \models \remgeq{\atermset{\asetvar}\times\atermset{\asetvar}}{\inbound}$ and if $\inbound < \bound$ then
$\pair{\astore}{\aheap} \models \lnot\remgeq{\atermset{\asetvar}\times\atermset{\asetvar}}{\inbound{+}1}$.
\end{itemize}
Thus, a symbolic memory state $\triple{\symbterms}{\amap}{\symbrem}$ over $\pair{\asetvar}{\bound}$ simply stores the truth values for
 equalities, $\mathtt{sees}$ and $\mathtt{rem}$ predicates with respect to a memory state.
Its semantics
is best given through
the \defstyle{characteristic formula} $\charsymbform\triple{\symbterms}{\amap}{\symbrem}$ defined below (sets understood as conjunctions):
\begin{nscenter}
$
\begin{aligned}[t]
&
\formulasubset{\rem{\atermset{\asetvar}\times\atermset{\asetvar}}{\sim}{\symbrem}}{\bmat[\text{if } \symbrem {\neq} \bound \text{ then } (\sim \text{ is } =) \text{ else } (\sim \text{ is } \geq)]}
\land\formulasubset{\aterm_1 \neq \aterm_2}{\bmat[\equivclass{\aterm_1}{\symbterms}\ \text{or}\ \equivclass{\aterm_2}{\symbterms}\ \text{undefined, or}\ \equivclass{\aterm_1}{\symbterms}\neq\equivclass{\aterm_2}{\symbterms}]}
 \\
& {\land}\formulasubset{\aterm_1 = \aterm_2}{\bmat[\equivclass{\aterm_1}{\symbterms}=\equivclass{\aterm_2}{\symbterms} \ \text{defined}]}
  \land
  \formulasubset{\lnot\sees{\aterm_1}{\aterm_2}{\atermset{\asetvar}}}{\bmat[
  \equivclass{\aterm_1}{\symbterms}\ \text{undefined or}\
  \forall\inbound\in\interval{1}{\bound}: \amap(\equivclass{\aterm_1}{\symbterms}) \neq (\equivclass{\aterm_2}{\symbterms},\inbound)]}
  \\
&{\land}
\formulasubset{\sees{\aterm_1}{\aterm_2}{\atermset{\asetvar}}{=}{\inbound}}{\bmat[
\amap(\equivclass{\aterm_1}{\symbterms}) = (\equivclass{\aterm_2}{\symbterms},\inbound)\ \text{and}\ \inbound < \bound]}
\land
\formulasubset{\seesgeq{\aterm_1}{\aterm_2}{\atermset{\asetvar}}{\inbound}}{\bmat[
\amap(\equivclass{\aterm_1}{\symbterms}) = (\equivclass{\aterm_2}{\symbterms},\inbound)\ \text{and}\ \inbound = \bound]}
\end{aligned}
$
\end{nscenter}
From the definitions of $\charsymbform(\asms)$ and $\symbms{\astore}{\aheap}{\asetvar}{\bound}$, we can easily prove the following result.
\begin{restatable}{lemma}{lemmamsmodelsabs}\label{lemma:msmodelsabs}
For every $\pair{\astore}{\aheap}$ and every
$\asms \in \symbdomain{\asetvar}{\bound}$,
$\pair{\astore}{\aheap} \models \charsymbform(\asms)$
iff $\asms = \symbms{\astore}{\aheap}{\asetvar}{\bound}$.
\end{restatable}%
\noindent Thanks to this lemma, it is easy to see that every satisfiable characteristic formula $\charsymbform(\asms)$ of a symbolic memory state $\asms$ over
$\pair{\asetvar}{\bound}$ is equivalent to exactly one core type in $\coretype{\asetvar}{\bound}$.
Indeed, by definition of core types, the conjunction $\aformula \land \aformulabis$ of two core types $\aformula$ and $\aformulabis$ that are not syntactically equivalent up to associativity and commutativity of $\land$ is unsatisfiable.
Hence, by Lemma~\ref{lemma:msmodelsabs}, if a core type $\aformula \in \coretype{\asetvar}{\bound}$ is satisfied by a memory state $\pair{\astore}{\aheap}$, it must be equivalent to $\charsymbform(\symbms{\astore}{\aheap}{\asetvar}{\bound})$.
By Theorem~\ref{theo:core2PSLcompl} this equivalence is provable in $\coresys$.

The fundamental reason for taking symbolic memory states over memory states is that, given $\asetvar$ and $\bound$, there are finitely many symbolic memory states in $\symbdomain{\asetvar}{\bound}$. This leads to the definition of the
axiom~\ref{starAx2:StarElim}, which given two characteristic formulae $\aformula$ and $\aformulabis$ computes a finite disjunction of characteristic formulae that is equivalent to $\aformula \separate \aformulabis$. This disjunction is defined over a
new composition operator $\asymbunion$ on symbolic memory states that mimicks the
disjoint union $+$ on memory states.
More precisely, the following property
shall be satisfied.
\begin{nscenter}
For all $\pair{\astore}{\aheap}$ and all
$\asms_1,\asms_2$ resp. over $\pair{\asetvar}{\bound_1}$ and $\pair{\asetvar}{\bound_2}$,
$\symbunion{\asms_1}{\asms_2}{\symbms{\astore}{\aheap}{\asetvar}{\bound_1+\bound_2}}$
iff there are $\aheap_1$ and $\aheap_2$ such that
$\aheap_1 + \aheap_2 = \aheap$,
$\asms_1 = \symbms{\astore}{\aheap_1}{\asetvar}{\bound_1}$ and
$\asms_2 = \symbms{\astore}{\aheap_2}{\asetvar}{\bound_2}$,
\end{nscenter}
where $\asymbunion \subseteq \sum_{\asetvar,\bound_1,\bound_2} \symbdomain{\asetvar}{\bound_1}\times\symbdomain{\asetvar}{\bound_2}\times\symbdomain{\asetvar}{\bound_1+\bound_2}$,
and $\asms_1$, $\asms_2$  have satisfiable characteristic formulae.
Defining  $\asymbunion$ is clearly challenging.
Unlike the disjoint union of memory states,
$\asymbunion$ is not functional on its first two components. For instance, let
$\asms = (\{\avariable,\ameetvar{\avariable}{\avariable}{\avariable}\},\emptyset,1)$
and let us determine for which  $\asms'$, we have
$\symbunion{\asms}{\asms}{\asms'}$:
\begin{enumerate}[nosep]
\item As $\asms$ is the abstraction of the memory states $(\astore,\{\alocation_1\pto\alocation_2\})$ and $(\astore,\{\alocation_2\pto\alocation_1)\})$ where $\astore(\avariable) = \alocation_1 \neq \alocation_2$,
the abstraction of $(\astore,\{\alocation_1\pto\alocation_2,\alocation_2\pto\alocation_1\})$ must be a solution for $\asms'$.
More precisely, this abstraction is
$(\asymbterm,\{\asymbterm \pto (\asymbterm,2) \},0)$ where $\asymbterm = \{\avariable,\ameetvar{\avariable}{\avariable}{\avariable}\}$.
\item $\asms$ is however also the abstraction of $(\astore,\{\alocation_1\pto\alocation_2\})$ and $(\astore,\{\alocation_3\pto\alocation_4)\})$ such that $\astore(\avariable) \not\in \{\alocation_1,\alocation_3\}$.
Then, the abstraction $\triple{\{\avariable,\ameetvar{\avariable}{\avariable}{\avariable}\}}{\emptyset}{2}$ must also be a solution for $\asms'$.
\end{enumerate}
The main challenge for defining $\asymbunion$ is the composition of the two  ``garbage'':
memory cells that are abstracted with $\symbrem_1$ and $\symbrem_2$ in $\symbms{\astore}{\aheap_1}{\asetvar_1}{\bound_1}$ and $\symbms{\astore}{\aheap_2}{\asetvar_2}{\bound_2}$ may generate new paths between program variables in $\aheap_1+\aheap_2$.
This possibility was depicted in the first case above.
The definition of $\asymbunion$ can be found in
\ifCSLProc \cite{Demri&Lozes&Mansutti19bis}
\else
Appendix~\ref{appendix-definition-symbolic-composition}
\fi
and is too long
to be presented herein.
Roughly speaking,  for $\triple{\triple{\symbterms}{\amap_1}{\symbrem_1}}{\triple{\symbterms}{\amap_2}{\symbrem_2}}{\triple{\symbterms}{\amap}{\symbrem}}$ being in $\asymbunion$,
one needs to witness
two graph homomorphisms from the graphs $\pair{\symbterms_1}{\amap_1}$ and $\pair{\symbterms_2}{\amap_2}$ to $\pair{\symbterms}{\amap}$,
together with the existence of a partition that guarantees that paths that do not belong to the homomorphisms can be generated using the
memory cells from the garbage (abstracted by $\symbrem_1$ and $\symbrem_2$).
\LongVersionOnly{
Before introducing the axioms of $\separate$ and complete the axiom system of \intervalSL, it is necessary to do a digression on symbolic memory states.
As we clarified in Section~\ref{subsection:axiom2CoreFormulae}, symbolic memory states link together memory states and core types.
Then, it is natural to ask ourselves how the composition heaps lifts to symbolic memory states.
More explicitly we want to characterise the relation $\asymbunion \subseteq \sum_{\asetvar,\bound_1,\bound_2} \symbdomain{\asetvar_1}{\bound_1}\times\symbdomain{\asetvar_2}{\bound_2}\times\symbdomain{\asetvar}{\bound_1+\bound_2}$ on symbolic memory states having satisfiable characteristic formulae, satisfying the following property:
\begin{description}
\item[] For every $\pair{\astore}{\aheap}$ and every
$\asms_1,\asms_2$ resp. over $\pair{\asetvar}{\bound_1}$ and $\pair{\asetvar}{\bound_2}$,
$\symbunion{\asms_1}{\asms_2}{\symbms{\astore}{\aheap}{\asetvar}{\bound_1+\bound_2}}$
iff there are $\aheap_1$ and $\aheap_2$ such that
$\aheap_1 + \aheap_2 = \aheap$,
$\asms_1 = \symbms{\astore}{\aheap_1}{\asetvar}{\bound_1}$ and
$\asms_2 = \symbms{\astore}{\aheap_2}{\asetvar}{\bound_2}$.
\end{description}
The definition of such a relation is challenging. First of, differently from the union of memory state,
$\asymbunion$ cannot be functional on its first two components. Indeed, let $\asms = (\{\avariable,\ameetvar{\avariable}{\avariable}{\avariable}\},\emptyset,1)$. Lets ask ourselves for which symbolic memory states $\asms'$ it should holds that $\symbunion{\asms}{\asms}{\asms'}$:
\begin{enumerate}
\item As $\asms$ is the abstraction of the memory states $(\astore,\{\alocation_1\pto\alocation_2\})$ and $(\astore,\{\alocation_2\pto\alocation_1)\})$ where $\astore(\avariable) = \alocation_1 \neq \alocation_2$,
the abstraction of $(\astore,\{\alocation_1\pto\alocation_2,\alocation_2\pto\alocation_1\})$ must be a solution for $\asms'$.
More precisely, this abstraction is
$(\asymbterm,\{\asymbterm \pto (\asymbterm,2) \},0)$ where $\asymbterm = \{\avariable,\ameetvar{\avariable}{\avariable}{\avariable}\}$.
\item However, $\asms$ is  also the abstraction of $(\astore,\{\alocation_1\pto\alocation_2\})$ and $(\astore,\{\alocation_3\pto\alocation_4)\})$ such that $\astore(\avariable) \not\in \{\alocation_1,\alocation_3\}$.
Then, the abstraction $\triple{\{\avariable,\ameetvar{\avariable}{\avariable}{\avariable}\}}{\emptyset}{2}$ must also be a solution for $\asms'$.
\end{enumerate}
It is clear what is the main challenge in defining $\asymbunion$: memory cells that are abstracted with $\symbrem_1$ and $\symbrem_2$ in $\symbms{\astore}{\aheap_1}{\asetvar_1}{\bound_1}$ and $\symbms{\astore}{\aheap_2}{\asetvar_2}{\bound_2}$ can generate new paths between program variables in $\aheap_1+\aheap_2$.
This possibility was depicted in the first case above.
\emph{The definition of $\asymbunion$ can be found in Appendix~\ref{appendix:CompositionMemoryState}}.
Informally, solving the inclusion problem of a triple of symbolic memory states $\triple{\triple{\symbterms_1}{\amap_1}{\symbrem_1}}{\triple{\symbterms_2}{\amap_2}{\symbrem_2}}{\triple{\symbterms}{\amap}{\symbrem}}$ in $\asymbunion$ can be reduced to the computation of
two graph homomorphisms from the graphs $\pair{\symbterms_1}{\amap_1}$ and $\pair{\symbterms_2}{\amap_2}$ to $\pair{\symbterms}{\amap}$, together with the computation of a partition problem that checks if paths that do not belongs to the homomorphisms can be generated using the additional memory cells represented with the quantities $\symbrem_1$ and $\symbrem_2$.
The definition of $\asymbunion$ lead us to a concise set of axioms (System~\ref{axioms3:Star}) that allows us to handle the $\separate$ conjunction and therefore obtain an axiomatisation of \intervalSL.
}

Together with the other axioms in System~\ref{axioms3:Star}, which essentially allows to rewrite every formula 
into a disjunction of $\aformula \separate \aformulabis$ where $\aformula$ and $\aformulabis$ are 
characteristic formulae, the axiom~\ref{starAx2:StarElim} 
allows us to 
eliminate 
$\separate$, as done in Lemma~\ref{lemma:starPSLelim} for \slSW.

\begin{restatable}{lemma}{lemmaaxiomtwoStarElimination}\label{lemma:axiomtwoStarElimination}
  Let $\aformula \in \boolcomb{\coreformulae{\asetvar}{\bound_1}}$ and $\aformulabis \in \boolcomb{\coreformulae{\asetvar}{\bound_2}}$. There is a Boolean combination of core formulae $\aformulater \in \boolcomb{\coreformulae{\asetvar}{\bound_1+\bound_2}}$ such that $\prove_{\coresys(\separate,\weirdexists)} \aformula \separate \aformulabis \iff \aformulater$.
\end{restatable}
\noindent The adequacy of $\coresys(\separate,\weirdexists)$ then stems from Theorem~\ref{theo:core2PSLcompl} and Lemmata~\ref{lemma:axiomseliminateexists} and~\ref{lemma:axiomtwoStarElimination}.

\begin{restatable}{theorem}{theoaxiomstwosoundcomplete}\label{theo:axioms2soundcomplete}
$\coresys(\separate,\weirdexists)$ is sound and complete for \intervalSL.
\end{restatable}

\subsection{A \pspace upper bound for checking \intervalSL satisfiability}\label{subsection:PSpace}
In this short section, we  explain why the satisfiability problem for \intervalSL is in \pspace.
The \defstyle{memory size} of a formula $\aformula$, written $\msize{\aformula}$, is defined inductively as:
\begin{itemize*}
\item[] $\msize{\avariable = \avariablebis} \egdef \msize{\emp} {=} 1$,
\item[] $\msize{\avariable \Ipto \avariablebis} \egdef 2$,
\item[] $\msize{\inpath{\avariable}{\avariablebis}{\avariableter}\aformula} \egdef 2\times\msize{\aformula}$,
\item[] $\msize{\lnot \aformulabis} \egdef \msize{\aformulabis}$,
\item[] $\msize{\aformulabis_1 \separate \aformulabis_2} \egdef \msize{\aformulabis_1} + \msize{\aformulabis_2}$ and
\item[] $\msize{\aformulabis_1 \land \aformulabis_2} \egdef \max(\msize{\aformulabis_1},\msize{\aformulabis_2})$.
\end{itemize*}
Given  $\aformula$ with tree height $\alength$, $\msize{\aformula} \leq 2^{\alength+1}$.
Intuitively, $\msize{\aformula}$ provides an upper bound on the path length between  terms
and on the size of the garbage on models for $\aformula$ (above $\msize{\aformula}$, $\aformula$
cannot see the difference).
As a consequence of the proofs for the elimination of the connectives $\weirdexists$ and $\separate$ in the calculus,
for each $\aformula$  in \intervalSL,  there is a Boolean combination  of core formulae
from $\coreformulae{\chars{\aformula}}{\msize{\aformula}}$  logically equivalent to $\aformula$.
\cut{
As a consequence of
Lemma~\ref{lemma:axiomseliminateexists} and Lemma~\ref{lemma:axiomtwoStarElimination} (Appendix~\ref{appendix-starelimination})
about elimination of $\weirdexists$ and $\separate$,
for each $\aformula$  in \intervalSL,  there is a Boolean combination  of core formulae
from $\coreformulae{\chars{\aformula}}{\msize{\aformula}}$  logically equivalent to $\aformula$.
}
\LongVersionOnly{
we get the result below.
\begin{corollary}
\label{corollary:small-core-formulae}
Let $\aformula$ be in \intervalSL. There is a Boolean combination  of core formulae
from $\coreformulae{\chars{\aformula}}{\msize{\aformula}}$ that is logically equivalent to $\aformula$.
\end{corollary}
}

\intervalSL may require small memory states whose heap has an exponential
amount of memory cells, as shown in Section~\ref{section-introduction-to-guarded-logic} with the formula $\mbox{\sc R}^n(\avariable,\avariablebis)$.
So, to establish a \pspace
bound, we cannot rely on an algorithm that guesses a polynomial-size memory state
and performs model-checking on it without further refinements.
\cut{
Assuming that $\forallinpath{\avariable}{\avariablebis}{\avariableter} \aformula \egdef
      \neg \forallinpath{\avariable}{\avariablebis}{\avariableter} \neg \aformula$ and
$\mbox{\sc R}^0(\avariable,\avariablebis) \egdef \avariable \neq \avariablebis \wedge
\inpath{\avariable}{\avariablebis}{\avariableter} \top$,
\LongVersionOnly{
(existence of a path of length at least one between $\avariable$ and $\avariablebis$),
}
let  $\mbox{\sc R}^{n+1}(\avariable,\avariablebis)$
be
\cut{
{\small $
 \avariable \neq \avariablebis \wedge
\inpath{\avariable}{\avariablebis}{\avariableter}   \
\forallinpath{\avariable}{\avariablebis}{\avariableter'} \ \forallinpath{\avariable}{\avariablebis}{\avariableter''}
\ \left((\avariableter' = \avariable \wedge \avariableter'' = \avariableter) \vee
      (\avariableter' = \avariableter \wedge \avariableter'' = \avariablebis)\right) \Rightarrow \mbox{\sc R}^n(\avariableter',\avariableter'')$}.
}
$$
 \avariable \neq \avariablebis \wedge
\inpath{\avariable}{\avariablebis}{\avariableter}   \
\forallinpath{\avariable}{\avariablebis}{\avariableter'} \ \forallinpath{\avariable}{\avariablebis}{\avariableter''}
\ \left((\avariableter' = \avariable \wedge \avariableter'' = \avariableter) \vee
      (\avariableter' = \avariableter \wedge \avariableter'' = \avariablebis)\right) \Rightarrow \mbox{\sc R}^n(\avariableter',\avariableter'').$$

Note that $\mbox{\sc R}^n(\avariable,\avariablebis)$ is of linear size in $n$,
$\msize{\aformula} \in 2^{\mathcal{O}(n)}$
 and $\mbox{\sc R}^n(\avariable,\avariablebis)$ enforces the existence of a path of length at least $2^n$ between
$\avariable$ and $\avariablebis$.
}
Nevertheless, polynomial-size symbolic memory states are able to abstract a garbage of exponential size or a path between
terms
of exponential length by encoding these quantities in binary, which leads to \pspace.

\begin{restatable}{theorem}{theorempspace}\label{theorem:pspace}
The satisfiability problem for \intervalSL is \pspace-complete.
\end{restatable}

\pspace-hardness is inherited from~\cite{Calcagno&Yang&OHearn01}.
To establish the \pspace upper bound, there is a nondeterministic
polynomial-space algorithm that guesses a satisfiable symbolic memory state over $\pair{\chars{\aformula}}{\msize{\aformula}}$
and that performs a symbolic model-checking on it against $\aformula$.
The symbolic approach works fine as separating conjunction and guarded quantification have symbolic counterparts, and
these symbolic operators can be decided in polynomial space.
\LongVersionOnly{
Hence, working on the symbolic memory states amounts to evaluate
the formulae on concrete memory states and the other way around. This is apart from the fact, that the \pspace upper bound
can be guaranteed by the use of small core formulae (see e.g. Corollary~\ref{corollary:small-core-formulae}).}

\section{Conclusion}\label{section:conclusion}
We  presented a method to axiomatise internally separation logics based
on the axiomatisation of Boolean combinations of core formulae
(based on the even more restricted fragment of core types).
We designed the first proof system for \slSW
that is completely internal and highlights the essential ingredients
of the heaplet semantics.
\LongVersionOnly{
 and therefore whose axioms and rules
involve schemas instantiated by formulae in \slSW.
}
To further illustrate our method,
we  provided an internal
Hilbert-style axiomatisation for the new  separation logic \intervalSL.
It contains the ``list quantifier'' $\inpath{\avariable}{\avariablebis}{\avariableter}$
that, we believe, is of interest for its own sake as it allows to quantify over elements of a list.
The completeness proof, following our general pattern,
still reveals to be very complex as not only we had to invent the adequate family of core formulae
but their axiomatisation was
challenging. As far as we know, this is the first Hilbert-style axiomatisation of a separation logic
having $\ls$ and a guarded form of quantification. Moreover, through a small model property derived from its proof system, we proved that \intervalSL has a \pspace-complete satisfiability problem. Finally, both  completeness proofs are constructive, see e.g.~\cite{Doumane17,Luck18}, as we provide means to transform
formulae into Boolean combinations of core formulae and then into disjunctions of core types.

Obviously, Hilbert-style proof systems for separation logics, as presented in the paper,
are of theoretical interest, at least to grasp what are the essential features of the logics. Still, it remains
to be seen whether applications are possible for designing decision procedures, for instance
to feed provers with appropriate axiom instances to accelerate the proof search.\\[2pt]
{\bf Acknowledgements:} We would like to thank the anonymous CSL'20 reviewers for their suggestions and remarks that help us to improve
the quality of this paper.
\LongVersionOnly{
For instance, we have
seen that a small model property has been derived from the Hilbert-style proof system for \intervalSL, which
allowed us to conclude a \pspace upper bound for  \intervalSL satisfiability.
It remains open to identify what more can be concluded for  mechanisation purposes.
}

\newpage

\appendix
 \newpage
\section{A derivation in the proof system for \slSW}
\label{appendix-paragraphe-an-example}
We develop the proof of $\emp \implies (\alloc{\avariable} \wedge \size = 1 \magicwand \size = 1)$ as a more complete example with respect to the one shown in the body of the paper.
We use the following theorem and rules, which can be shown admissible in the calculus:
\begin{nscenter}
  \scalebox{0.9}{
  \lemmalab{\textbf{($\magicwand\land$-DistrL)}}{axiom:magicsep}\,
  $
  (\aformula \magicwand \aformulabis) \land (\aformula \magicwand \aformulater) \implies (\aformula \magicwand \aformulabis \land \aformulater)
  $}%
  \hfill
  \scalebox{0.9}{
  \lemmalab{\textbf{($\land\true$IL)}}{axiom:andtrue}\,
  $
  \aformula \implies \true \land \aformula
  $}%
  \hfill
  \scalebox{0.9}{
  \rulelab{\textbf{$\land$-InfL}}{rule:andinf}
  $
  \inference{\aformula \implies \aformulater}{\aformula \land \aformulabis \implies \aformulater \land \aformulabis}
  $}%
\end{nscenter}
Recall that $\size = \inbound$ is a shortcut for $\size \geq \inbound \land \lnot \size \geq \inbound{+}1$.\\
\scalebox{0.9}{
$
\begin{nd}
\have {5} {\true \separate (\alloc{\avariable} \land \size = 1) \implies (\alloc{\avariable} \land \size = 1) \separate \true} \by{\ref{starAx:Commute}}{}
\have {1} {\alloc{\avariable} \land \size = 1 \implies \size \geq 1} \by{\ref{axiom:andelim}}{}
\have {2} {\alloc{\avariable} \land \size = 1 \separate \true \implies \size \geq 1 \separate \true} \by{\ref{rule:starinference} rule}{1}
\have {3} {\size \geq 1 \separate \true \implies \size \geq 1} \by{\ref{starAx:MonoCore} ($\size \geq 1 \egdef \lnot \emp$)}{}
\have {6} {\true \separate (\alloc{\avariable} \land \size = 1) \implies \size \geq 1} \by{\ref{rule:imptr} twice}{5,2,3}
\have {7} {\true \implies (\alloc{\avariable} \land \size = 1 \magicwand \size \geq 1)} \by{\ref{rule:staradj} rule}{6}
\have {12} {\emp \implies (\alloc{\avariable} \land \size = 1 \magicwand \lnot \size \geq 2)} \by{\small{derived in the body of the paper}}{}
\have {14} {(\alloc{\avariable} \land \size = 1 \magicwand \lnot \size \geq 2) \implies\\
 \true \land (\alloc{\avariable} \land \size = 1 \magicwand \lnot \size \geq 2)} \by{\ref{axiom:andtrue}}{}
\have {15} { \true \land (\alloc{\avariable} \land \size = 1 \magicwand \lnot \size \geq 2) \implies\\
\big((\alloc{\avariable} \land \size = 1 \magicwand \size \geq 1) \land\\ (\alloc{\avariable} \land \size = 1 \magicwand \lnot \size \geq 2)\big)} \by{\ref{rule:andinf}}{7}
\have {16} {\big((\alloc{\avariable} \land \size = 1 \magicwand \size \geq 1) \land\\ (\alloc{\avariable} \land \size = 1 \magicwand \lnot \size \geq 2)\big) \implies\\
(\alloc{\avariable} \land \size = 1 \magicwand \size = 1)
} \by{\ref{axiom:magicsep}}{}
\have {17} {(\alloc{\avariable} \land \size = 1 \magicwand \lnot \size \geq 2) \implies\\ (\alloc{\avariable} \land \size = 1 \magicwand \size = 1)} \by{\ref{rule:imptr} twice}{14,15,16}
\have {18} {\emp \implies
(\alloc{\avariable} \land \size = 1 \magicwand \size = 1)
} \by{\ref{rule:imptr}}{12,17}
\end{nd}
$}

\section{Definition of the symbolic composition $\asymbunion$}
\label{appendix-definition-symbolic-composition}
In this section we provide useful information on the symbolic composition $\asymbunion$ evoked in Section~\ref{section:axiom2starelim}.
\ifLongVersionWithAppendix
We refer the reader to this section for the definitions (central in this appendix) of symbolic memory state, abstraction
of a memory state and characteristic formula.
\fi
We denote with $\pi_i$ the projection map, so that
given a symbolic memory state $\triple{\symbterms}{\amap}{\symbrem}$ over $\pair{\asetvar}{\bound}$, the function
$\pi_1 \circ \amap$ goes from $\symbterms$ to $\symbterms$, whereas $\pi_2 \circ \amap$ goes from $\symbterms$ to  $\interval{1}{\bound}$
($\pi_i$ is the projection on the $i$th component).

We want to characterise the relation $\asymbunion \subseteq \sum_{\asetvar,\bound_1,\bound_2}
\symbdomain{\asetvar}{\bound_1}\times\symbdomain{\asetvar}{\bound_2}\times\symbdomain{\asetvar}{\bound_1+\bound_2}$, on symbolic memory states having satisfiable characteristic formulae, such that
\begin{description}
\item[\lemmalab{(+$_{\mathsf{sms}}$-adequacy)}{compositionproperty}] For all $\pair{\astore}{\aheap}$ and all
$\asms_1$ and $\asms_2$,
respectively over $\pair{\asetvar}{\bound_1}$ and $\pair{\asetvar}{\bound_2}$,
$\triple{\asms_1}{\asms_2}{\symbms{\astore}{\aheap}{\asetvar}{\bound_1+\bound_2}} \in \asymbunion$
iff
$\aheap_1 + \aheap_2 = \aheap$,
$\asms_1 = \symbms{\astore}{\aheap_1}{\asetvar}{\bound_1}$ and
$\asms_2 = \symbms{\astore}{\aheap_2}{\asetvar}{\bound_2}$ for some $\aheap_1$ and $\aheap_2$.
\end{description}
We write $\symbunion{\triple{\symbterms_1}{\amap_1}{\symbrem_1}}{\triple{\symbterms_2}{\amap_2}{\symbrem_2}}{\triple{\symbterms}{\amap}{\symbrem}}$ for
$\triple{\triple{\symbterms_1}{\amap_1}{\symbrem_1}}{\triple{\symbterms_2}{\amap_2}{\symbrem_2}}{\triple{\symbterms}{\amap}{\symbrem}} \in \asymbunion$.
By Lemma~\ref{lemma:msmodelsabs}, if $\asymbunion$ satisfies \ref{compositionproperty}, then it correctly behaves with respect 
to the separating conjunction $\separate$, as follows:
\begin{nscenter}
$\symbunion{\triple{\symbterms_1}{\amap_1}{\symbrem_1}}{\triple{\symbterms_2}{\amap_2}{\symbrem_2}}{\triple{\symbterms}{\amap}{\symbrem}}\quad
\text{iff}\quad
\models \charsymbform{\triple{\symbterms}{\amap}{\symbrem}} \implies \charsymbform{\triple{\symbterms_1}{\amap_1}{\symbrem_1}} \separate \charsymbform{\triple{\symbterms_2}{\amap_2}{\symbrem_2}}$
\end{nscenter}
As shown in Section~\ref{section:axiom2starelim}, $\asymbunion$ cannot be functional in its first two components.
Moreover, in Section~\ref{section:axiom2starelim} we described the membership problem $\triple{\triple{\symbterms}{\amap_1}{\symbrem_1}}{\triple{\symbterms}{\amap_2}{\symbrem_2}}{\triple{\symbterms}{\amap}{\symbrem}} \in \asymbunion$ to be roughly equivalent to the problem of finding
two graph homomorphisms from the graphs $\pair{\symbterms_1}{\amap_1}$ and $\pair{\symbterms_2}{\amap_2}$ to $\pair{\symbterms}{\amap}$,
together with the existence of a partition that guarantees that paths that do not belong to the homomorphisms can be generated using the
memory cells from the garbage (abstracted by $\symbrem_1$ and $\symbrem_2$) -- see the details below.
Now, we  define this problem and we give an informal explanation to its components.
\cut{
Notice that $\asymbunion$ cannot be functional on its first two components.
Take for example the symbolic memory state $\asms = (\{\avariable,\ameetvar{\avariable}{\avariable}{\avariable}\},\emptyset,1)$ and let us
ask ourselves for which symbolic memory states $\asms$,  it should hold that $\symbunion{\asms}{\asms}{\asms'}$
There are (at least) two possibilities.
\begin{enumerate}
\item As $\asms$ is the abstraction of the memory states $(\astore,\{\alocation_1\pto\alocation_2\})$ and $(\astore,\{\alocation_2\pto\alocation_1)\})$ where $\astore(\avariable) = \alocation_1 \neq \alocation_2$,
the abstraction of $(\astore,\{\alocation_1\pto\alocation_2,\alocation_2\pto\alocation_1\})$ must be a solution for $\asms'$.
More precisely, this abstraction is
$(\asymbterm,\{\asymbterm \pto (\asymbterm,2) \},0)$ where $\asymbterm = \{\avariable,\ameetvar{\avariable}{\avariable}{\avariable}\}$.
\item However, $\asms$ is  also the abstraction of $(\astore,\{\alocation_1\pto\alocation_2\})$ and $(\astore,\{\alocation_3\pto\alocation_4)\})$ such that $\astore(\avariable) \not\in \{\alocation_1,\alocation_3\}$.
Then, the abstraction $\triple{\{\avariable,\ameetvar{\avariable}{\avariable}{\avariable}\}}{\emptyset}{2}$ must also be a solution for $\asms'$.
\end{enumerate}
From the above example, it is clear what is the main challenge in defining
This possibility was depicted in the first case above.
}

Let $\asms_1 = {\triple{\symbterms_1}{\amap_1}{\symbrem_1}}$, $\asms_2 = {\triple{\symbterms_2}{\amap_2}{\symbrem_2}}$ and $\asms = {\triple{\symbterms}{\amap}{\symbrem}}$ be three symbolic memory states
respectively over $\pair{\asetvar}{\bound_1}$, $\pair{\asetvar}{\bound_2}$ and
$\pair{\asetvar}{\bound_1+\bound_2}$
We have $\symbunion{\asms_1}{\asms_2}{\asms}$ $\equivdef$ \\

{
\small
\noindent $\charsymbform(\asms_1)$, $\charsymbform(\asms_2)$ and $\charsymbform(\asms)$ are satisfiable
and
there are injections $\ainj_1:\symbterms_1 \to \symbterms$ and $\ainj_2:\symbterms_2 \to \symbterms$
satisfying the following 4 conditions:

\begin{enumerate}
\item[\lemmalab{1}{C-functionality}.] for each $\asymbterm \in \domain{\amap}$ then it cannot be that
$\ainj_1^{-1}(\asymbterm) \in \domain{\amap_1}$ and $\ainj_2^{-1}(\asymbterm) \in \domain{\amap_2}$;
\end{enumerate}

\begin{enumerate}
\item[\lemmalab{2}{C-meet-points}.] for all $i \in \{1,2\}$, for each $\asymbterm \in \symbterms_i$,
  \begin{enumerate}[nosep]
  \item for each $\avariable \in \asetvar$, \ $\avariable \in \asymbterm \iff \avariable \in \ainj_i(\asymbterm)$;
  \item for each meet-point term $\ameetvar{\avariable}{\avariablebis}{\avariableter} \in \asymbterm$,
  if $\ameetvar{\avariablebis}{\avariable}{\avariableter} \not\in \asymbterm$ then
  $\ameetvar{\avariable}{\avariablebis}{\avariableter} \in \ainj_i(\asymbterm)$, otherwise
  $\ameetvar{\avariable}{\avariablebis}{\avariableter} \in \ainj_i(\asymbterm)$ or $\ameetvar{\avariablebis}{\avariable}{\avariableter} \in \ainj_i(\asymbterm)$.
  \end{enumerate}
\end{enumerate}

\begin{enumerate}
\item[\lemmalab{3}{C-homomorphism}.] for all $i \in \{1,2\}$, for all $\aterm_1,\aterm_2 \in \atermset{\asetvar}$ and 
for all $\inbound \in \interval{1}{\bound_i}$, if $\amap_i(\equivclass{\aterm_1}{\symbterms_i}) = (\equivclass{\aterm_2}{\symbterms_i},\inbound)$ then
there is $n \in \Nat$ and a set $\{\asymbterm_0,\dots,\asymbterm_{n+1}\} \subseteq \symbterms$ such that
\begin{enumerate}[nosep]
\item $\asymbterm_0 {=} \ainj_i(\equivclass{\aterm_1}{\symbterms_i})$ and $\asymbterm_{n+1} {=} \ainj_i(\equivclass{\aterm_2}{\symbterms_i})$;
\item for each $j \in \interval{1}{n}$, $\asymbterm_j \not\in\range{\ainj_1} \cup \range{\ainj_2}$;
\item for each $j \in \interval{0}{n}$, $(\pi_1 \circ \amap)(\asymbterm_j) = \asymbterm_{j+1}$;
\item if $\inbound < \bound_i$ then $\sum_{j\in\interval{0}{n}} (\pi_2 \circ \amap)(\asymbterm_j) = \inbound$,
else $\sum_{j\in\interval{0}{n}} (\pi_2 \circ \amap)(\asymbterm_j) \geq \bound_i$.
\end{enumerate}
\end{enumerate}

\noindent Let $\symbsources{\asms_1}{\asms_2}{\asms}$, $\symbtargets{\asms_i}{\asms}$ ($i \in \{1,2\}$) and
$\symbinternal{\asms_1}{\asms_2}{\asms}$ be the sets defined below
\begin{itemize}[nosep]
\item $\symbsources{\asms_1}{\asms_2}{\asms}$ $\egdef$
$\{ \asymbterm \in \symbterms \mid$ $\asymbterm {\cap} \asetvar \neq \emptyset$ and for all $\avariable \in \asymbterm {\cap} \asetvar$ and $i {\in} \{1,2\}$,
$\equivclass{\avariable}{\symbterms_i} {\not \in} \domain{\amap_i}\}$.
\item $\symbtargets{\asms_i}{\asms}$ is the set of $\asymbterm \in \symbterms$\ such that $\asymbterm \cap \asetvar \neq \emptyset$ or there are
$\aterm_1,\aterm_2 \in \atermset{\asetvar}$, $\alength_1,\alength_2 \geq 0$ such that
\begin{itemize}[nosep]
\item $(\pi_1\circ\amap_i)(\equivclass{\aterm_1}{\symbterms_i}) = \equivclass{\aterm_2}{\symbterms_i}$;
\item $(\pi_1 \circ \amap)^{\alength_1}(\ainj_i(\equivclass{\aterm_1}{\symbterms_i})) = \asymbterm$ and
$(\pi_1 \circ \amap)^{\alength_2}(\asymbterm) = \ainj_i(\equivclass{\aterm_2}{\symbterms_i})$;
\item for every $\alength < \alength_1+\alength_2$, $(\pi_2 \circ \amap)^{\alength}(\equivclass{\aterm_1}{\symbterms_i}) \neq \ainj_i(\equivclass{\aterm_2}{\symbterms_i})$.
\end{itemize}
\item
$\symbinternal{\asms_1}{\asms_2}{\asms} \egdef \domain{\amap} \setminus
((\cup_{i\in\{1,2\}}\symbtargets{\asms}{\asms_i}) \setminus \symbsources{\asms_1}{\asms_2}{\asms})$.
\end{itemize}

\cut{
\begin{enumerate}
\item[\lemmalab{4}{C-stillgarbage}.] If $\symbrem_1 < \bound_1$ and $\symbrem_2 < \bound_2$ then
$R \egdef \symbrem +
\sum_{\asymbterm \in \symbinternal{\asms_1}{\asms_2}{\asms}} (\pi_2 \circ \amap)(\asymbterm)$ equals $\symbrem_1 + \symbrem_2$,
else $R \geq \symbrem_1 + \symbrem_2$.
\end{enumerate}
}

\begin{enumerate}
\item[\lemmalab{4}{C-wasgarbage}.]
there are two functions $\amapbis_1,\amapbis_2: \symbinternal{\asms_1}{\asms_2}{\asms} \to \interval{1}{\bound_1+\bound_2}$ and $\symbrem_1',\symbrem_2' \in \Nat$ such that
\begin{enumerate}
\item $\symbrem = \symbrem_1' + \symbrem_2'$ and for every $\asymbterm \in \symbinternal{\asms_1}{\asms_2}{\asms}$, $(\pi_2 \circ \amap)(\asymbterm) = \amapbis_1(\asymbterm) + \amapbis_2(\asymbterm)$;
\item for all $i \in \{1,2\}$,
 $\min(\symbrem_i' + \sum_{\asymbterm \in \symbinternal{\asms_1}{\asms_2}{\asms}} \amapbis_i(\asymbterm), \bound_i) = \min(\symbrem_i, \bound_i)$;
\item for every $\asymbterm \in \symbsources{\asms_1}{\asms_2}{\asms} \cap \domain{\amap}$ there is
$\{\asymbterm_0 {=} \asymbterm, \asymbterm_1,\dots,\asymbterm_n\} \subseteq \symbinternal{\asms_1}{\asms_2}{\asms}$ such that
\begin{itemize}
\item for every
$j \in \interval{0}{n-1}$, $(\pi_1 \circ \amap)(\asymbterm_j) = \asymbterm_{j+1}$;
\item $(\pi_1 \circ \amap)(\asymbterm_n) \in \bigcup_{i \in \{1,2\}} \symbtargets{\asms_i}{\asms}$;
\item given $i \in \{1,2\}$, if $(\pi_1 \circ \amap)(\asymbterm_n) \in \symbtargets{\asms_i}{\asms}$ then
$\sum_{j \in \interval{0}{n}} \amapbis_{(3-i)}(\asymbterm_j) \geq 1$ .
\end{itemize}
\end{enumerate}
\end{enumerate}
\cut{
\begin{enumerate}

\item[\lemmalab{5}{C-wasgarbage}.]
there are two functions $\amapbis_1,\amapbis_2: \symbinternal{\asms_1}{\asms_2}{\asms} \to \interval{1}{\bound_1+\bound_2}$ such that
\begin{enumerate}
\item for all $i \in \{1,2\}$, the sum $\sum_{\asymbterm \in \symbinternal{\asms_1}{\asms_2}{\asms}} \amapbis_i(\asymbterm)$ is at most $\symbrem_i$ whenever $\symbrem_i < \bound_i$;
\item for every $\asymbterm \in \symbinternal{\asms_1}{\asms_2}{\asms}$, $(\pi_2 \circ \amap)(\asymbterm) = \amapbis_1(\asymbterm) + \amapbis_2(\asymbterm)$;
\item for every $\asymbterm \in \symbsources{\asms_1}{\asms_2}{\asms} \cap \domain{\amap}$ there is
$\{\asymbterm_0 {=} \asymbterm, \asymbterm_1,\dots,\asymbterm_n\} \subseteq \symbinternal{\asms_1}{\asms_2}{\asms}$ such that
\begin{itemize}
\item for every
$j \in \interval{0}{n-1}$, $(\pi_1 \circ \amap)(\asymbterm_j) = \asymbterm_{j+1}$;
\item $(\pi_1 \circ \amap)(\asymbterm_n) \in \bigcup_{i \in \{1,2\}} \symbtargets{\asms_i}{\asms}$;
\item given $i \in \{1,2\}$, if $(\pi_1 \circ \amap)(\asymbterm_n) \in \symbtargets{\asms_i}{\asms}$ then
$\sum_{j \in \interval{0}{n}} \amapbis_{(3-i)}(\asymbterm_j) \geq 1$ .
\end{itemize}
\end{enumerate}
\end{enumerate}
}
}

\noindent Let us dissect this involved definition.
First of, notice that we only consider symbolic memory states that abstract concrete ones. This follows directly from the fact that we require $\charsymbform(\asms_1)$, $\charsymbform(\asms_2)$ and $\charsymbform(\asms)$ to be satisfiable (first line in the definition).
The main objects in this definition are the two injections $\ainj_1$ and $\ainj_2$.
Given $i \in \{1,2\}$, the injection $\ainj_i$ goes from elements of $\symbterms_i$ to elements of $\symbterms$. Recall that, by definition of symbolic memory states, members of these sets represent equivalence classes between terms.
Informally speaking, the role of $\ainj_1$ and $\ainj_2$ is then to explain how the terms in $\asms_1$ and $\asms_2$ are updated when the two symbolic memory states are combined into $\asms$.
For the remaining part of the section, let us fix three memory states
$\pair{\astore}{\aheap_1}$, $\pair{\astore}{\aheap_2}$ and $\pair{\astore}{\aheap}$, and let us assume $\asms_1 = \symbms{\astore}{\aheap_1}{\asetvar}{\bound_1} = {\triple{\symbterms_1}{\amap_1}{\symbrem_1}}$, $\asms_2 = \symbms{\astore}{\aheap_2}{\asetvar}{\bound_2} = {\triple{\symbterms_2}{\amap_2}{\symbrem_2}}$
and $\asms = \symbms{\astore}{\aheap}{\asetvar}{\bound_1+\bound_2} = {\triple{\symbterms}{\amap}{\symbrem}}$ to be their abstractions.
The key property of $\ainj_1$ and $\ainj_2$ is that, if $\triple{\asms_1}{\asms_2}{\asms} \in \asymbunion$ holds, then
for every equivalence class $\asymbterm \in \symbterms_i$ ($i \in \{1,2\}$),
the location $\alocation$ corresponding in $\pair{\astore}{\aheap_i}$ to the terms in $\asymbterm$ also corresponds
to the terms in $\ainj_i(\asymbterm)$ w.r.t. $\pair{\astore}{\aheap}$.
Essentially then, each pair in $\ainj_i$ (seen as a binary relation) uniquely corresponds to a location.
When this correspondence is understood, most of the conditions on these two injections are easy to follow.
When possible, we explain these conditions with easy ``small-step'' tautologies of \intervalSL (e.g.\ in \ref{C-meet-points}(a)).
\begin{enumerate}
\item Condition~\ref{C-functionality} essentially tells us that $\symbunion{\asms_1}{\asms_2}{\asms}$ cannot hold if the two symbolic memory states only correspond to concrete models that do not have a disjoint heap domain.
For example, let us consider the symbolic memory state $\asms' = (\asymbterm,\{\asymbterm \pto (\asymbterm,2) \},0)$ where $\asymbterm = \{ \avariable, \ameetvar{\avariable}{\avariable}{\avariable}\}$ and $\avariable \in \asetvar$.
Trivially, it cannot be that there is a symbolic memory state $\asms''$ such that $\symbunion{\asms'}{\asms'}{\asms''}$.
Indeed, every concrete memory state $\pair{\astore'}{\aheap'}$ having $\symbms{\astore'}{\aheap'}{\asetvar}{\bound} = \asms'$ must present a cycle involving the location corresponding to $\avariable$.
Hence, it cannot be that both $\pair{\astore}{\aheap_1}$ and $\pair{\astore}{\aheap_2}$ are abstracted by $\asms'$, as it implies that $\domain{\aheap_1} \cap \domain{\aheap_2} \neq \emptyset$ and therefore
$\aheap_1 + \aheap_2$ is not defined (which is required by \ref{compositionproperty}).
Alternatively, this condition can be easily explained with the formula $\lnot(\seesgeq{\aterm}{\aterm_1}{\asetmeetvar_1}{\inbound_1} \separate \seesgeq{\aterm}{\aterm_2}{\asetmeetvar_2}{\inbound_2})$, which is tautological in \intervalSL.
\item Let us now consider the Condition~\ref{C-meet-points}.
 \ref{C-meet-points}(a) essentially tells us that (dis)equivalences between program variables is preserved, and is better explained just by considering the tautology $\avariable = \avariablebis \iff (\avariable = \avariablebis \separate \true)$ (together with the commutativity of $\separate$).
To explain  \ref{C-meet-points}(b), consider $\pair{\astore}{\aheap_1}$ such that for a meet-point term $\ameetvar{\avariable}{\avariablebis}{\avariableter}$,
we have $\semantics{\ameetvar{\avariable}{\avariablebis}{\avariableter}}_{\astore,\aheap_1} = \alocation$.
Moreover, let $\aheap$
\begin{minipage}{0.8\linewidth}
be such that $\aheap_1 \subheap \aheap$.
The condition \ref{C-meet-points}(b) essentially
states the following:
\begin{itemize}[nosep]
  \item
  if $\ameetvar{\avariable}{\avariablebis}{\avariableter}$ is an asymmetric meet-point then the location $\alocation$ corresponds to $\ameetvar{\avariable}{\avariablebis}{\avariableter}$ also in $\pair{\astore}{\aheap}$. Formally,
  if $\pair{\astore}{\aheap_1} \models \ameetvar{\avariable}{\avariablebis}{\avariableter} \neq \ameetvar{\avariablebis}{\avariable}{\avariableter}$ then
  $\semantics{\ameetvar{\avariable}{\avariablebis}{\avariableter}}_{\astore,\aheap} = \semantics{\ameetvar{\avariable}{\avariablebis}{\avariableter}}_{\astore,\aheap_1}$.
  \item if instead $\ameetvar{\avariable}{\avariablebis}{\avariableter}$ is a symmetric meet-point, then the location $\alocation$ corresponds in $\pair{\astore}{\aheap}$ to at least one term between $\ameetvar{\avariable}{\avariablebis}{\avariableter}$ and $\ameetvar{\avariablebis}{\avariable}{\avariableter}$.
  Formally, if $\pair{\astore}{\aheap_1} \models \ameetvar{\avariable}{\avariablebis}{\avariableter} = \ameetvar{\avariablebis}{\avariable}{\avariableter}$ then
  $\semantics{\ameetvar{\avariable}{\avariablebis}{\avariableter}}_{\astore,\aheap} = \semantics{\ameetvar{\avariable}{\avariablebis}{\avariableter}}_{\astore,\aheap_1}$ or
  $\semantics{\ameetvar{\avariable}{\avariablebis}{\avariableter}}_{\astore,\aheap} = \semantics{\ameetvar{\avariablebis}{\avariable}{\avariableter}}_{\astore,\aheap_1}$.
\end{itemize}
The validity of this condition is pretty straightforward. If $\ameetvar{\avariable}{\avariablebis}{\avariableter}$ is an
asymmetric meet-point, it must be that $\alocation$ belongs to a cycle, as depicted on the first figure on the right.
Then, in every extension of this heap, $\alocation$ must still be the first location reachable from $\avariable$ that belongs to~the~cycle.%
\parfillskip=0pt
\end{minipage}%
\begin{minipage}{0.2\linewidth}
  \begin{flushright}
    \begin{minipage}{0.9\linewidth}
    \begin{tikzpicture}
      \coordinate (cir) at (0,0);
      \def\radius{0.45cm}

      \draw (cir) ++(90:\radius) node[highlightnode,label={[xshift=-15pt,yshift=-6pt]{$\scriptstyle{\ameetvar{\avariable}{\avariablebis}{\avariableter}}$}}] (k) {};

      \draw (cir) ++(270:\radius) node[dot,label={[yshift=-15pt]{$\scriptstyle{{\avariableter}}$}}] (u) {};

      \node[dot,label={[xshift=0pt]$\scriptstyle{\avariable}$}] (m) [above = 0.5cm of k] {};

      \draw (cir) ++(180:\radius) node[highlightnode,label={[xshift=17pt,yshift=-10pt]$\scriptstyle{
      \ameetvar{\avariablebis}{\avariable}{\avariableter}}$}] (ze) {};

      \node[dot,label=below:{$\scriptstyle{\avariablebis}$}] (z) [left=0.6cm of u] {};

      \draw[pto] (m.center) -- (k);
      \draw[pto] (z.center) -- (ze);

      \draw[pto] (u) .. controls +(0:1) and +(0:1) .. (k);
      \draw[pto] (k) .. controls +(180:0.3) and +(90:0.3) .. (ze);
      \draw[pto] (ze) .. controls +(-90:0.3) and +(180:0.3) .. (u);
    \end{tikzpicture}

    \begin{tikzpicture}
      \coordinate (cir) at (0,0);
      \def\radius{0.45cm}

      \draw (cir) ++(90:\radius) node[dot] (k) {};

      \draw (cir) ++(270:\radius) node[dot,label={[yshift=-15pt]{$\scriptstyle{{\avariableter}}$}}] (u) {};

      \node[dot,label={[xshift=0pt]$\scriptstyle{\avariable}$}] (m) [above = 0.5cm of k] {};

      \draw (cir) ++(180:\radius) node[highlightnode,label={[xshift=32pt,yshift=-10pt]$\scriptstyle{
      \ameetvar{\avariable}{\avariablebis}{\avariableter}} = \ameetvar{\avariablebis}{\avariable}{\avariableter}$}] (ze) {};

      \node[dot,label=below:{$\scriptstyle{\avariablebis}$}] (z) [left=0.6cm of u] {};

      \draw[pto] (m.center) -- (k);
      \draw[pto] (z.center) -- (ze);

      \draw[pto] (k) .. controls +(180:0.3) and +(90:0.3) .. (ze);
      \draw[pto] (ze) .. controls +(-90:0.3) and +(180:0.3) .. (u);
    \end{tikzpicture}
    \end{minipage}
    \end{flushright}
\end{minipage}

\noindent
Hence, by definition of meet-points, $\alocation$ still corresponds to $\ameetvar{\avariable}{\avariablebis}{\avariableter}$.
If instead $\ameetvar{\avariable}{\avariablebis}{\avariableter}$ is a symmetric meet-point (as depicted on the second figure) then there are essentially three possibilities.
First, it could be that $\ameetvar{\avariable}{\avariablebis}{\avariableter}$ is still a symmetric meet-point in $\pair{\astore}{\aheap}$. Then, as shortest paths are preserved when taking extensions of a heap, this meet-point must correspond to $\alocation$.
If instead $\ameetvar{\avariable}{\avariablebis}{\avariableter}$ is no longer a symmetric meet-point, $\alocation$ now belongs to a cycle. There are essentially two distinct ways to introduce such a cycle.
\begin{itemize}
  \item It could be that the cycle is closed on a location in the path from $\astore(\avariable)$ to $\semantics{\ameetvar{\avariable}{\avariablebis}{\avariableter}}_{\astore,\aheap_1}$ (excluded). Essentially, this means moving from the memory state depicted in the second figure to the one depicted in the first one.
  Then, $\alocation$ corresponds to $\ameetvar{\avariablebis}{\avariable}{\avariableter}$.
  \item It could be that the cycle is closed on a location in the path from $\astore(\avariablebis)$ to $\semantics{\ameetvar{\avariable}{\avariablebis}{\avariableter}}_{\astore,\aheap_1}$ (excluded).
  For instance, in the memory state depicted in the second figure, this can be done by adding a path from $\astore(\avariableter)$ to $\astore(\avariablebis)$. This case is symmetrical to the previous one, and $\alocation$ corresponds to $\ameetvar{\avariable}{\avariablebis}{\avariableter}$.
\end{itemize}

\item
For Condition~\ref{C-homomorphism},
let us start by reasoning on the two concrete memory states $\pair{\astore}{\aheap_1}$ and $\pair{\astore}{\aheap_2}$, and suppose that $\aheap_1 + \aheap_2$ is defined as $\aheap$. By looking at $\aheap$, it could be that new paths between program variables are generated. An example of this is depicted in the following figure, where a non-empty path from $\astore(\avariablebis)$ to itself is introduced.

\noindent\hfill%
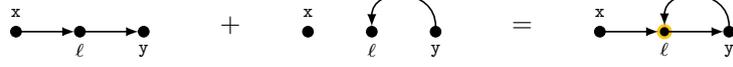
\begin{tikzpicture}[baseline]

\node[dot,label={[xshift=0pt]$\scriptstyle{\avariable}$}] (x) at (0,0) {};
\node[dot,label={[yshift=-15pt]{$\scriptstyle{{\alocation}}$}}] (l) [right = 0.7cm of x] {};
\node[dot,label={[yshift=-15pt]{$\scriptstyle{{\avariablebis}}$}}] (y) [right = 0.7cm of l] {};

\draw[pto] (x.center) -- (l);
\draw[pto] (l.center) -- (y);

\end{tikzpicture}
\qquad\heapsum\qquad
\begin{tikzpicture}[baseline]

\node[dot,label={[xshift=0pt]$\scriptstyle{\avariable}$}] (x) at (0,0) {};
\node[dot,label={[yshift=-15pt]{$\scriptstyle{{\alocation}}$}}] (l) [right = 0.7cm of x] {};
\node[dot,label={[yshift=-15pt]{$\scriptstyle{{\avariablebis}}$}}] (y) [right = 0.7cm of l] {};

\draw[pto] (y) .. controls +(90:0.6) and +(90:0.6) .. (l);
\end{tikzpicture}
\qquad=\qquad
\begin{tikzpicture}[baseline]

\node[dot,label={[xshift=0pt]$\scriptstyle{\avariable}$}] (x) at (0,0) {};
\node[highlightnode,label={[yshift=-15pt]{$\scriptstyle{{\alocation}}$}}] (l) [right = 0.7cm of x] {};
\node[dot,label={[yshift=-15pt]{$\scriptstyle{{\avariablebis}}$}}] (y) [right = 0.7cm of l] {};

\draw[pto] (x.center) -- (l);
\draw[pto] (l.center) -- (y);

\draw[pto] (y) .. controls +(90:0.6) and +(90:0.6) .. (l);
\end{tikzpicture}
\hfill\,

Because of this, the truth values of $\mathtt{sees}$ predicates change. Indeed, it is clear that the memory state on the left satisfies $\seesgeq{\avariable}{\avariablebis}{\atermset{\{\avariable,\avariablebis\}}}{2}$.
However, this core formula is  not satisfied in the memory state on the right (the disjoint union), as the location $\alocation$ corresponds to the meet-point term $\ameetvar{\avariable}{\avariablebis}{\avariablebis}$. However, this memory state still satisfies
$\seesgeq{\avariable}{\avariablebis}{\emptyset}{2}$.
Then, Condition~\ref{C-homomorphism} essentially states that paths between locations corresponding to terms are preserved when the heap is extended.
With respect to program variables, this condition can therefore be understood from the tautology
$\seesgeq{\avariable}{\avariablebis}{\asetmeetvar}{\inbound} \separate \true \implies \seesgeq{\avariable}{\avariablebis}{\emptyset}{\inbound}$.
\end{enumerate}
As we just saw, Condition~\ref{C-homomorphism} deals with existing paths between terms, and states that they still exist when a heap is extended. The last condition (Condition~\ref{C-wasgarbage}) mainly deals instead with the new paths, created by the union of two heaps.
This condition uses the sets $\symbsources{\asms_1}{\asms_2}{\asms}$, $\symbtargets{\asms_i}{\asms}$ ($i \in \{1,2\}$) and
$\symbinternal{\asms_1}{\asms_2}{\asms}$, whose role is explained below.
\begin{itemize}
\item The set $\symbsources{\asms_1}{\asms_2}{\asms}$ contains the set of equivalence classes of $\symbterms$ having variables from which it is possible to start new paths.
An example of such a variable is given by $\avariablebis$ in the three memory states introduced in order to explain Condition~\ref{C-homomorphism}.
The key property is that, when considering the two leftmost memory states, $\astore(\avariablebis)$ does not reach
(in at least one step) any location corresponding to terms (formally, for every $i \in \{1,2\}$, $\equivclass{\avariablebis}{\symbterms_i} {\not \in} \domain{\amap_i}\}$).
\item The definition of $\symbtargets{\asms_i}{\asms}$ is more involved. Essentially, it contains the set of equivalence classes in $\symbterms$ corresponding to locations where new paths end.
By considering the example in Condition~\ref{C-homomorphism}, these locations are $\astore(\avariable)$, $\astore(\avariablebis)$ or $\alocation$, where this last location is taken into account since it belongs to the path from $\astore(\avariable)$ to $\astore(\avariablebis)$ of the leftmost memory state.
\item Lastly, $\symbinternal{\asms_1}{\asms_2}{\asms}$ is the set of equivalence classes in $\symbterms$ that corresponds to locations inside new paths. In particular, a member $\asymbterm$ of $\symbinternal{\asms_1}{\asms_2}{\asms}$ satisfies two 
properties.
\begin{itemize}
  \item In the concrete memory state abstracted with $\asms$, $\asymbterm$ corresponds to a location $\alocation$ in the shortest path from a location corresponding to some element in $\symbsources{\asms_1}{\asms_2}{\asms}$ to a location corresponding to an element in $\symbtargets{\asms_i}{\asms}$ (this latter element excluded);
  \item In the concrete memory states abstracted with $\asms_1$ and $\asms_2$, $\alocation$ does not belong to any non-empty path between location corresponding to terms.
\end{itemize}
As an example, let us consider the following three memory states:

\noindent\hfill%
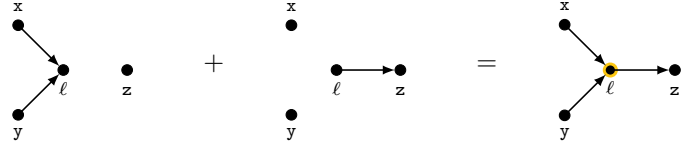
\begin{tikzpicture}[baseline]

\node[dot,label={[xshift=0pt,yshift=-15pt]$\scriptstyle{\alocation}$}] (l) at (0,0) {};
\node[dot,label={[yshift=0pt]{$\scriptstyle{{\avariable}}$}}] (x) [above left = 0.7cm of l] {};
\node[dot,label={[yshift=-15pt]{$\scriptstyle{{\avariablebis}}$}}] (y) [below left = 0.7cm of l] {};
\node[dot,label={[yshift=-15pt]{$\scriptstyle{{\avariableter}}$}}] (z) [right = 0.7cm of l] {};

\draw[pto] (x.center) -- (l);
\draw[pto] (y.center) -- (l);

\end{tikzpicture}
\qquad\heapsum\qquad
\begin{tikzpicture}[baseline]

\node[dot,label={[yshift=-15pt]{$\scriptstyle{{\alocation}}$}}] (l) at (0,0) {};
\node[dot,label={[yshift=0pt]{$\scriptstyle{{\avariable}}$}}] (x) [above left = 0.7cm of l] {};
\node[dot,label={[yshift=-15pt]{$\scriptstyle{{\avariablebis}}$}}] (y) [below left = 0.7cm of l] {};
\node[dot,label={[yshift=-15pt]{$\scriptstyle{{\avariableter}}$}}] (z) [right = 0.7cm of l] {};

\draw[pto] (l.center) -- (z);

\end{tikzpicture}
\qquad=\qquad
\begin{tikzpicture}[baseline]

\node[highlightnode,label={[xshift=0pt,yshift=-15pt]$\scriptstyle{\alocation}$}] (l) at (0,0) {};
\node[dot,label={[yshift=0pt]{$\scriptstyle{{\avariable}}$}}] (x) [above left = 0.7cm of l] {};
\node[dot,label={[yshift=-15pt]{$\scriptstyle{{\avariablebis}}$}}] (y) [below left = 0.7cm of l] {};
\node[dot,label={[yshift=-15pt]{$\scriptstyle{{\avariableter}}$}}] (z) [right = 0.7cm of l] {};

\draw[pto] (x.center) -- (l);
\draw[pto] (y.center) -- (l);
\draw[pto] (l.center) -- (z);
\end{tikzpicture}
\hfill\,

\noindent Suppose $\asms_1$, $\asms_2$ and $\asms$ to be the abstractions of these three memory states (from the left to the right).
By definition, both $\symbsources{\asms_1}{\asms_2}{\asms}$ and $\symbtargets{\asms_i}{\asms}$ (for $i \in \{1,2\}$) contain the equivalence classes
$\equivclass{\avariable}{\symbterms}$, $\equivclass{\avariablebis}{\symbterms}$ and $\equivclass{\avariableter}{\symbterms}$.
The equivalence class $\equivclass{\ameetvar{\avariable}{\avariablebis}{\avariableter}}{\symbterms}$ is instead in $\symbinternal{\asms_1}{\asms_2}{\asms}$ (together with $\equivclass{\avariable}{\symbterms}$ and $\equivclass{\avariablebis}{\symbterms}$), as the location $\alocation$ that corresponds to it does not belong to any path between terms in the two subheaps abstracted with $\asms_1$ and $\asms_2$, whereas it belongs to the path from $\equivclass{\avariable}{\symbterms}$ to $\equivclass{\avariableter}{\symbterms}$ in the memory state abstracted with $\asms$.
\end{itemize}
We are now ready to explain Condition~\ref{C-wasgarbage}, which mainly deals with quantitative aspects of the \emph{new paths} generated by the union of memory states.
More precisely, supposing that $\pair{\astore}{\aheap_1}$ and $\pair{\astore}{\aheap_2}$ are such that $\aheap_1 + \aheap_2$ is defined, these paths must be created solely from location in $\domain{\aheap_1}$ or $\domain{\aheap_2}$ that (in both memory states) are not inside paths between locations corresponding to program variables.
From the definition of the symbolic memory states $\asms_1 = {\triple{\symbterms_1}{\amap_1}{\symbrem_1}}$ and $\asms_2 = {\triple{\symbterms_2}{\amap_2}{\symbrem_2}}$, these locations are abstracted with the two quantities $\symbrem_1$ and $\symbrem_2$. These two quantities give precise bounds and constraints on the lengths and the types of paths that can occur in $\asms$. Condition~\ref{C-wasgarbage} faithfully reflects these constraints.
\begin{enumerate}[start=4]
\item  Condition~\ref{C-wasgarbage} makes use of two functions $\amapbis_1$ and $\amapbis_2$, and two quantities $\symbrem_1'$ and $\symbrem_2'$. The role of these objects is, roughly speaking, to explain how the quantities $\symbrem_1$ and $\symbrem_2$ are distributed in $\asms$. Let us be more precise.

Given $i \in \{1,2\}$ and an element $\asymbterm$ in $\symbinternal{\asms_1}{\asms_2}{\asms}$,
$\amapbis_i(\asymbterm)$ reflects the number of locations abstracted with $\symbrem_i$ that are in the path from the location corresponding to $\asymbterm$ to the one corresponding to $\pi_1 \circ \amap(\asymbterm)$ (see \ref{C-wasgarbage}(a)).
Similarly, $\symbrem_i'$ reflects the number of locations abstracted with $\symbrem_i$ that are still garbage locations in $\asms$ (i.e. they are abstracted by $\symbrem$).
Condition~\ref{C-wasgarbage} then distinguishes two cases (both treated by \ref{C-wasgarbage}(b)): $\symbrem_i < \bound_i$ and $\symbrem_i = \bound_i$.
\begin{itemize}
\item If $\symbrem_i < \bound_i$ then, by definition of symbolic memory states, it holds that $\pair{\astore}{\aheap_i}$ (which is abstracted by $\asms_i$) contains exactly $\symbrem_i$ memory cells that are not inside paths between locations corresponding to program variables. Hence, exactly this number of locations must be retrieved from $\symbrem$ and the paths generated by terms in $\symbinternal{\asms_1}{\asms_2}{\asms}$.
This is where $\amapbis_i$ and $\symbrem_i$ come into play, as we require $\symbrem_i' {+} \sum_{\asymbterm \in \symbinternal{\asms_1}{\asms_2}{\asms}} \amapbis_i(\asymbterm)$ to be equal~to~$\symbrem_i$.
\item If instead $\symbrem_i = \bound_i$, by definition of symbolic memory states it holds that $\pair{\astore}{\aheap_i}$ contains at least $\symbrem_i$ memory cells that are not inside paths between locations corresponding to program variables.
Hence, we require $\symbrem_i' + \sum_{\asymbterm \in \symbinternal{\asms_1}{\asms_2}{\asms}} \amapbis_i(\asymbterm)$ to be at least $\symbrem_i$.
\end{itemize}
Lastly, let us briefly explain \ref{C-wasgarbage}(c).
Let us consider a non-empty \emph{new} shortest path from a location corresponding to an equivalence class of $\symbsources{\asms_1}{\asms_2}{\asms}$ to a location corresponding to an equivalence class of $\symbtargets{\asms_i}{\asms}$.
Essentially, as it reaches an element in $\symbtargets{\asms_i}{\asms}$, \ref{C-wasgarbage}(c) states that there must be at least one location in this path that belongs to $\domain{\aheap_{3-i}}$ (i.e. the other heap).
Indeed, if this was not the case, the whole path would belong to the heap $\aheap_i$, and hence it is not generated by the union of the two memory states.
\end{enumerate}

\end{document}